\documentclass[11pt]{article}
\usepackage{hyperref}
\usepackage{amsmath,amsthm,amssymb,amsopn,amsfonts,pdfpages,bm}
\usepackage{algpseudocode,algorithm}
\usepackage{fullpage}
\usepackage{parskip}
\usepackage[numbers,sort]{natbib}
\usepackage{enumerate,url}
\usepackage{graphicx}
\usepackage{subfig}
\usepackage{graphics,multicol}
\usepackage{fancybox}
\usepackage{fancyhdr}
\usepackage{MnSymbol}
\makeatletter
\def\thm@space@setup{%
 \thm@preskip=\parskip \thm@postskip=0pt\label{key}
}
\makeatother
\makeatletter
\let \@fnsymbol\@arabic
\makeatother

\newtheorem{theorem}{Theorem}
\newtheorem{lemma}{Lemma}
\newtheorem{corollary}{Corollary}

\newtheorem{proposition}{Proposition}
\theoremstyle{definition}
\newtheorem{definition}{Definition}
\newtheorem{remark}{Remark}
\newtheorem{assumption}{Assumption}
\newcommand{\R}{\mathbb{R}}

\newcommand{\E}{\mathbb{E}}
\newcommand{\bA}{\mathbf{A}}
\newcommand{\bB}{\mathbf{B}}
\newcommand{\bD}{\mathbf{D}}
\newcommand{\bE}{\mathbf{E}}
\newcommand{\bJ}{\mathbf{J}}
\newcommand{\bH}{\mathbf{H}}
\newcommand{\bI}{\mathbf{I}}
\newcommand{\bL}{\mathbf{L}}
\newcommand{\bM}{\mathbf{M}}

\newcommand{\bP}{\mathbf{P}}
\newcommand{\bR}{\mathbf{R}}
\newcommand{\bS}{\mathbf{S}}
\newcommand{\bT}{\mathbf{T}}
\newcommand{\bQ}{\mathbf{Q}}
\newcommand{\bU}{\mathbf{U}}
\newcommand{\bV}{\mathbf{V}}
\newcommand{\bW}{\mathbf{W}}
\newcommand{\bX}{\mathbf{X}}
\newcommand{\bY}{\mathbf{Y}}
\newcommand{\bZ}{\mathbf{Z}}

\newcommand{\bv}{\mathbf{v}}
\newcommand{\bx}{\mathbf{x}}
\newcommand{\by}{\mathbf{y}}

\newcommand{\calL}{\mathcal{L}}
\newcommand{\calN}{\mathcal{N}}
\newcommand{\calO}{\mathcal{O}}

\newcommand{\mC}{\mathcal{C}}

\newcommand{\Xhat}{\hat{\bX}}
\newcommand{\dhat}{\hat{d}}
\newcommand{\Khat}{\hat{K}}
\newcommand{\Yhat}{\hat{\bY}}
\newcommand{\Zhat}{\hat{\bZ}}

\newcommand{\zeromx}{\mathbf{0}}

\newcommand{\supp}{\operatorname{supp}}
\newcommand{\tr}{\operatorname{tr}}

\newcommand{\Ptilde}{\tilde{\bP}}

\newcommand{\Wtilde}{\tilde{\bW}}
\newcommand{\Wn}{\bW_n}
\newcommand{\Wntilde}{\Wtilde_n}

\newcommand{\Xstar}{\bX^*}

\newcommand{\Zstar}{\bZ^*}

\newcommand{\tti}{2 \rightarrow \infty}
\newcommand{\Abar}{\bar{\bA}}
\newcommand{\Xbar}{\bar{\bX}}

\newcommand{\balpha}{\mathbf{\alpha}}
\newcommand{\bDelta}{\mathbf{\Delta}}
\newcommand{\bSigma}{\mathbf{\Sigma}}
\newcommand{\Sigmatilde}{\tilde{\bSigma}}

\newcommand{\alphavec}{\vec{\balpha}}

\newcommand{\UM}{\bU_{\bM}}
\newcommand{\UA}{\bU_{\bA}}
\newcommand{\ULA}{\bU_{\calL(\bA)}}
\newcommand{\SLA}{\bS_{\calL(\bA)}}
\newcommand{\SLP}{\bS_{\calL(\bP)}}
\newcommand{\ULP}{\bU_{\calL(\bP)}}
\newcommand{\SM}{\bS_{\bM}}
\newcommand{\SA}{\bS_{\bA}}
\newcommand{\UPt}{\bU_{\Ptilde}}
\newcommand{\SPt}{\bS_{\Ptilde}}
\newcommand{\UP}{\bU_{\bP}}
\newcommand{\SP}{\bS_{\bP}}

\newcommand{\ASE}{\operatorname{ASE}}

\newcommand{\OMNI}{\operatorname{OMNI}}

\newcommand{\RDPG}{\operatorname{RDPG}}
\newcommand{\JRDPG}{\operatorname{JRDPG}}

\newcommand{\Bern}{\operatorname{Bernoulli}}

\newcommand{\inlaw}{\xrightarrow{\calL}}
\newcommand{\inprob}{\xrightarrow{P}}
\newcommand{\iid}{\stackrel{\text{i.i.d.}}{\sim}}

\newcommand{\ErdosRenyi}{Erd\H{o}s-R\'{e}nyi }
\newcommand{\Furedi}{F\"{u}redi }
\newcommand{\Erdos}{Erd\H{o}s}
\newcommand{\Renyi}{R\'{e}nyi}
\newcommand{\Komlos}{Koml\'{o}s }
\newcommand{\mclustase}{GMM \circ ASE}
\newcommand{\smclustase}{SemiparGMM \circ ASE}
\newcommand{\thetahat}{\widehat{\theta}}

\newcommand{\rhohat}{\widehat{\rho}}
\newcommand{\Chat}{\widehat{\mathcal{C}}_{KC}}
\newcommand{\CKC}{\mathcal{C}_{KC}}

\renewcommand{\Re}{\mathbb{R}}
\begin{document}
\title{Statistical inference on random dot product graphs: a survey}
\author{Avanti Athreya\thanks{Department of Applied Mathematics and Statistics, Johns Hopkins University, Baltimore, MD.  {\em Email correspondence}: Avanti Athreya at \texttt{dathrey1@jhu.edu.}}, \, Donniell E. Fishkind\footnotemark[1], \,Keith Levin\thanks{Department of Statistics, University of Michigan, Ann Arbor, MI.}, \, Vince Lyzinski\thanks{Department of Mathematics and Statistics, University of Massachusetts, Amherst, MA.}, \, Youngser Park\thanks{Center for Imaging Science, Johns Hopkins University, Baltimore, MD.}, \\
Yichen Qin\thanks{Department of Operations, Business Analytics, and Information Systems, University of Cincinnati, Cincinnati, OH.}, \,Daniel L. Sussman\thanks{Department of Mathematics and Statistics, Boston University, Boston, MA..},\, Minh Tang\footnotemark[1], \, Joshua T. Vogelstein\thanks{Department of Biomedical Engineering, Johns Hopkins University, Baltimore, MD.},\, and Carey E. Priebe\footnotemark[1]}
	\maketitle
\begin{abstract}
	 The random dot product graph (RDPG) is an independent-edge random graph that is analytically tractable and, simultaneously, either encompasses or can successfully approximate a wide range of random graphs, from relatively simple stochastic block models to complex latent position graphs. In this survey paper, we describe a comprehensive paradigm for statistical inference on random dot product graphs, a paradigm centered on spectral embeddings of adjacency and Laplacian matrices. We examine the analogues, in graph inference, of several canonical tenets of classical Euclidean inference: in particular, we summarize a body of existing results on the consistency and asymptotic normality of the adjacency and Laplacian spectral embeddings, and the role these spectral embeddings can play in the construction of single- and multi-sample hypothesis tests for graph data. We investigate several real-world applications, including community detection and classification in large social networks and the determination of functional and biologically relevant network properties from an exploratory data analysis of the {\em Drosophila} connectome. We outline requisite background and current open problems in spectral graph inference.\\
	

{\em Keywords}: Random dot product graph, adjacency spectral embedding, Laplacian spectral embedding, multi-sample graph hypothesis testing, semiparametric modeling\\

\end{abstract}
\newpage
	\tableofcontents 
	\newpage
	\section{Introduction}\label{sec:Intro}
 Random graph inference is an active, interdisciplinary area of current research, bridging combinatorics, probability, statistical theory, and machine learning, as well as a wide spectrum of application domains from neuroscience to sociology. Statistical inference on random graphs and networks, in particular, has witnessed extraordinary growth over the last decade: for example, \cite{goldenberg2010survey,kolaczyk:_statistical} discuss the considerable applications in recent network science of several canonical random graph models. 
 
 Of course, combinatorial graph theory itself is centuries old---indeed, in his resolution to the problem of the bridges of K\"onigsberg, Leonard Euler first formalized graphs as mathematical objects consisting of vertices and edges. The notion of a random graph, however, and the modern theory of inference on such graphs, is comparatively new, and owes much to the pioneering work of \Erdos, \Renyi, and others in the late 1950s. E.N. Gilbert's short 1959 paper \cite{gilbert_1959} considered a random graph for which the existence of edges between vertices are independent Bernoulli random variables with common probability $p$; roughly concurrently, \Erdos  ~and \Renyi  ~provided the first detailed analysis of the probabilities of the emergence of certain types of subgraphs within such graphs \cite{Erdos_Renyi_1960_original}, and today, graphs in which the edges arise independently and with common probability $p$ are known as {\em \ErdosRenyi} (or ER) graphs.
 
 The \ErdosRenyi (ER) model is one of the simplest generative models for random graphs, but this simplicity belies astonishingly rich behavior (\cite{alon_spencer_prob_method}, \cite{bollobas07}). Nevertheless, in many applications, the requirement of a common connection probability is too stringent: graph vertices often represent heterogeneous entities, such as different people in a social network or cities in a transportation graph, and the connection probability $p_{ij}$ between vertex $i$ and $j$ may well change with $i$ and $j$ or depend on underlying attributes of the vertices.  Moreover, these heterogeneous vertex attributes may not be observable; for example, given the adjacency matrix of a Facebook community, the specific interests of the individuals may remain hidden. To more effectively model such real-world networks, we consider {\em latent position} random graphs \cite{hoff_raftery_handcock}. In a latent position graph, to each vertex $i$ in the graph there is associated an element $x_i$ of the so-called {\em latent space} $\mathcal{X}$, and the probability of connection $p_{ij}$ between any two edges $i$ and $j$ is given by a  {\em link} or {\em kernel} function $\kappa: \mathcal{X} \times \mathcal{X} \rightarrow [0,1]$. That is, the edges are generated independently (so the graph is an {\em independent-edge} graph) and $p_{ij}=\kappa(x_i, x_j)$.
 
 The {\em random dot product graph} (RDPG) of Young and Scheinerman \cite{young2007random} is an especially tractable latent position graph; here, the latent space is an appropriately constrained subspace of Euclidean space $\mathbb{R}^d$, and the link function is simply the dot or inner product of the pair of $d$-dimensional latent positions.  Thus, in a $d$-dimensional random dot product graph with $n$ vertices, the latent positions associated to the vertices can be represented by an $n \times d$ matrix $\bX$ whose rows are the latent positions, and the matrix of connection probabilities $\bP=(\bP_{ij})$ is given by $\bP=\bX\bX^{\top}$.  Conditional on this matrix $\bP$, the RDPG has an adjacency matrix $\bA=(\bA_{ij})$ whose entries are Bernoulli random variables with probability $\bP_{ij}$. For simplicity, we will typically consider symmetric, {\em hollow} RDPG graphs; that is, undirected, unweighted graphs in which $\bA_{ii}=0$, so there are no self-edges. In our real data analysis of a neural connectome in Section \ref{subsec:MBStructure}, however, we describe how to adapt our results to weighted and directed graphs.
 
 In any latent position graph, the latent positions associated to graph vertices can themselves be random; for instance, the latent positions may be independent, identically distributed random variables with some distribution $F$ on $\mathbb{R}^d$.  The well-known {\em stochastic blockmodel} (SBM), in which each vertex belongs to one of $K$ subsets known as {\em blocks}, with connection probabilities determined solely by block membership \cite{Holland1983}, can be represented as a random dot product graph in which all the vertices in a given block have the same latent positions (or, in the case of random latent positions, an RDPG for which the distribution $F$ is supported on a finite set). Despite their structural simplicity, stochastic block models are the building blocks for all independent-edge random graphs; \cite{wolfe13:_nonpar} demonstrates that any independent-edge random graph can be well-approximated by a stochastic block model with a sufficiently large number of blocks. Since stochastic block models can themselves be viewed as random dot product graphs, we see that suitably high-dimensional random dot product graphs can provide accurate approximations of latent position graphs \cite{tang2012universally}, and, in turn, independent-edge graphs.  Thus, the architectural simplicity of the random dot product graph makes it particularly amenable to analysis, and its near-universality in graph approximation renders it expansively applicable.  In addition, the cornerstone of our analysis of randot dot product graphs is a set of classical probabilistic and linear algebraic techniques that are useful in much broader settings, such as random matrix theory.  As such, the random dot product graph is both a rich and interesting object of study in its own right and a natural point of departure for wider graph inference. 

A classical inference task for Euclidean data is to estimate, from sample data, certain underlying distributional parameters. Similarly, for a latent position graph, a classical graph inference task is to infer the graph parameters from an observation of the adjacency matrix $\bA$. Indeed, our overall paradigm for random graph inference is inspired by the fundamental tenets of classical statistical inference for Euclidean data. Namely, our goal is to construct methods and estimators of graph parameters or graph distributions; and, for these estimators, to analyze their (1) consistency; (2) asymptotic distributions; (3) asymptotic relative efficiency; (4) robustness to model misspecification; and (5) implications for subsequent inference including one- and multi-sample hypothesis testing.  In this paper, we summarize and synthesize a considerable body of work on spectral methods for inference in random dot product graphs, all of which not only advance fundamental tenets of this paradigm, but do so within a unified and parsimonious framework.
The random graph estimators and test statistics we discuss all exploit the {\em adjacency spectral embedding} (ASE) or the {\em Laplacian spectral embedding} (LSE), which are eigendecompositions of the adjacency matrix $\bA$ and {\em normalized} Laplacian matrix $\bL=\bD^{-1/2} \bA \bD^{-1/2}$, where $\bD$ is the diagonal degree matrix $\bD_{ii}=\sum_{j\neq i} \bA_{ij}$.


The ambition and scope of our approach to graph inference means that mere upper bounds on discrepancies between parameters and their estimates will not suffice. Such bounds are legion. In our proofs of consistency, we improve several bounds of this type, and in some cases improve them so drastically that concentration inequalities and asymptotic limit distributions emerge in their wake. We stress that aside from specific cases (see \cite{furedi1981eigenvalues}, \cite{tao2012random}, \cite{lei2014}), limiting distributions for eigenvalues and eigenvectors of random graphs are notably elusive. For the adjacency and Laplacian spectral embedding, we discuss not only consistency, but also asymptotic normality, robustness, and the use of the adjacency spectral embedding in the nascent field of multi-graph hypothesis testing. We illustrate how our techniques can be meaningfully applied to thorny and very sizable real data, improving on previously state-of-the-art methods for inference tasks such as community detection and classification in networks. What is more, as we now show, spectral graph embeddings are relevant to many complex and seemingly disparate aspects of graph inference.

A bird's-eye view of our methodology might well start with the stochastic blockmodel, where, for an SBM with a finite number of blocks of stochastically equivalent vertices,
\cite{STFP-2011} and \cite{fishkind2013consistent} show that $k$-means clustering of the rows of the adjacency spectral embedding accurately partitions the vertices into the correct blocks, even when the embedding dimension is misspecified or the number of blocks is unknown. Furthermore, \cite{lyzinski13:_perfec} and \cite{lyzinski15_HSBM}
give a significant improvement in the misclassification rate, by exhibiting an almost-surely perfect clustering in which, in the limit, no vertices whatsoever are misclassified. For random
dot product graphs more generally, \cite{sussman12:_univer} shows that the latent positions are consistently
estimated by the embedding, which then allows for accurate learning in a supervised vertex classification framework. In \cite{tang2012universally} these results are extended to more general latent position models, establishing a powerful universal consistency result for vertex classification in general latent position graphs, and also exhibiting an efficient embedding of vertices which were not observed in the original graph. In \cite{athreya2013limit} and \cite{tang_lse}, the authors supply distributional results, akin to a central limit theorem, for both the adjacency and Laplacian spectral embedding, respectively; the former leads to a nontrivially superior algorithm for the estimation of block memberships in a stochastic block model (\cite{suwan14:_empbayes}), and the latter resolves, through an elegant comparison of Chernoff information, a long-standing open question of the relative merits of the adjacency and Laplacian graph representations.

Morever, graph embedding plays a central role in foundational work
of Tang et al. \cite{tang14:_semipar} and \cite{tang14:_nonpar} on two-sample graph comparison: these papers provide theoretically justified, valid and consistent hypothesis tests for the semiparamatric problem of determining whether two random dot product graphs have the same latent positions and the nonparametric problem of determining whether two random dot product graphs have the
same underlying distributions. This, then, yields a systematic framework for determining statistical similarity across graphs, which in turn underpins yet another provably consistent algorithm for the decomposition of random graphs with a hierarchical structure \cite{lyzinski15_HSBM}.  In \cite{levin_omni_2017}, distributional results are given for an omnibus embedding of multiple random dot product graphs on the same vertex set, and this embedding performs well both for latent position estimation and for multi-sample graph testing. For the critical inference task of vertex nomination, in which the inference goal is to produce an ordering of vertices of interest (see, for instance, \cite{Coppersmith2014}) Fishkind and coauthors introduce in \cite{FisLyzPaoChePri2015} an array of principled vertex nomination algorithms –--the canonical, maximum likelihood and spectral vertex
nomination schemes---and demonstrate the algorithms' effectiveness on both synthetic and real data. In \cite{LyzLevFisPri2016} the consistency of the maximum likelihood vertex nomination scheme is established, a scalable restricted version of the algorithm is introduced, and the algorithms are adapted to incorporate general vertex features. 

Overall, we stress that these principled techniques for random dot product graphs exploit the Euclidean nature of graph embeddings but are general enough to yield meaningful results for a wide variety of random graphs.  Because our focus is, in part, on spectral methods, and because the adjacency matrix $\bA$ of an independent-edge graph can be regarded as a noisy version of the matrix of probabilities $\bP$ \cite{oliveira2009concentration}, we rely on several classical results on matrix perturbations, most prominently the Davis-Kahan Theorem (see \cite{Bhatia1997} for the theorem itself, \cite{rohe2011spectral} for an illustration of its role in graph inference, and \cite{DK_usefulvariant} for a very useful variant). We also depend on the aforementioned spectral bounds of Oliveira in \cite{oliveira2009concentration} and a more recent sharpening due to Lu and Peng \cite{lu13:_spect}. We leverage probabilistic concentration inequalities, such as those of Hoeffding and Bernstein \cite{Tropp2015}. Finally, several of our results do require suitable eigengaps for $\bP$ and lower bounds on graph density, as measured by the maximum degree and the size of the smallest eigenvalue of $\bP$.  It is important to point out that in our analysis, we assume that the embedding dimension $d$ of our graphs is known and fixed.  In real data applications, such an embedding dimension is not known, and in Section \ref{subsec:MBStructure}, we discuss approaches (see \cite{chatterjee2015} and \cite{zhu06:_autom}) to estimating the embedding dimension.  Robustness of our procedures to errors in embedding dimension is a problem of current investigation.

In the study of stochastic blockmodels, there has been a recent push to understand the fundamental information-theoretic limits for community detection and graph partitioning \cite{Abbe2015,Mossel2014,Abbe2016,Mossel2013}. 
These bounds are typically algorithm-free and focus on stochastic blockmodels with constant or logarithmic average degree, in which differences between vertices in different blocks are assumed to be at the boundary of detectability.
Our efforts have a somewhat different flavor, in that we seek to understand the precise behavior of a widely applicable procedure in a more general model.
Additionally, we treat sparsity as a secondary concern, and typically do not broach the question of the exact limits of our procedures.
Our spectral methods may not be optimal for stochastic models \cite{Krzakala2013,Kawamoto2015} but they are very useful, in that they rely on well-optimized computational methods, can be implemented quickly in many standard languages, extend readily to other models, and serve as a foundation for more complex analyses.

Finally, we would be remiss not to point out that while spectral decompositions and clusterings of the adjacency matrix are appropriate for graph inference, they are also of considerable import in combinatorial graph theory: readers may recall, for instance, the combinatorial {\em ratio-cut} problem, whose objective is to partition the vertex set of a graph into two disjoint sets in a way that minimizes the number of edges between vertices in the two sets. The minimizer of a relaxation to the ratio-cut problem \cite{fiedler1973} is the eigenvector associated to the second smallest eigenvalue of the graph Laplacian $\bL$.
While we do not pursue more specific combinatorial applications of spectral methods here, we note that \cite{chung1997spectral} provides a comprehensive overview, and \cite{von2007tutorial} gives an accessible tutorial on spectral methods.

We organize the paper as follows. In Section \ref{sec:Def_Note_Background}, we define random dot product graphs and the adjacency spectral embedding, and we recall important linear algebraic background. In Section \ref{sec:ASE_Inference_RDPG}, we discuss consistency, asymptotic normality, and hypothesis testing, as well as inference for hierarchical models. In Section \ref{sec:Applications}, we discuss applications of these results to real data. Finally, in Section \ref{sec:Complexities} we discuss current theoretical and computational difficulties and open questions, including issues of optimal embedding dimension, model limitations, robustness to errorful observations, and joint graph inference.

\section{Definitions, notation, and background}\label{sec:Def_Note_Background}
\subsection{Preliminaries and notation}
We begin by establishing notation.
For a positive integer $n$, we let $[n]=\{1, 2, \cdots, n\}$. For a vector $\bv \in \mathbb{R}^n$, we let $\| \bv \|$ denote the Euclidean norm of $\bv$. We denote the identity matrix, zero matrix, and the square matrix of all ones by, 
$\bI$, $\zeromx$, and $\bJ$, respectively. 
We use $\otimes$ to denote the Kronecker product.
For an $n_1 \times n_2$ matrix $\bH$, we let $\bH_{ij}$ denote its $i,j$th entry;
we denote by $\bH_{\cdot j}$
the column vector formed by the $j$-th column of $\bH$;
and we denote by $\bH_{i \cdot}$ the row vector
formed by the $i$-th row of $\bH$.
For a slight abuse of notation, we also let $\bH_i \in \R^{n_2}$
denote the \emph{column} vector formed by transposing the $i$-th row
of $\bH$. That is, $\bH_i = (\bH_{i \cdot})^{\top}$.
Given any suitably specified ordering on eigenvalues of a square matrix $\bH$, we let $\lambda_i(\bH)$ denote
the $i$-th eigenvalue (under such an ordering) of $\bH$ and $\sigma_i(\bH) = \sqrt{\lambda_i(\bH^{\top}\bH)}$ the $i$-th singular value of $\bH$. 
We let $\| \bH \|$ denote the spectral norm of $\bH$
and $\| \bH \|_F$ denote the Frobenius norm of $\bH$.
We let $\|\bH\|_{\tti}$ denote the maximum of the Euclidean norms of the rows of $\bH$, i.e.
$\|\bH\|_{\tti}=\max_{i} \| \bH_i \|$.
We denote the trace of a matrix $\bH$ by $\tr(\bH)$.
For an $n \times n$ symmetric matrix $\bH$ whose entries are all non-negative, we will frequently have to account for terms related to matrix sparsity, and we define $\delta(\bH)$ and $\gamma(\bH)$ as follows:
\begin{equation}\label{eq:max_degree}
\delta(\bH) = \max\limits_{1 \leq i \leq n} \sum_{j=1}^n
\bH_{ij} \qquad
\gamma(\bH) = 
\frac{\sigma_{d}(\mathbf{\bH}) - \sigma_{d+1}(\bH)}{\delta(\bH)} \leq 1
\end{equation}
In a number of cases, we need to consider a sequence of matrices. We will denote such a sequence by $\bH_n$, where $n$ is typically used to denote the index of the sequence.  The distinction between a particular element $\bH_n$ in a sequence of matrices and a particular row $\bH_i$ of a matrix will be clear from context, and our convention is typically to use $n$ to denote the index of a sequence and $i$ or $h$ to denote a particular row of a matrix. In the case where we need to consider the $i$th row of a matrix that is itself the $n$th element of a sequence, we will use the notation
$(\bH_n)_i$.

We define a {\em graph} $G$ to be an ordered pair of $(V,E)$ where $V$ is the so-called {\em vertex} or {\em node} set, and $E$, the set of {\em edges}, is a subset of the Cartesian product of $V \times V$. In a graph whose vertex set has cardinality $n$, we will usually represent $V$ as $V=\{1, 2, \cdots, n\}$, and we say there {\em is an edge between} $i$ and $j$ if $(i,j)\in E$.  The {\em adjacency} matrix $\bA$ provides a compact representation of such a graph:
$$\bA_{ij}=1 \textrm{   if   }(i,j) \in E, \textrm{  and  }\bA_{ij}=0 \textrm{ otherwise. }$$ 
Where there is no danger of confusion, we will often refer to a graph $G$ and its adjacency matrix $\bA$
interchangeably.

Our focus is random graphs, and thus we will let $\Omega$ denote our sample space, $\mathcal{F}$ the  $\sigma$-algebra of subsets of $\Omega$  and $\mathbb{P}$ our probability measure $\mathbb{P}: \mathcal{F} \rightarrow [0,1]$. We will denote the expectation of a (potentially multi-dimensional) random variable $X$ with respect to this measure by $\mathbb{E}$. Given an event $F \in \mathcal{F}$, we denote its complement by $F^c$, and we let $\Pr(F)$ denote the probability of $F$.  As we will see, in many cases we can choose $\Omega$ to be subset of Euclidean space. Because we are interested in large-graph inference, we will frequently need to demonstrate that probabilities of certain events decay at specified rates. This motivates the following definition.

\begin{definition} [Convergence asymptotically almost surely and convergence with high probability]
	\label{def:whp}
	Given a sequence of events $\{ F_n \} \in \mathcal{F}$, where $n=1, 2, \cdots$, we say that $F_n$ occurs {\em asymptotically almost surely} if $\Pr(F_n) \rightarrow 1$ as $n \rightarrow \infty$. 
	We say that $F_n$ {\em occurs with high probability},
	and write $F_n \text{ w.h.p. }$,
	if for any $c_0 > 1$, there exists finite positive constant $C_0$ depending on $c_0$ such that $\Pr[ F_n^c ] \le C_0n^{-c_0}$ for all $n$.
	We note that $F_n$ occurring w.h.p. is stronger than $F_n$ occurring asymptotically almost surely. Morever, $F_n$ occurring with high probability implies, by the Borel-Cantelli Lemma \cite{chung1974course},
	that with probability $1$ there exists an $n_0$ such that
	$F_n$ holds for all $n \ge n_0$.
\end{definition}

Moreover, since our goal is often to understand large-graph inference, we need to consider asymptotics as a function of graph size $n$. As such, we recall familiar asymptotic notation:
\begin{definition} [Asymptotic notation] If $w(n)$ is a quantity depending on $n$, we will say that {\em $w$ is of order $\alpha(n)$} and use the notation $w(n) \sim \Theta(\alpha(n))$ to denote that there exist positive constants $c, C$ such that for $n$ sufficiently large,
	$$c\alpha(n) \leq w(n) \leq C \alpha(n).$$
	When the quantity $w(n)$ is clear and $w(n)\sim \Theta(\alpha(n))$, we sometimes simply write ``$w$ is of order $\alpha(n)$".
	We write $w(n) \sim O(n)$ if there exists a constant $C$ such that for $n$ sufficiently large, $w(n) \leq Cn$.  We write $w(n) \sim o(n)$ if $w(n)/n \rightarrow 0$ as $n \rightarrow \infty$, and $w(n)\sim o(1)$ if $w(n) \rightarrow 0$ as $n \rightarrow \infty$. We write $w(n) \sim \Omega(n)$ if there exists a constant $C$ such that for all $n$ sufficiently large, $w(n) \geq Cn$.
\end{definition}

Throughout, we will use $C > 0$ to denote a constant, not depending on $n$,
which may vary from one line to another.

\subsection{Models}
Since our focus is on $d$-dimensional random dot product graphs, we first define an {\em an inner product distribution} as a probability distribution over a suitable subset of $\R^d$, as follows:
\begin{definition}
	[ $d$-dimensional inner product distribution]\label{def:innerprod}
	Let $F$ be a probability distribution whose support is given by $\supp F={\bf \mathcal{X}}_d \subset \R^d$.
	We say that $F$ is a
	\emph{$d$-dimensional inner product distribution}
	on $\R^d$ if for all $\bx,\by \in \mathcal{X}_d=\supp F$, we have $\bx^{\top} \by \in [0,1]$.
\end{definition}

Next, we define a random dot product graph as an independent-edge random graph
for which the edge probabilities are given by the dot products of the latent
positions associated to the vertices.
We restrict our attention here to graphs that are undirected and
in which no vertex has an edge to itself.
\begin{definition} [Random dot product graph with distribution $F$] \label{def:RDPG}
	Let $F$ be a $d$-dimensional inner product distribution
	with $\bX_1,\bX_2,\dots,\bX_n \iid F$, collected in the rows of the matrix
	$\bX=[\bX_1, \bX_2, \dots, \bX_n]^{\top} \in \R^{n \times d}$.
	Suppose $\bA$ is a random adjacency matrix given by
	\begin{equation} \label{eq:rdpg}
	\Pr[\bA|\bX]=
	\prod_{i<j}(\bX_i^{\top}\bX_j)^{\bA_{ij}}(1-\bX_i^{\top}\bX_j)^{1-\bA_{ij}}
	\end{equation}
	We then write $(\bA,\bX) \sim \RDPG(F,n)$ and say that $\bA$ is the adjacency
	matrix of a {\em random dot product graph of dimension or rank at most} $d$ and with {\em latent positions} given by the rows of $\bX$. If $\bX \bX^{\top}$ is, in fact, a rank $d$ matrix, we say $\bA$ is the adjacency matrix of a rank $d$ random dot product graph.
\end{definition}
While our notation for a random dot product graph with distribution $F$ is $(\mathbf{A}, \mathbf{X}) \sim \mathrm{RDPG}(F)$, we emphasize that in this paper the latent positions $\mathbf{X}$ are always assumed to be unobserved. An almost identical definition holds for random dot product graphs with fixed but unobserved latent positions:
\begin{definition}[RDPG with fixed latent positions]\label{rem:RDPG_fixed_latentpos} In the definition ~\ref{def:RDPG} given 
above, the latent positions are themselves random.  If, instead, the latent positions are given by a fixed matrix $\bX$ and, given this matrix, the graph is generated according to Eq.\eqref{eq:rdpg}, we say that $\bA$ is a realization of a random dot product graph with latent positions $\bX$, and we write $\bA \sim \mathrm{RDPG}(\bX)$.
\end{definition}

\begin{remark} [Nonidentifiability]\label{rem:nonid}
	Given a graph distributed as an RDPG,
	the natural task is to recover the latent positions $\bX$ that gave
	rise to the observed graph.
	However, the RDPG model has an inherent nonidentifiability:
	let $\bX \in \R^{n \times d}$ be a matrix of latent positions
	and let $\bW \in \R^{d \times d}$ be a unitary matrix.
	Since $\bX \bX^{\top} = (\bX \bW) (\bX \bW)^{\top}$, it is clear that the latent positions
	$\bX$ and $\bX\bW$ give rise to the same distribution over graphs in
	Equation~\eqref{eq:rdpg}.
	Note that most latent position models, as defined below, also suffer from similar types of non-identifiability as edge-probabilities may be invariant to various transformations.
\end{remark}

As we mentioned, the random dot product graph is a specific instance of the more general {\em latent position random graph} with {\em link} or {\em kernel} function $\kappa$.
Indeed, the latent positions themselves need not belong to Euclidean space per se, and the link function need not be an inner product.

\begin{definition}[Latent position random graph with kernel $\kappa$]\label{def:latentpos_graph} 
	Let $\mathcal{X}$ be a set and $\kappa: \mathcal{X} \times \mathcal{X} \rightarrow [0,1]$ a symmetric function.  Suppose to each $i \in [n]$ there is associated a point $\bX_i \in \mathcal{X}$.  Given $\bX=\{\bX_1, \cdots, \bX_n\}$ consider the graph with adjacency matrix $\bA$ defined by
	\begin{equation} \label{eq:lpg}
	\Pr[\bA|\bX]=
	\prod_{i<j}\kappa(\bX_i,\bX_j)^{\bA_{ij}}(1-\kappa(\bX_i,\bX_j))^{1-\bA_{ij}}
	\end{equation}
	Then $\bA$ is the adjacency matrix of a latent position random graph with latent position $\bX$ and link function $\kappa$.
\end{definition}

Similarly, we can define independent edge graphs for which latent positions need not play a role.

\begin{definition}[Independent-edge graphs]
For a matrix symmetric matrix $\bP$ of probabilities, we say that $\bA$ is distributed as an independent edge graph with probabilities $\bP$ if 
\begin{equation} \label{eq:lpg}
	\Pr[\bA|\bX]=
	\prod_{i<j}\bP_{ij}^{\bA_{ij}}(1-\bP_{ij})^{1-\bA_{ij}}
\end{equation} 
\end{definition}

By their very structure, latent position random graphs, for fixed latent positions, are independent-edge random graphs.
In general, for any latent position graph the matrix of edge probabilities $\bP$ is given by $\bP_{ij}=\kappa(\bX_i,\bX_j)$
Of course, in the case of an random dot product graph with latent position matrix $\bX$, the probability $\bP_{ij}$ of observing an edge between vertex $i$ and vertex $j$ is simply $\bX_i^{\top}\bX_j$.
Thus, for an RDPG with latent positions $\bX$, the matrix $\bP=[p_{ij}]$ is given by $\bP=\bX\bX^{\top}$.

In order to more carefully relate latent position models and RDPGs, we can consider the set of positive semidefinite latent position graphs. 
Namely, we will say that a latent position random graph is positive semidefinite if the matrix $\bP$ is positive semidefinite.
In this case, we note that an RDPG can be used to approximate the latent position random graph distribution.
The best rank-$d$ approximation of $\bP$, in terms of the Frobenius norm \citep{Eckart1936}, will correspond to a RDPG with $d$-dimensional latent positions.
In this sense, by allowing $d$ to be as large as necessary, any positive semi-definite latent position random graph distribution can be approximated by a RDPG distribution to arbitrary precision \citep{tangs.:_univer}.

While latent position models generalize the random dot product graph, RDPGs  can be easily related to the more limited {\em stochastic blockmodel} graph \cite{Holland1983}. The stochastic block model is also an independent-edge random graph whose vertex set is partitioned into $K$ groups, called {\em blocks}, and the stochastic blockmodel is typically parameterized by (1) a
$K\times K$ matrix of probabilities $\bB$ of adjacencies between vertices in
each of the blocks, and (2) a {\em block-assignment vector} $\tau:[n] \rightarrow [K]$ which assigns each vertex to its block. That is, for any two vertices $i,j$, the probability of their connection is 
$$\bP_{ij}=\bB_{\tau(i), \tau(j)},$$
and we typically write $\bA \sim \mathrm{SBM}(\bB, \tau)$.
Here we present an alternative definition in terms of the RDPG model.

\begin{definition}[Positive semidefinite $k$-block stochastic block model]\label{def:PS_SBM} We say an RDPG with latent positions $\bX$ is an SBM with $K$ blocks if the
	number of distinct rows in $\bX$ is $K$, denoted $\bX_{(1)}, \cdots, \bX_{(K)}$  In this case, we define the
	block membership function $\tau:[n]\mapsto [K]$ to be a function
	such that $\tau(i)=\tau(j)$ if and only if $\bX_i=\bX_j$.  
	We then write $$\bA \sim \mathrm{SBM}(\tau, \{\bX_{(i)}\}_{i=1}^{K})$$
	In addition, we also consider the case of a stochastic block model in which the block memberships of each vertex is randomly assigned. More precisely, let $\pi \in (0,1)^{K}$ with $\sum_{k=1}^{n} \pi_k=1$ and suppose that $\tau(1), \tau(2), \dots, \tau(n)$ are now i.i.d. random variables with distribution $\mathrm{Categorical}(\pi)$, i.e., $\mathrm{Pr}(\tau(i) = k) = \pi_k$ for all $k$. Then we say $\bA$ is an {\em SBM with i.i.d block memberships}, and we write $$\bA \sim \mathrm{SBM}(\pi, \{X_{(i)}\}).$$
\end{definition}
We also consider the {\em degree-corrected} stochastic block model:
\begin{definition}[Degree Corrected Stochastic Blockmodel (DCSBM) \citep{karrer2011stochastic}]\label{def:DCSBM} We
	say an RDPG is a DCSBM with $K$ blocks if there exist $K$ unit vectors
	$y_1,\dotsc,y_K\in\Re^{d}$ such that for each $i\in[n]$, there exists
	$k\in[K]$ and $c_i\in(0,1)$ such that $X_i=c_i y_k$.
\end{definition}

\begin{remark} The degree-corrected stochastic blockmodel model is inherently more flexible than the standard SBM because it allows for vertices within each block/community to have different expected degrees.  This flexibility has made it a popular choice for modeling network data \cite{karrer2011stochastic}.
\end{remark}

\begin{definition}[Mixed Membership Stochastic Blockmodel (MMSBM) \citep{Airoldi2008}]
We say an RDPG is a MMSBM with K blocks if there exists $K$ unit vectors $y_1,\dotsc, y_K\in \Re^d$ such that for each $i\in [n]$, there exists $\alpha_1,\dotsc,\alpha_K>0$ such that $\sum_{k=1}^K \alpha_k=1$ and $X_i=\sum_{k=1}^{K} \alpha_k y_k$.
\end{definition}

\begin{remark}
The mixed membership SBM is again more general than the SBM by allowing for each vertex to in a mixture of different blocks. Additionally, note that every RDPG is a MMSBM for some choice of $K$.
\end{remark}

Our next theorem summarizes the relationship between these models.

\begin{theorem}
Considered as statistical models for graphs, i.e. sets of probability distributions on graphs, the positive-semidefinite $K$-block SBM is a subset of the $K$-block DCSBM and the $K$-block MMSBM.
Both the positive semidefinite $K$-block DCSBM and $K$-block MMSBM are subsets of the RDPG model with $K$-dimensional latent positions. 
Finally, the union of all possible RDPG models, without restriction of latent position dimension, is dense in the set of positive semidefinite latent position models. 
\end{theorem}

\subsection{Embeddings}

Since we rely on spectral decompositions, we begin with describing the notations for the
spectral decomposition of the rank $d$ positive semidefinite matrix $\bP=\bX\bX^{\top}$.
\begin{definition} [Spectral Decomposition of $\bP$]\label{def:spec_decomp_P}
	Since $\bP$ is symmetric and positive semidefinite, let
	$\bP= \UP \SP \UP^{\top}$ denote its spectral decomposition,
	with $\UP \in \R^{n \times d}$ having orthonormal columns
	and $\SP \in \R^{d \times d}$ a diagonal matrix
	with nonincreasing entries
	$(\SP)_{1,1}\ge (\SP)_{2,2} \ge \cdots \ge (\SP)_{d,d} > 0$.
\end{definition}
As with the spectral decomposition of the matrix $\bP$, given an adjancency matrix $\bA$, we define its adjacency spectral embedding as follows; 
\begin{definition} [Adjacency spectral embedding (ASE)]\label{def:ASE}
	Given a positive integer $d \geq 1$, the {\em adjacency spectral embedding} (ASE) of $\mathbf{A}$ into
	$\mathbb{R}^{d}$ is given by $\hat{{\bf X}}={\bf U}_{\mathbf{A}}
	{\bf S}_{\mathbf{A}}^{1/2}$ where
	$$|{\bf A}|=[{\bf U}_{\mathbf{A}}|{\bf U}^{\perp}_{\mathbf{A}}][{\bf
		S}_{\mathbf{A}} \bigoplus {\bf S}^{\perp}_{\mathbf{A}}][{\bf
		U}_{\mathbf{A}}|{\bf U}^{\perp}_{\mathbf{A}}]$$ is the spectral
	decomposition of $|\bf{A}| = (\bf{A}^{T} \bf{A})^{1/2}$ and
	$\mathbf{S}_{\mathbf{A}}$ is the diagonal matrix of the $d$ largest eigenvalues
	of $|\mathbf{A}|$ and $\mathbf{U}_{\mathbf{A}}$ is the $n \times d$ matrix whose
	columns are the corresponding eigenvectors.
\end{definition}
\begin{remark}
  The intuition behind the notion of adjacency spectral embedding is
  as follows.
  Given the goal of estimating $\bX$, had we observed $\bP$ then the spectral embedding of $\bP$, given by $\UP \SP^{1/2}$, will be a orthogonal transformation of $\bX$.
  Of course, $\bP$ is not observed but instead we observe $\bA$, a noisy version of $\bP$. 
  The ASE will be a good estimate of $\bX$ provided that the noise does not greatly impact the embedding.
  As we will see shortly,
  one can show that
  $\|\mathbf{A} - \mathbf{X} \mathbf{X}^{\top} \| = O(\|\mathbf{X}\|)
  = o(\|\mathbf{X} \mathbf{X}^{\top}\|)$
  with high probability \cite{oliveira2009concentration,lu13:_spect,Tropp2015,rinaldo_2013}. 
  That is to say, $\mathbf{A}$ can be viewed as a ``small''
  perturbation of $\mathbf{X} \mathbf{X}^{\top}$. 
  Weyl's inequality or the Kato-Temple inequality \cite{cape_16_conc,kato-temple} then yield that the eigenvalues of $\mathbf{A}$ are ``close'' to the eigenvalues of $\mathbf{X} \mathbf{X}^{\top}$.
  In addition, by the
  Davis-Kahan theorem \cite{davis70}, the subspace
  spanned by the top $d$ eigenvectors of $\mathbf{X}
  \mathbf{X}^{\top}$ is well-approximated by the subspace spanned by
  the top $d$ eigenvectors of $\mathbf{A}$.
\end{remark}

We also define the analogous Laplacian spectral embedding which uses the spectral decomposition of the normalized Laplacian matrix.

\begin{definition}
\label{def:LSE}
Let $\mathcal{L}(\mathbf{A}) = \mathbf{D}^{-1/2} \mathbf{A} \mathbf{D}^{-1/2}$ denote the normalized Laplacian of $\mathbf{A}$ where $\mathbf{D}$ is the diagonal matrix whose diagonal entries $\mathbf{D}_{ii} = \sum_{j \not = i} \mathbf{A}_{ij}$. Given a positive integer $d \geq 1$, the {\em Laplacian spectral embedding} (LSE) of $\mathbf{A}$ into
	$\mathbb{R}^{d}$ is given by $\breve{{\bf X}}={\bf U}_{\mathcal{L}(\mathbf{A})}
	\tilde{{\bf S}}_{\mathbf{A}}^{1/2}$
	where $$|\mathcal{L}({\bf A})|=\Bigl[{\bf U}_{\mathcal{L}(\mathbf{A})}|{\bf U}^{\perp}_{\mathcal{L}(\mathbf{A})}\Bigr]\Bigl[{\bf
		S}_{\mathcal{L}(\mathbf{A})} \bigoplus {\bf S}^{\perp}_{\mathcal{L}(\mathbf{A})}\Bigr]\Bigl[{\bf
		U}_{\mathcal{L}(\mathbf{A})}|{\bf U}^{\perp}_{\mathcal{L}(\mathbf{A})}\Bigr]$$ is the spectral
	decomposition of $|\mathcal{L}(\mathbf{A})| = (\mathcal{L}(\mathbf{A})^{\top} \mathcal{L}(\mathbf{A}))^{1/2}$ and
	$\mathbf{S}_{\mathcal{L}(\mathbf{A})}$ is the diagonal matrix containg the $d$ largest eigenvalues
	of $|\mathcal{L}(\mathbf{A})|$ on the diagonal and $\mathbf{U}_{\mathcal{L}(\mathbf{A})}$ is the $n \times d$ matrix whose
	columns are the corresponding eigenvectors.
\end{definition}

Finally, there are a variety other matrices for which spectral decompositions may be applied to yield an embedding of the graph \cite{Le2017}.
These are often dubbed as regularized embeddings and seek to improve the stability of these methods in order to accommodate sparser graphs.
While we do not analyze these embeddings directly, many of our approaches can be adapted to these other embeddings.


\section{Core proof techniques: probabilistic and linear algebraic bounds}\label{sec:core_techniques}

In this section, we give a overview of the core background results used in our proofs.
The key tools to several of our results on consistency and normality of the adjacency spectral embedding depend on a triumvirate of matrix concentration inequalities, the Davis-Kahan Theorem, and detailed bounds via the power method. 

\subsection{Concentration inequalities}

Concentration inequalities for real- and matrix-valued data are a critical component to our proofs of consistency for spectral estimates.  We make use of classical inequalities, such as Hoeffding's inequality, for real-valued random variables, and we also exploit more recent work on the concentration of sums of random matrices and matrix martingales around their expectation. For a careful study of several important matrix concentration inequalities, see \cite{Tropp2015}. 

We begin by recalling Hoeffding's inequality, which bounds the deviations between a sample mean of independent random variables and the expected value of that sample mean.
\begin{theorem}
	\label{thm:Hoeffding}
	Let $X_i$, $1 \leq i \leq n$, be independent, bounded random variables defined on some probability space $(\Omega, \mathcal{F}, \mathbb{P})$.  Suppose $a_i, b_i$ are real numbers such that $a_i \leq X_i \leq b_i$.  Let $\bar{X}$ be their sample  mean:
	$$\bar{X}=\frac{1}{n} \sum_{i=1}^n X_i$$
	Then
	\begin{equation}\label{eq:Hoeffding_bound1}
	\Pr\left(\bar{X}-\mathbb{E}(\bar{X}) \geq t\right) \leq \exp\left(\frac{[-2n^2t^2]}{\sum_{i=1}^n (b_i -a_i)^2}\right)
	\end{equation}
	and
	\begin{equation}\label{eq:Hoeffding_bound2}
	\Pr\left(|\bar{X}-\mathbb{E}(\bar{X})| \geq t\right) \leq 2\exp\left(\frac{[-2n^2t^2]}{\sum_{i=1}^n (b_i -a_i)^2}\right)
	\end{equation}
\end{theorem}

For an undirected, hollow RDPG with probability matrix $\bP$, $\E(\bA_{ij})=\bP_{ij}$ for all $i \neq j$. As such, one can regard $\bA$ as a ``noisy'' version of $\bP$. 
It is tempting to believe that $\bA$ and $\bP$ are close in terms of the Frobenius norm, but this is sadly not true; indeed, it is easy to see that $$\|\bA-\bP\|_F^2=\Theta(\|\bP\|_F^2)$$
To overcome this using only Hoeffding's inequality, we can instead consider the difference $(\bA^2-\bP^2)_{ij}$, which is a sum of independent random variables. 
Hence, Hoeffding's inequality implies that
$$|(\bA^2-\bP^2)_{ij}|^2=o(|{\bP^2}_{ij}|^2)$$
Since the eigenvectors of $\bA$ and $\bA^2$ coincide, this is itself sufficient to show concentration of the adjacency spectral embedding \citep{sussman12,rohe2011spectral}.
However, somewhat stronger and more elegant results can be shown by considering the spectral norm instead. In particular, a nontrivial body of recent work on matrix concentration implies that, under certain assumptions on the sparsity of $\bP$, the spectral norm of $\bA-\bP$ can be well-controlled. We focus on the following important result of Oliveira \cite{oliveira2009concentration} and Tropp \cite{Tropp2015} and  further improvements of Lu and Peng \cite{lu13:_spect} and Lei and Rinaldo \cite{rinaldo_2013}, all of which establish that the $\bA$ and $\bP$ are close in spectral norm.
\begin{theorem} [Spectral norm control of $\bA-\bP$ from \cite{oliveira2009concentration, Tropp2015}]\label{thm:oliveira}
		Suppose
Let $\bA$ be the adjacency matrix of an independent-edge random graph on $[n]$ with matrix of edge probabilities $\bP$. For any constant $c$, there exists another constant $C$, independent of $n$ and $\bP$, such that if $\delta(\bP)>C \ln n$, then for any $n^{-c}<\eta<1/2$,
\begin{equation}\label{eq:Oliveira_orig}
\Pr \left(\|\bA-\bP\| \leq 4\sqrt{\delta(P) \ln (n/\eta)}\right) \geq 1-\eta.
\end{equation} 
\end{theorem}
In \cite{lu13:_spect}, the authors give an improvement under slightly stronger density assumptions\footnote{A similar bound is provided in \cite{rinaldo_2013}, but with $\delta(\mathbf{P})$ defined as $\delta(\mathbf{P}) = n \max_{ij} \mathbf{P}_{ij}$ and a density assumption of the form $(n \max_{ij} \mathbf{P}_{ij}) > (\log n)^{1 + a}$.}:
\begin{theorem}[Spectral norm control of $\bA-\bP$  \cite{lu13:_spect}]\label{thm:lu_peng} 
	With notation as above, suppose there exist positive constants such that for $n$ sufficiently large, $\delta(\bP)>(\log n)^{4 + a}$.
	Then for any $c > 0$ there exists a constant $C$ depending on $c$ such that
	\begin{equation}\label{eq:lu_peng_spec_norm}
	\mathbb{P}\left(\|\bA-\bP\|\leq 2 \sqrt{\delta(\bP)}+C\delta^{1/4}(\bP) \ln n\right) \geq 1 - n^{-c}.
	\end{equation}
	\end{theorem} 
	
\subsection{Matrix perturbations and spectral decompositions}

The above results formalize our intuition that $\bA$ provides a ``reasonable" estimate for $\bP$. Moreover, in the RDPG case, where $\bP$ is of low rank and is necessarily positive semidefinite, these results have implications about the nonnegativity of the eigenvalues of $\bA$. Specifically, we use Weyl's Theorem to infer bounds on the differences between the eigenvalues of $\bA$ and $\bP$ from the spectral norm of their difference, and the Gerschgorin Disks Theorem to infer lower bounds on the maximum row sums of $\bP$ from assumptions on the eigengap of $\bP$ (since both $\bP$ and $\bA$ are nonnegative matrices, one could also obtain the same lower bounds by invoking the Perron-Frobenius Theorem). For completeness, we recall the Gerschgorin Disks Theorem and Weyl's Theorem. The former relates the eigenvalues of a matrix to the sums of the absolute values of the entries in each row, and the latter establishes bounds on the differences in eigenvalues between a matrix and a perturbation.
	\begin{theorem}[Gerschgorin Disks \cite{horn85:_matrix_analy}]
		\label{thm:Gerschgorin}
		Let $\bH$ be a complex $n\times n$ matrix, with entries $\bH_{ij}$\,. For $i\in \{1,\cdots ,n\}$let $R_{i}=\sum _{j\neq {i}}\left|\bH_{ij}\right|$. Let the $i$th \textrm{Gerschgorin disk} $D(\bH_{ii},R_{i})$ be the closed disk centered at $\bH_{ii}$ with radius $R_{i}$.
		Then every eigenvalue of $\bH$ lies within at least one of the Gershgorin discs $D(\bH_{ii},R_{i})$.
	\end{theorem}
	\begin{theorem} [Weyl \cite{horn85:_matrix_analy}] Let $\bM, \bH,$ and $\bR$ be $n\times n$ Hermitian matrices, and suppose $\bM=\bH+\bR$. Suppose $\bH$ and $\bR$ have eigenvalues $\nu_1 \geq \cdots \geq \nu_n$ and $r_1 \geq \cdots \geq r_n$, respectively.  Suppose the eigenvalues of $\bM$ are given by $\mu_1 \geq \cdots \geq \mu_n$. Then
	$$\nu_i + r_n \leq \mu_i \leq \nu_i + r_1$$ 
	\end{theorem}
	From our random graph model assumptions and our graph density assumptions, we can conclude that with for sufficiently large $n$, the top $d$ eigenvalues of $\bA$ will be nonnegative:
	\begin{remark}[Nonnegativity of the top $d$ eigenvalues of $\bA$]\label{rem:nonneg_evals_A} Suppose $\mathbf{A} \sim \mathrm{RDPG}(\mathbf{X})$. Since $\bP=\bX\bX^{\top}$, it is necessarily positive semidefinite, and thus has nonnegative eigenvalues. If we now assume that $\gamma(\mathbf{P}) > c_0$ for some constant $c_0$, then along with the Gershgorin Disks Theorem, guarantee that the top $d$ eigenvalues of $\bP$ are all of order $\delta(\bP)$, and our rank assumption on $\bP$ mandates that the remaining eigenvalues be zero.  If $\delta(\bP)> \log^{4+a'}n$, the spectral norm bound in \eqref{eq:lu_peng_spec_norm} applies, ensuring that for $n$ sufficiently large, $\|A-P\|_2 \sim O(\sqrt{\delta(\bP)})$ with high probability. Thus, by Weyl's inequality, we see that the top $d$ eigenvalues of $\bA$ are, with high probability, of order $\delta$, and the remaining are, with high probability, within $\sqrt{\delta}$ of zero.    
	\end{remark}

Since $\bP=\bX \bX^{\top}=\UP \SP^{1/2} (\UP \SP^{1/2})^{\top}$ and $\bA$ is close to $\bP$, it is intuitively appealing to conjecture that, in fact, $\hat{\bX}=\UA \SA^{1/2}$ should be close to some rotation of $\UP \SP^{1/2}$. That is, if $\bX$ is the matrix of true latent positions----so $\bX\bW=\UP \SP^{1/2}$ for some orthogonal matrix $\bW$---then it is plausible that $\|\Xhat-\bX \bW\|_F$ ought to be comparatively small. To make this precise, however, we need to understand how both eigenvalues and eigenvectors of a matrix behave when the matrix is perturbed.  Weyl's inequality \cite{horn85:_matrix_analy} addresses the former. The impact of matrix perturbations on associated eigenspaces is significantly more complicated, and the Davis-Kahan Theorem \cite{davis70, Bhatia1997} provides one approach to the latter.  The Davis-Kahan has a significant role in several approaches to spectral estimation for graphs: for example, Rohe, Chatterjee, and Yu leverage it in \cite{rohe2011spectral} to prove the accuracy of spectral estimates in high-dimensional stochastic blockmodels. The version we give below is from \cite{DK_usefulvariant}, which is a user-friendly guide to the the Davis-Kahan Theorem and its statistical implications. 

The Davis-Kahan Theorem is often stated as a result on canonical angles between subspaces. To that end, we recall that if $\bU$ and $\bV$ are two $n \times d$ matrices with orthonormal columns, then we define the vector of $d$ {\em canonical} or {\em principal} angles between their column spaces to be the vector $\Theta$ such that
$$\Theta=(\theta_1=\cos^{-1}\sigma_1, \cdots, \theta_d=\cos^{-1} \sigma_d )^{\top}$$
where $\sigma_1, \cdots, \sigma_d$ are the singular values of $\bU^{\top} \bV$. We define the matrix $\sin(\Theta)$ to be the $d \times d$ diagonal matrix for which $\sin(\theta)_{ii}=\sin \theta_i$. 
\begin{theorem}[A variant of Davis-Kahan \cite{DK_usefulvariant}]
	\label{thm:davis-kahan}
	Suppose $\bH$ and $\bH'$ are two symmetric $n \times n$ matrices with real entries with spectrum given by
	$\lambda_1 \geq \lambda_2 \cdots \geq \lambda_n$
	and $\lambda_1' \geq \lambda_2' \cdots \geq \lambda_n'$, respectively; and let $\bv_1, \cdots, \bv_n$ and $\bv_1', \cdots, \bv_n'$ denote their corresponding orthonormal eigenvectors. Let $0 < d \leq n$ be fixed, and let $\bV$ be the matrix of whose columns are the eigenvectors $v_1, \cdots, v_d$, and similarly $\bV'$ the matrix whose columns are the eigenvectors $\bv_1', \cdots \bv_n'$.   Then
$$\|\sin(\Theta)\|_F \leq \frac{2 \sqrt{d}\|\bV-\bV'\|}{\lambda_d(\bH)-\lambda_{d+1}(\bH)}.$$
\end{theorem}
Observe that if we assume that $\bP$ is of rank $d$ and has a sufficient eigengap, the Davis-Kahan Theorem gives us an immediate bound on the spectral norm of the difference between $\UA \UA^{\top}$ and $\UP \UP^{\top}$ in terms of this eigengap and the spectral norm difference of $\bA-\bP$, namely:
\begin{equation*}
\| \UA \UA^{\top} - \UP\UP^{\top} \|=\max_{i} \|\sin(\theta_i)\|
\le \frac{ C \| \bA - \bP \| }{ \lambda_d(\bP) }.
\end{equation*}
Recall that the Frobenius norm of a matrix $\bH$ satisfies
$$(\|\bH\|_F)^2=\sum_{i,j} \bH^2_{ij}=\tr (\bH^{\top} \bH) \geq \|\bH\|^2$$
and further that if $\bH$ is of rank $d$, then
$$(\|\bH\|_F)^2 \leq d\|\bH\|^2$$
and hence for rank $d$ matrices, spectral norm bounds are easily translated into bounds on the Frobenius norm.  It is worth noting that \cite{rohe2011spectral} guarantees that a difference in projections can be transformed into a difference between eigenvectors themselves: namely,
given the above bound for $\|\UA \UA^{\top} - \UP \UP^{\top}\|_F$, there is a constant $C$ and an orthonormal matrix $W \in \R^{d \times d}$ such that
\begin{equation}
\label{eq:Davis_Kahan_variant1}
\|\UP \bW-\UA\|_F \leq C\sqrt{d} \frac{\|\bA-\bP\|}{\lambda_d(\bP)}.
\end{equation}


While these important results provide the backbone for much of our theory, the detailed bounds and distributional results described in the next section rely on a decomposition of $\hat{\bX}$ in terms of $(\bA - \bP) \UP	\SP^{-1/2}$ and a remainder.
This first term can be viewed as an application of the power method for finding eigenvectors.
Additionally, standard univariate and multivariate concentration inequalities and distributional results can be readily applied to this term.
On the other hand, the remainder term can be shown to be of smaller order than the first, and much of the technical challenges of this work rely on carefully bounding the remainder term. 

\section{Spectral embeddings and estimation for RDPGs}\label{sec:ASE_Inference_RDPG}

There is a wealth of literature on spectral methods for estimating model parameters in random graphs, dating back more than half a century to estimation in simple \ErdosRenyi models. More specifically, for \ErdosRenyi   graphs, we would be remiss not to point to Furedi and Komlos's classic work \cite{furedi1981eigenvalues} on the eigenvalues and eigenvectors of the adjacency matrix of a E-R graph. Again, despite their model simplicity, \ErdosRenyi  graphs veritably teem with open questions; to cite but one example, in a very recent manuscript, Arias-Castro and Verzhelen \cite{Arias-Castro_ER} address, in the framework of hypothesis testing, the question of subgraph detection within an ER graph.

Moving up to the slightly more heterogeneous stochastic block model, we again find a rich literature on consistent block estimation in stochastic block models. Fortunato  \cite{fortunato} provides an overview of partitioning techniques for graphs in general, and
consistent partitioning of stochastic block models for two blocks was accomplished by Snijders
and Nowicki \cite{snijders_nowicki} and for equal-sized blocks by Condon and Karp  in 2001.
For the more general case, Bickel and Chen \cite{bickel_chen_2009} demonstrate a stronger version of
consistency via maximizing Newman-Girvan modularity \cite{newman2006modularity} and other modularities. For
a growing number of blocks, Choi et al. \cite{bickel2013asymptotic} prove consistency of likelihood based
methods, and Bickel et al. \cite{bickel2011method} provide a method to consistently estimate the stochastic
block model parameters using subgraph counts and degree distributions. This work and the
work of Bickel and Chen \cite{Bickel2009} both consider the case of very sparse graphs. In \cite{Airoldi2008}, Airoldi et al define the important generalization of a {\em mixed-membership} stochastic blockmodel, in which block membership may change depending on vertex-to-vertex interactions, and the authors demonstrate methods of inference for the mixed membership and block probabilities.

Rohe, Chatterjee and Yu show in \cite{rohe2011spectral} that spectral embeddings of the Laplacian give consistent estimates of block memberships in a stochastic block model, and one of the earliest corresponding results on the consistency of the adjacency spectral embedding is given by Sussman, Tang, Fishkind, and Priebe in \cite{STFP-2011}. In \cite{STFP-2011}, it is proved that for a stochastic block model with $K$ blocks and a rank $d$ block probability matrix $B$, clustering the rows of the adjacency spectral embedding via $k$-means clustering (see \cite{pollard81:_stron_k}) results in at most $\log n$ vertices being misclassified.  An improvement to this can be found in \cite{fishkind2012consistent}, where consistency recovery is possible even if the rank of the embedding dimension is unknown. 

In \cite{lyzinski13:_perfec}, under the assumption of distinct eigenvalues for the second moment matrix $\Delta$ of a random dot product graph, it is shown that clustering the rows of the adjacency spectral embedding results in asymptotically almost surely perfect recovery of the block memberships in a stochastic blockmodel---i.e. for sufficiently large $n$, the probability of all vertices being correctly assigned is close to 1. An especially strong recovery is exhibited here: it is shown that in the $2 \rightarrow \infty$ norm, $\Xhat$ is sufficiently close to a rotation of the the true latent positions. In fact, each {\em row} in $\Xhat$ is within $C\log n/\sqrt{n}$ of the  corresponding row in $\bX$. Unlike 
a Frobenius norm bound, in which it is possible that some rows of $\Xhat$ may be close to the true positions but others may be significantly farther away, this $2 \rightarrow \infty$ bound implies that the adjacency spectral embedding has a {\em uniform} consistency. 

Furthermore, \cite{tang14:_semipar} gives a nontrivial tightening of the Frobenius norm bound for the difference between the (rotated) true and estimated latent positions: in fact the Frobenius norm is not merely bounded from above by a term of order $\log n$, but rather concentrates around a {\em constant}.  This constant-order Frobenius bound forms the basis of a principled two-sample hypothesis test for determining whether two random dot product graphs have the same generating latent positions (see Section \ref{subsec:Testing}). 

In \cite{lyzinski15_HSBM}, the $2 \to \infty$-norm bound is extended even to the case when the second moment matrix $\Delta$ does not have distinct eigenvalues. This turns out to be critical in guaranteeing that the adjacency spectral embedding can be effectively deployed for community detection in hierarchical block models. We present this bound for the $2 \to \infty$ norm in some detail here; it illustrates the confluence of our key techniques and provides an effective roadmap for several subsequent results on asymptotic normality and two-sample testing.

\subsection{Consistency of latent position estimates} \label{subsec:Estimation}

We state here one of our lynchpin results on consistency, in the $2 \rightarrow \infty$ norm, of the adjacency spectral embedding for latent position estimation of a random dot product graph. We given an outline of the proof here, and refer the reader to the Appendix \ref{sec:Appendix} for the details, which essentially follow the proof given in \cite{lyzinski15_HSBM}. We emphasize our setting is a sequence of random dot product graphs $\mathbf{A} \sim \mathrm{RDPG}(\mathbf{X}_n)$ for increasing $n$ and thus we make the following density assumption on $\mathbf{P}_n$ as $n$ increases:
\begin{assumption}\label{ass:max_degree_assump}
Let $\mathbf{A}_n \sim \mathrm{RDPG}(\mathbf{X}_n)$ for $n \geq 1$ be a sequence of random dot product graphs with $\mathbf{A}_n$ being a $n \times n$ adjacency matrix. Suppose that $\mathbf{X}_n$ is of rank $d$ for all $n$ sufficiently large. Suppose also that there exists constants $a > 0$ and $c_0 > 0$ such that for all $n$ sufficiently large,
	$$\delta(\bP_n)= \max_{i} \sum_{j=1}^n (\mathbf{P}_n)_{ij} \geq \log^{4+a}(n); \quad \gamma(\mathbf{P}_n) = \frac{\lambda_d(\mathbf{P}_n)}{\delta(\bP_n)} \geq c_0.$$
	\end{assumption}

Our consistency result for the $2 \rightarrow \infty$ norm is Theorem~\ref{thm:minh_sparsity} below. In the proof of this particular result, we consider a particular random dot product graph with non-random (i.e. fixed) latent positions, but our results apply also to the case of random latent positions. In Section \ref{subsec:Distributional}, where we provide a central limit theorem, we focus on the case in which the latent positions are themselves random. Similarly, in Section \ref{subsec:Testing}, in our analysis of the semiparametric two-sample hypothesis test for the equality of latent positions in a pair of random dot product grahs, we return to the setting in which the latent positions are fixed, and in the nonparametric hypothesis test of equality of distributions, we analyze again the case when the latent positions are random.  It is convenient to be able to move fluidly between the two versions of a random dot product graph, adapting our results as appropriate in each case.

 In the Appendix (\ref{sec:Appendix}), we give a detailed proof of Theorem~\ref{thm:minh_sparsity} and we point out the argument used therein also sets the stage for the central limit theorem for the rows of the adjacency spectral embedding given in Subsection~\ref{subsec:Distributional}.

\begin{theorem}\label{thm:minh_sparsity}
	Let $\bA_n \sim \mathrm{RDPG}(\bX_n)$ for $n \geq 1$ be a sequence of random dot product graphs where the $\bX_n$ is assumed to be of rank $d$ for all $n$ sufficiently large. Denote by $\hat{\bX}_n$ the adjacency spectral embedding of $\bA_n$ and let $(\hat{\bX}_{n})_{i}$ and $(\bX_n)_{i}$ be the $i$-th row of $\hat{\bX}_n$ and $\bX_n$, respectively. Let $E_n$ be the event that there
	exists an orthogonal transformation $\bW_n \in \mathbb{R}^{d \times d}$ such that
	\begin{equation*}
	\max_{i} \| (\hat{\bX}_n)_{i} - \bW_n (\bX_n)_{i} \| \leq 
	\frac{C d^{1/2} \log^2{n}}{\delta^{1/2}(\mathbf{P}_n)}
	\end{equation*}
	where $C > 0$ is some fixed constant and $\mathbf{P}_n = \mathbf{X}_n \mathbf{X}_n^{\top}$. Then $E_n$ occurs asymptotically almost surely; that is, $\Pr(E_n) \rightarrow 1$ as $n \rightarrow \infty$.
\end{theorem}

Under the stochastic blockmodel, previous bounds on $\|\bX-\hat{\bX}\|_F$ implied that $k$-means applied to the rows of $\hat{\bX}$ would approximately correctly partition the vertices into their the true blocks with up to $O(\log n)$ errors.
However, this Frobenius norm bound does not imply that there are no large outliers in the rows of $\bX-\hat{\bX}$, thereby precluding any guarantee of zero errors.
The improvements provided by Theorem~\ref{thm:minh_sparsity} overcome this hurdle and, as shown in \cite{lyzinski13:_perfec}, under suitable sparsity and eigengap assumptions, $k$-means applied to $\hat{\bX}$ will {\em exactly correctly} partition the vertices into their true blocks.
This implication demonstrates the importance of improving the overall bounds and in focusing on the correct metrics for a given task---in this case, for instance, block identification.

For a brief outline of the proof of this result, we note several key ingredients.
First is a lemma guaranteeing the existence of an orthogonal matrix $\bW^*$ such that
$$\|\bW^* \SA - \SP \bW^{*}\|_F=O((\log n)\delta^{1/2}(P)$$
That is, there is an approximate commutativity between right and left multiplication of the corresponding matrices of eigenvalues by this orthogonal transformation.  The second essential component is, at heart, a bound inspired by the power method.  Specifically, we show that there exists an orthogonal matrix
$$\|\Xhat -\bX\bW\|=\|(\bA -\bP) \UP \SP^{-1/2}\|_F + O((\log n) \delta^{-1/2})$$
Finally, from this point, establishing the bound on the $2 \to \infty$ norm is a consequence of Hoeffding's inequality applied to sums of the form
$$ \sum_{k} (\bA_{ik}-\bP_{ik} \bU_{kj})$$.

The $2 \rightarrow \infty$ bound in Theorem \ref{thm:minh_sparsity} has several important implications. As we mentioned, \cite{lyzinski13:_perfec} establishes an earlier form of this result, with more restrictive assumptions on the the second moment matrix, and shows how this can be used to cluster vertices in an SBM perfectly, i.e. with no vertices misclassified. In addition, \cite{lyzinski13:_perfec} shows how clustering the rows of the ASE can be useful for inference in a degree-corrected stochastic block model as well.  In Section \ref{sec:Applications}, we see that because of Theorem \ref{thm:minh_sparsity}, the adjacency spectral embedding and a novel angle-based clustering procedure can be used for accurately identifying subcommunities in an affinity-structured, hierarchical stochastic blockmodel \cite{lyzinski15_HSBM}. In the next section, we see how our proof technique here can be used to obtain distributional results for the rows of the adjacency spectral embedding.

\subsection{Distributional results for the ASE}\label{subsec:Distributional}
In the classical statistical task of parametric estimation, one observes a collection of i.i.d observations $X_1, \cdots, X_n$ from some family of distributions $F_{\theta}: \theta \in \Theta$, where $\Theta$ is some subset of Euclidean space, and one seeks to find a consistent estimator $T(X_1, \cdots, X_n)$ for $\theta$. As we mentioned in Section~\ref{sec:Intro}, often a next task is to determine the asymptotic distribution, as $n \rightarrow \infty$, of a suitable normalization of this estimator $T$. Such distributional results, in turn, can be useful for generating confidence intervals and testing hypotheses.

We adopt a similar framework for random graph inference.  In the previous section, we demonstrate the consistency of the adjacency spectral embedding for the true latent position of a random dot product graph.  In this section, we establish the asymptotic normality of the rows of this embedding and the Laplacian spectral embedding.
In the subsequent section, we examine how our methodology can be deployed for multisample graph hypothesis testing.

We emphasize that distributional results for spectral decompositions of random graphs are comparatively few. The classic results of
\Furedi and \Komlos \cite{furedi1981eigenvalues} describe the eigenvalues of the \ErdosRenyi random graph and the work of Tao and Vu \cite{tao2012random} is focused on distributions of eigenvectors of more general random matrices under moment restrictions,
but \cite{athreya2013limit} and \cite{tang_lse} are among the only references for central limit theorems for spectral decompositions of adjacency and Laplacian matrices for a wider class of independent-edge random graphs than merely the \ErdosRenyi model. Apart from their inherent interest, these limit theorems point us to current open questions on efficient estimation and the relative merits of different estimators and embeddings, in part by rendering possible a comparison of asymptotic variances and allowing us to quantify relative efficiency (see \cite{runze_law_large_graphs} and to precisely conjecture a decomposition of the sources of variance in different spectral embeddings for multiple graphs (see \cite{levin_omni_2017}).  

Specifically, we show that for a $d$-dimensional random dot product graph with i.i.d latent positions, there exists a sequence of $d \times d$ orthogonal matrices $\bW_n$ such that for any row index $i$, $\sqrt{n}(\bW_n (\Xhat_n)_{i} - (\bX_n)_i)$ converges as $n \rightarrow \infty$ to a mixture of multivariate normals.
\begin{theorem}[Central Limit Theorem for the rows of the ASE]\label{thm:clt_orig_but_better}
	Let $(\bA_n, \bX_n) \sim \mathrm{RDPG}(F)$ be a sequence of adjacency matrices and associated latent positions of a $d$-dimensional random dot product graph according to an inner product distribution $F$. Let $\Phi(\bx,\bSigma)$ denote the cdf of a (multivariate)
	Gaussian with mean zero and covariance matrix $\bSigma$,
	evaluated at $\bx \in \R^d$. Then
	there exists a sequence of orthogonal $d$-by-$d$ matrices
	$( \Wn )_{n=1}^\infty$ such that for all $\bm{z} \in \R^d$ and for any fixed index $i$,
	$$ \lim_{n \rightarrow \infty}
	\Pr\left[ n^{1/2} \left( \Xhat_n \Wn - \bX_n \right)_i
	\le \bm{z} \right]
	= \int_{\supp F} \Phi\left(\bm{z}, \bSigma(\bx) \right) dF(\bx), $$
	where
	\begin{equation}
	\label{def:sigma}\bSigma(\bx) 
	= \Delta^{-1} \E\left[ (\bx^{\top} \bX_1 - ( \bx^{\top} \bX_1)^2 ) \bX_1 \bX_1^{\top} \right] \Delta^{-1}; \quad \text{and} \,\, \Delta = \mathbb{E}[\bX_1 \bX_1^{\top}].
	\end{equation}
\end{theorem}
We also note the following important corollary of Theorem~\ref{thm:clt_orig_but_better} for when $F$ is a mixture of $K$ point masses, i.e., $(\mathbf{X}, \mathbf{A}) \sim \mathrm{RDPG}(F)$ is a $K$-block stochastic blockmodel graph. Then for any fixed index $i$, the event that $\bX_i$ is assigned to block $k \in \{1,2,\dots,K\}$ has non-zero probability and hence one can condition on the block assignment of $\bX_i$ to show that the conditional distribution of $\sqrt{n}(\mathbf{W}_n (\hat{\bX}_n)_{i} - (\bX_n)_i)$ converges to a multivariate normal. This is in contrast to the unconditional distribution being a mixture of multivariate normals as given in Theorem~\ref{thm:clt_orig_but_better}.

\begin{corollary}[SBM]
\label{cor:ase_normality_sbm}
Assume the setting and notations of Theorem~\ref{thm:clt_orig_but_better} and let 
$$F = \sum_{k=1}^{K} \pi_{k} \delta_{\nu_k}, \quad \pi_1, \cdots, \pi_K > 0, \sum_{k} \pi_k = 1$$
be a mixture of $K$ point masses in $\mathbb{R}^{d}$ where $\delta_{\nu_k}$ is the Dirac delta measure at $\nu_k$. 
Then there exists a sequence of orthogonal matrices $\mathbf{W}_n$ such that for all $\bm{z} \in \mathbb{R}^{d}$ and for any fixed index $i$,
\begin{equation}
\mathbb{P}\Bigl\{\sqrt{n}(\mathbf{W}_n \hat{\bX}_n - \mathbf{X}_n)_{i} \leq \bm{z}  \mid \bX_i = \nu_k \Bigr\} \longrightarrow \Phi(\bm{z}, \Sigma_k)
\end{equation}
where $\Sigma_k = \Sigma(\nu_k)$ is as defined in Eq.~\eqref{def:sigma}. 
\end{corollary}

We relegate the full details of the proof of this central limit theorem to the Appendix, in Section \ref{subsec:CLT_proofdetails}, but a few points bear noting here. First, both Theorem~\ref{thm:clt_orig_but_better} and Corollary~\ref{cor:ase_normality_sbm} are very similar to results proved in \cite{athreya2013limit}, but with the crucial difference being that we no longer require that the second moment matrix $\Delta = \mathbb{E}[\bX_1 \bX_1^{\top}]$ of $\bm{X}_1 \sim F$ have distinct eigenvalues (for more details, see \cite{tang_lse}). As in \cite{athreya2013limit}, our proof here depends on writing the difference between a row of the adjacency spectral embedding and its  corresponding latent position as a pair of summands: the first, to which a classical Central Limit Theorem can be applied, and the second, essentially a combination of residual terms, which we show, using techniques similar to those in the proof of Theorem \ref{thm:minh_sparsity}, converges to zero. The weakening of the assumption of distinct eigenvalues necessitates significant changes from \cite{athreya2013limit} in how to bound the residual terms, because \cite{athreya2013limit} adapts a result of \cite{bickel_sarkar_2013}---the latter of which depends on the assumption of distinct eigenvalues---to control these terms. Here, we resort to somewhat different methodology: we prove instead that analogous bounds to those in \cite{lyzinski15_HSBM,tang_lse} hold for the estimated latent positions and this enables us to establish that here, too, the rows of the adjacency spectral embedding are also approximately normally distributed. 

We stress that this central limit theorem depends on a delicate bounding of a sequence of so-called residual terms, but its essence is straightforward.  Specifically, there exists an orthogonal transformation $\bW^*$ such that we can write the $i$th row of the matrix 
$$\sqrt{n}(\UA \SA^{1/2}-\UP \SP^{1/2}\bW^{*})$$ as 
\begin{equation}\label{eq:CLT_basic_explanation}
\sqrt{n}(\UA \SA^{1/2}-\UP \SP^{1/2}\bW^{*})_i=\sqrt{n}(\bA -\bP) \UP \SP^{-1/2}\bW^{*})_i + \textrm{Residual terms}
\end{equation}
where the residual terms are all of order $O(n^{-1/2} \log n)$ in probability.  Now, to handle the first term in Eq.\eqref{eq:CLT_basic_explanation}, we can condition on a fixed latent position $\bX_i={\bf x}_i$, and when this is fixed, the classical Lindeberg-Feller Central Limit Theorem establishes the asymptotic normality of this term. The order of the residual terms then guarantees, by Slutsky's Theorem, the desired asymptotic normality of the gap between estimated and true latent positions, and finally we need only integrate over the possible latent positions to obtain a mixture of normals. 

\subsection{An example under the stochastic block model}

To illustrate Theorem~\ref{thm:clt_orig_but_better}, we
consider random graphs generated according to a
stochastic block model with parameters
\begin{equation}
\label{eq:1}
B = \begin{bmatrix} 0.42 & 0.42 \\ 0.42 & 0.5 \end{bmatrix}
\quad \text{and} \quad \pi = (0.6,0.4).
\end{equation}
In this model, each node is either in block 1 (with probability 0.6)
or block 2 (with probability 0.4). Adjacency probabilities are
determined by the entries in $B$ based on the block memberships of the
incident vertices.  The above stochastic blockmodel corresponds to a random dot product
graph model in $\mathbb{R}^{2}$ where the distribution $F$ of the latent
positions is a mixture of point masses located at $x_1\approx (0.63,
-0.14)$ (with prior probability $0.6$) and $x_2\approx (0.69, 0.13)$
(with prior probability $0.4$).

We sample an adjacency matrix $\bA$ for graphs on $n$ vertices from the
above model for various choices of $n$.  For each graph $G$, let
$\Xhat\in\Re^{n\times 2}$ denote the embedding of $A$ and let
$\Xhat_{i\cdot}$ denote the $i$th row of $\Xhat$. In
Figure~\ref{fig:clusplot}, we plot the $n$ rows of $\Xhat$ for the
various choices of $n$.  The points are denoted by symbols according to the block
membership of the corresponding vertex in the stochastic blockmodel. The ellipses show
the 95\% level curves for the distribution of $\hat{X}_i$ for each
block as specified by the limiting distribution.

\begin{figure*}[!htbp]
	\centering
	\subfloat[$n = 1000$]{
		\includegraphics[width=5.5cm]{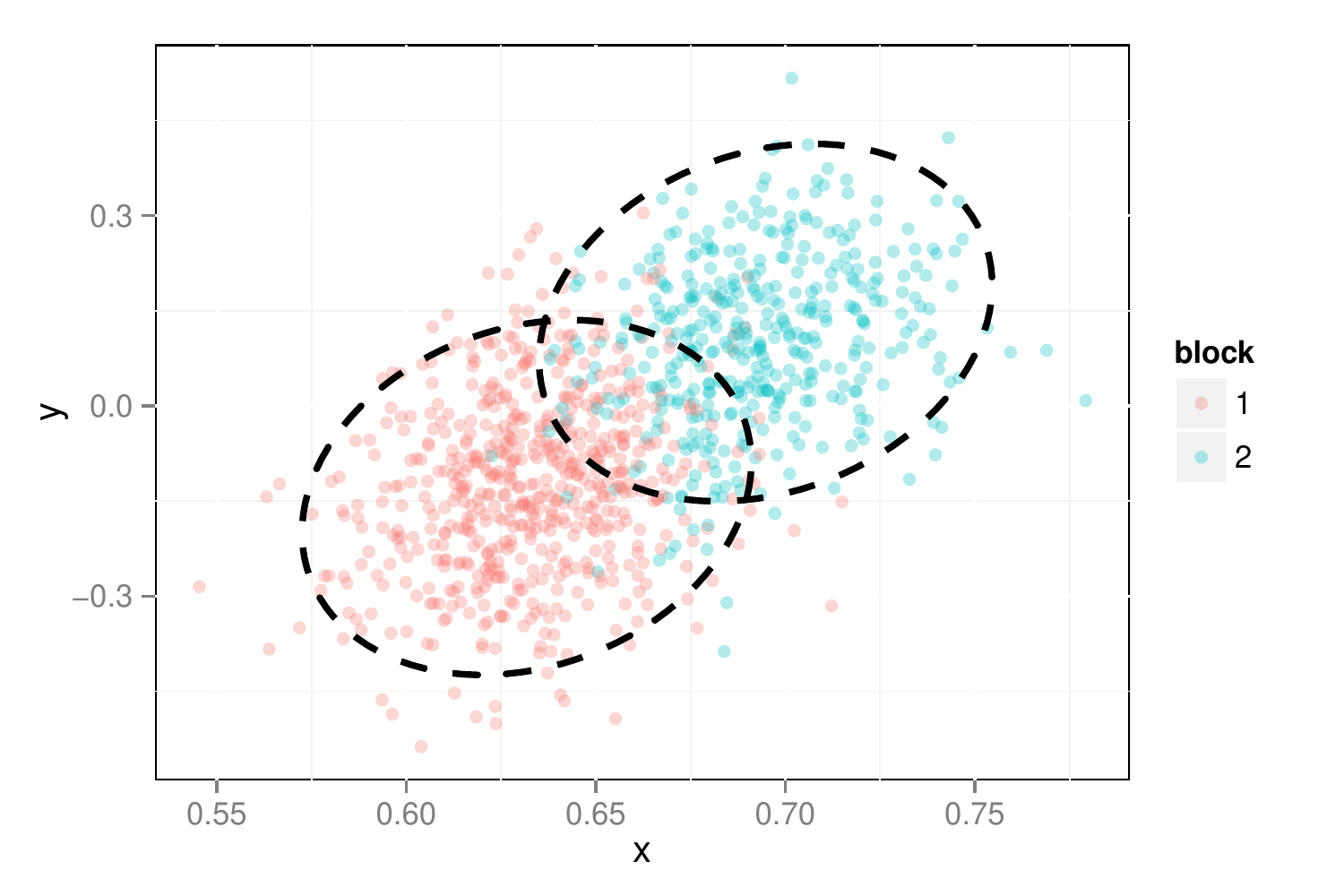}
	}
	\hfil
	\subfloat[$n = 2000$]{
		\includegraphics[width=5.5cm]{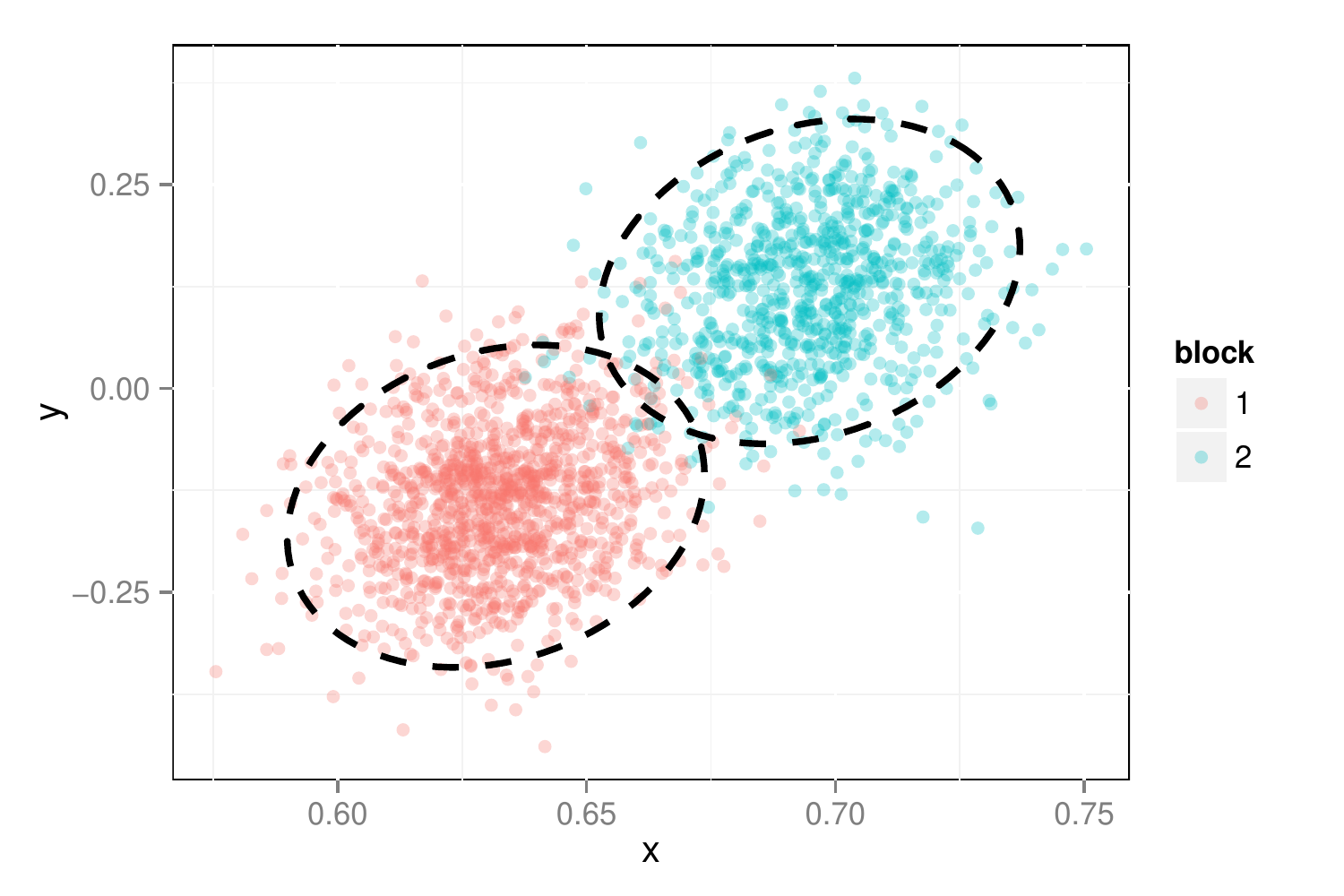}
	}
	\hfil
	\subfloat[$n = 4000$]{
		\includegraphics[width=5.5cm]{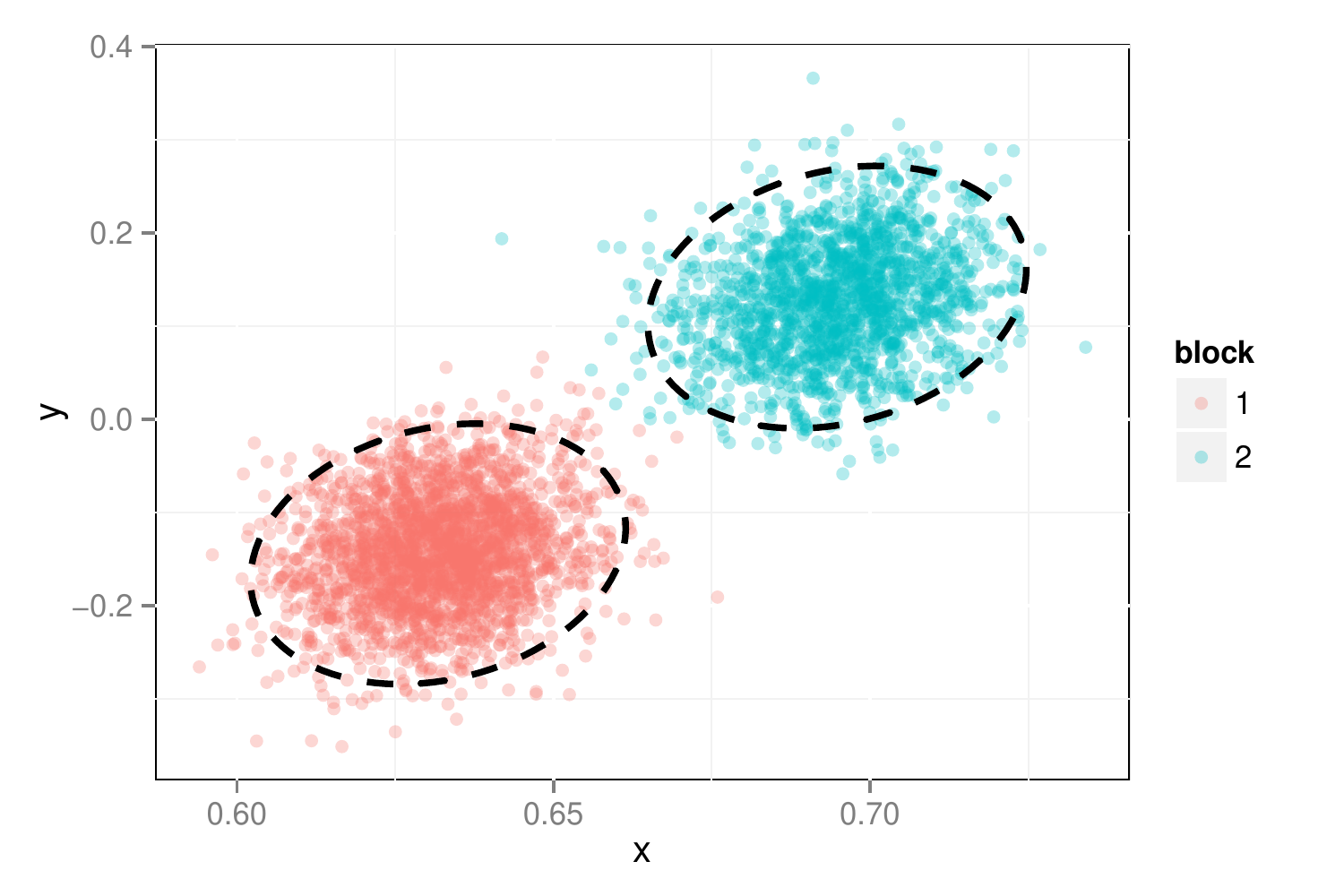}
		\label{fig5:subfig_scan}
	}
	\hfil
	\subfloat[$n = 8000$]{
		\includegraphics[width=5.5cm]{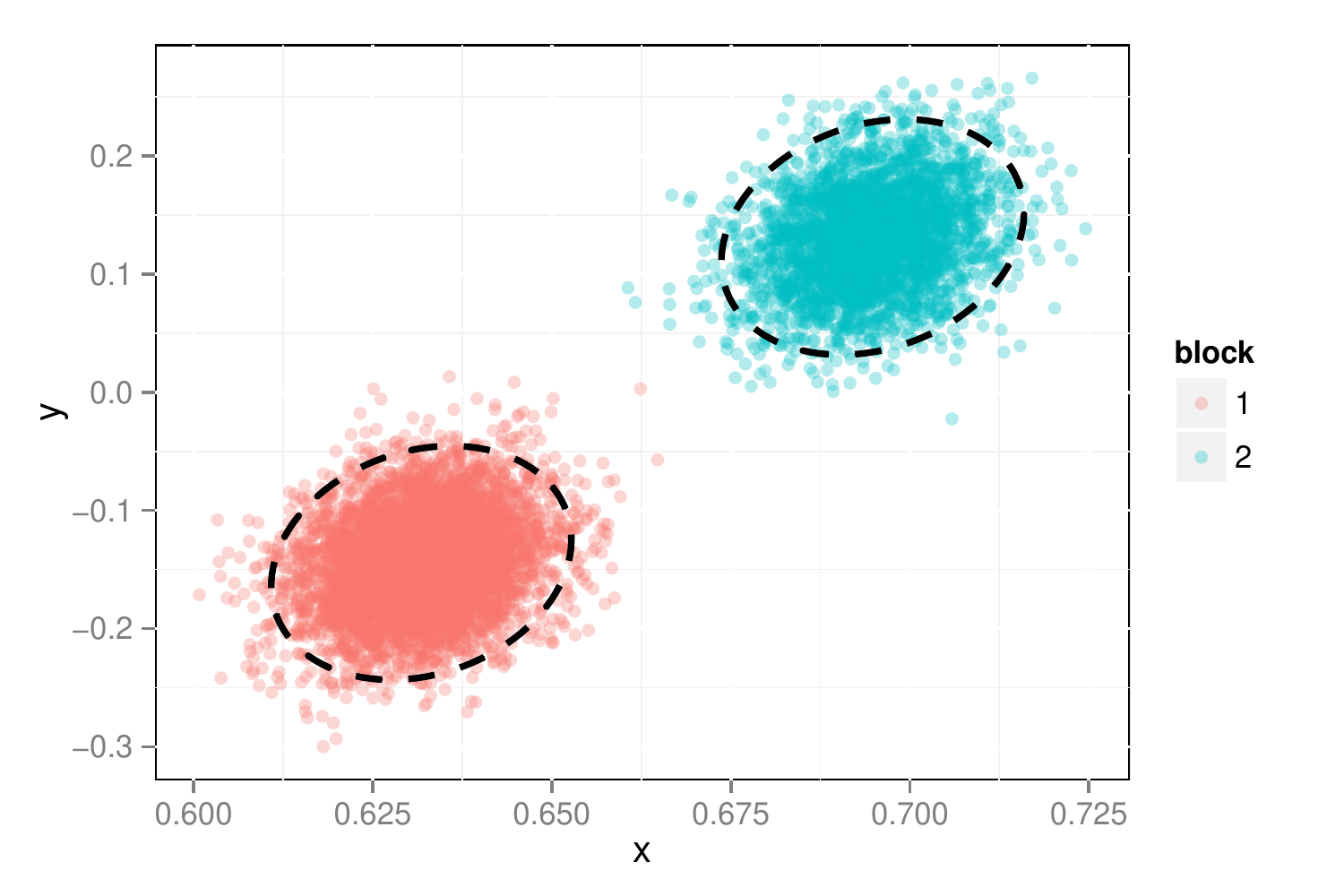}
	}
	\caption{Plot of the estimated latent positions in a two-block
		stochastic blockmodel graph on $n$ vertices. The
		points are denoted by symbols according to the blockmembership of the
		corresponding vertices. Dashed ellipses give the 95\% level curves
		for the distributions as specified in
		Theorem~\ref{thm:clt_orig_but_better}.}
	\label{fig:clusplot}
\end{figure*}
We then estimate the covariance matrices for the residuals. The
theoretical covariance matrices are given in the last column of
Table~\ref{tab:cov}, where $\Sigma_{1}$ and $\Sigma_{2}$ are the
covariance matrices for the residual $\sqrt{n}(\hat{\bX}_i - \bX_i)$ when
$X_i$ is from the first block and second block, respectively. The
empirical covariance matrices, denoted $\hat{\Sigma}_1$ and
$\hat{\Sigma}_2$, are computed by evaluating the sample covariance of
the rows of $\sqrt{n}\Xhat_i$ corresponding to vertices in block 1
and 2 respectively.  The estimates of the covariance matrices are
given in Table~\ref{tab:cov}.  We see that as $n$ increases, the
sample covariances tend toward the specified limiting covariance
matrix given in the last column.

\begin{table}[htbp]
	\footnotesize
	\begin{center}
		\begin{tabular}{cccccc}
			$n$ &  2000 & 4000 & 8000 & 16000 & $\infty$  \\ 
			$\hat{\Sigma}_1$ &%
			$\begin{bmatrix} .58 & .54 \\ .54 & 16.56 \end{bmatrix}$ &%
			$\begin{bmatrix} .58 & .63 \\ .63 & 14.87 \end{bmatrix}$ &%
			$\begin{bmatrix} .60 & .61 \\ .61 & 14.20 \end{bmatrix}$ &%
			$\begin{bmatrix} .59 & .58 \\ .58 & 13.96 \end{bmatrix}$ &%
			$\begin{bmatrix} .59 & .55 \\ .55 & 13.07 \end{bmatrix}$ \\ \\
			$\hat{\Sigma}_2$ &%
			$\begin{bmatrix} .58 & .75 \\ .75 & 16.28 \end{bmatrix}$ &%
			$\begin{bmatrix} .59 & .71 \\ .71 & 15.79 \end{bmatrix}$ &%
			$\begin{bmatrix} .58 & .54 \\ .54 & 14.23 \end{bmatrix}$ &%
			$\begin{bmatrix} .61 & .69 \\ .69 & 13.92 \end{bmatrix}$ &%
			$\begin{bmatrix} .60 & .59 \\ .59 & 13.26 \end{bmatrix}$ \\
		\end{tabular}
	\end{center}
	\caption{The sample covariance matrices for $\sqrt{n}(\hat{X}_i-X_i)$
		for each block in a stochastic blockmodel with two blocks. Here
		$n \in \{2000,4000,8000,16000\}$. In the
		last column are the theoretical covariance matrices for the limiting
		distribution.}
	\label{tab:cov}
\end{table}

We also investigate the effects of the multivariate normal distribution
as specified in Theorem~\ref{thm:clt_orig_but_better} on inference procedures. It is shown in
\cite{STFP-2011,sussman2012universally} that the approach of
embedding a graph into some Euclidean space, followed by inference
(for example, clustering or classification) in that space can be
consistent. However, these consistency results are, in a sense, only
first-order results. In particular, they demonstrate only that the
error of the inference procedure converges to $0$ as the number of
vertices in the graph increases. We now illustrate how
Theorem~\ref{thm:clt_orig_but_better} may lead to a more refined error analysis.

We construct a sequence of random graphs on $n$ vertices, where $n$
ranges from $1000$ through $4000$ in increments of $250$, following
the stochastic blockmodel with parameters as given above in
Eq.~\eqref{eq:1}. For each graph $G_n$ on $n$ vertices, we embed $G_n$
and cluster the embedded vertices of $G_n$ via Gaussian mixture
model and K-means clustering. Gaussian mixture model-based clustering
was done using the MCLUST
implementation of \cite{fraley99:_mclus}.
We then measure the classification error of the
clustering solution. We repeat this procedure 100 times to obtain an
estimate of the misclassification rate. The results are plotted in
Figure~\ref{fig:gmm_kmeans_bayes}. For comparison, we plot the
Bayes optimal classification error rate under the assumption that the
embedded points do indeed follow a multivariate normal mixture with
covariance matrices $\Sigma_1$ and $\Sigma_2$ as given in the
last column of Table~\ref{tab:cov}. We also plot the misclassification
rate of $(C \log{n})/n$ as given in \cite{STFP-2011}
where the constant $C$ was chosen to match the misclassification rate
of $K$-means clustering for $n = 1000$. For the number of
vertices considered here, the upper bound for
the constant $C$ from \cite{STFP-2011} will give a vacuous upper
bound of the order of $10^6$ for the misclassification rate in this
example.
Finally, we recall that the $2\to \infty$ norm bound of Theorem~\ref{thm:minh_sparsity} implies that, for large enough $n$, even the $k$-means algorithm will exactly recover the true block memberships with high probability \cite{lyzinski15:_relax}.

\begin{figure}[htbp]
	\centering
	\includegraphics[width=0.7\textwidth]{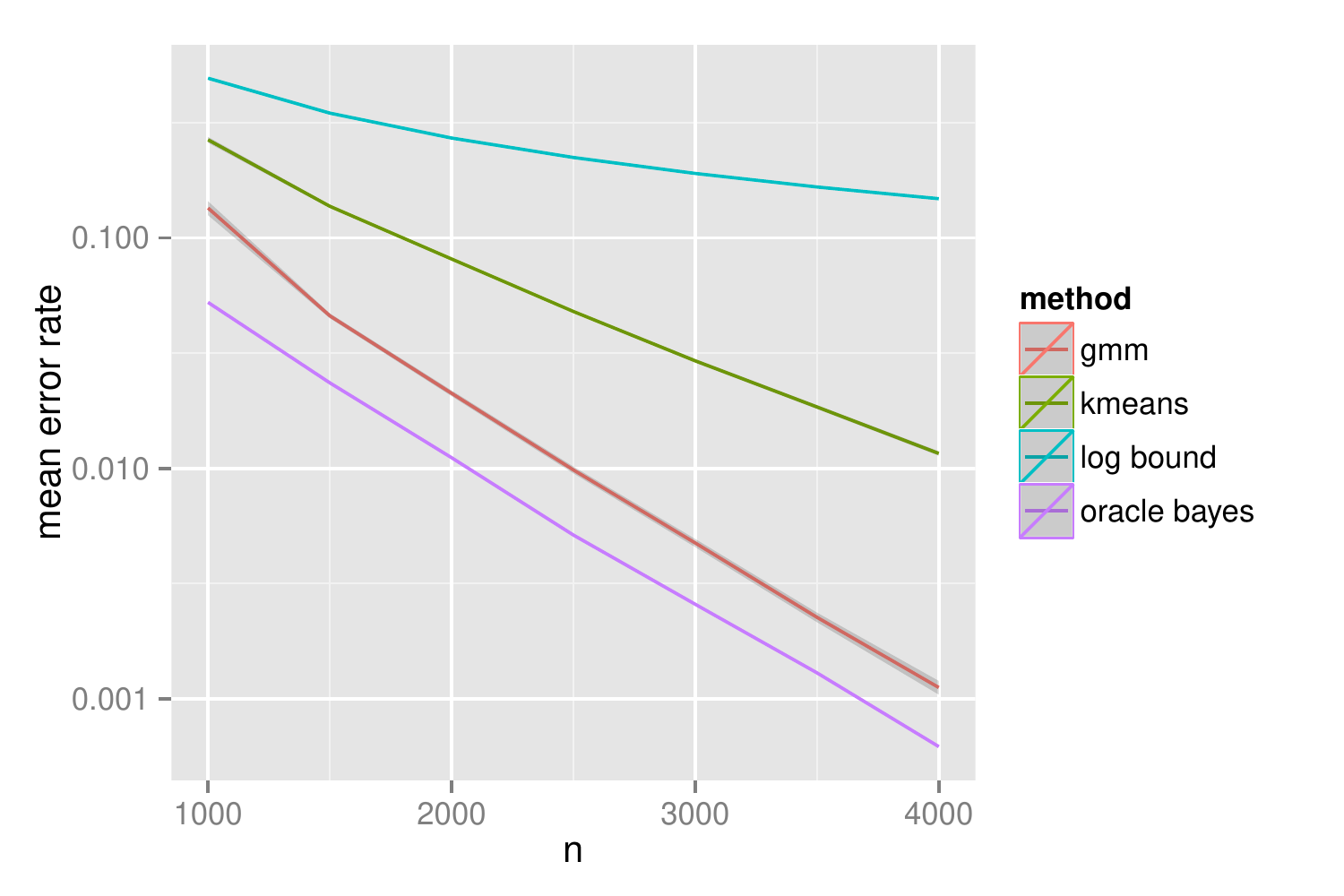}
	\caption{Comparison of classification error for Gaussian mixture
		model (red curve), K-Means (green curve), and Bayes classifier
		(cyan curve). The classification errors for each $n \in
		\{1000,1250,1500, \dots, 4000\}$ were obtained by averaging 100
		Monte Carlo iterations and are plotted on a $\log_{10}$ scale. The
		plot indicates that the assumption of a mixture of multivariate
		normals can yield non-negligible improvement in the accuracy of the inference
		procedure. The log-bound curve (purple) shows an upper bound on
		the error rate as derived in \cite{STFP-2011}. Figure duplicated from \cite{athreya2013limit}.}
	\label{fig:gmm_kmeans_bayes}
\end{figure}

For yet another application of the central limit theorem, we refer the reader to \cite{suwan14:_empbayes}, where the authors discuss the assumption of multivariate normality for estimated latent positions and how this can lead to a significantly improved empirical-Bayes framework for the estimation of block memberships in a stochastic blockmodel. 

\subsection{Distributional results for Laplacian spectral embedding}\label{subsec:lse}
We now present the analogous central limit theorem results of Section \ref{subsec:Distributional} for the normalized {\em Laplacian spectral embedding} (see Definition~\ref{def:LSE}). We first recall the definition of the Laplacian spectral embedding. 

\begin{theorem}[Central Limit Theorem for the rows of the LSE]
  \label{THM:LSE}
Let $(\mathbf{A}_n, \mathbf{X}_n) \sim \mathrm{RDPG}(F)$ for $n \geq 1$ be a sequence of $d$-dimensional random dot product graphs distributed according to some inner product distribution $F$. Let $\bm{\mu}$ and $\tilde{\Delta}$ denote the quantities
\begin{gather}
  \label{eq:mu_defn1}
  \bm{\mu} = \mathbb{E}[\bX_1] \in \mathbb{R}^{d}; \quad \tilde{\Delta} =
  \mathbb{E}\Bigl[ \frac{ \bX_1 \bX_1^{\top}}{\bX_1^{\top} \bm{\mu}} \Bigr] \in \mathbb{R}^{d \times d}.
\end{gather}
Also denote by $\tilde{\Sigma}(\bx)$ the $d \times d$ matrix
  \begin{equation}
  \label{eq:lse-sigma}
    \mathbb{E}\Bigl[ \Bigl(\frac{\tilde{\Delta}^{-1} \bX_1}{\bX_1^{\top} \bm{\mu}} - \frac{\bx}{2 \bx^{\top} \bm{\mu}}\Bigr) \Bigl(\frac{\bX_1^{\top} \tilde{\Delta}^{-1}}{\bX_1^{\top} \bm{\mu}} - \frac{\bx^{\top}}{2 \bx^{\top} \bm{\mu}}\Bigr) \frac{(\bx^{\top} \bX_1 - \bx^{\top} \bX_1 \bX_1^{\top} \bx)}{\bx^{\top} \bm{\mu}} \Bigr].
  \end{equation}
Then there exists a sequence of 
  $d \times d$ orthogonal matrices
 $(\mathbf{W}_n)$ such that for each fixed index $i$ and any $\bm{x} \in \mathbb{R}^{d}$,
  \begin{equation}
    \label{eq:Xtilde_clt}
    \Pr\Bigl\{n\bigl(  \mathbf{W}_n (\breve{\bX}_n)_{i} - \tfrac{(\bX_n)_{i}}{\sqrt{\sum_{j} (\bX_n)_i^{\top} (\bX_n)_j}} \bigr) \leq \bm{x} \Bigr\}
\longrightarrow  \int \Phi(\bm{x}, \tilde{\Sigma}(\bm{y})) dF(\bm{y})
  \end{equation}
\end{theorem}
When $F$ is a mixture of point masses---specifically, when $\mathbf{A} \sim \mathrm{RDPG}(F)$ is a stochastic blockmodel graph---we also have the following limiting conditional distribution for $n (\mathbf{W}_n (\breve{\bX}_n)_i - \tfrac{(\bX_n)_i}{\sqrt{\sum_{j} (\bX_n)_i^{\top} (\bX_n)_j}})$.
\begin{theorem}
\label{thm:clt-lse-sbm}
Assume the setting and notations of Theorem~\ref{THM:LSE} and let $$F = \sum_{k=1}^{K} \pi_{k} \delta_{\nu_k}, \quad \pi_1, \cdots, \pi_K > 0, \sum_{k} \pi_k = 1$$
be a mixture of $K$ point masses in $\mathbb{R}^{d}$. 
Then there exists a sequence of $d \times d$ orthogonal matrices $\mathbf{W}_n$ such that for any fixed index $i$,
\begin{equation}
\label{eq:lse-clt-sbm}
\mathbb{P}\Bigl\{n \bigl(\mathbf{W}_n (\breve{\bX}_n)_i - \tfrac{\nu_k}{\sqrt{\sum_{l} n_l \nu_k^{\top} \nu_l}} \bigr) \leq \bm{z}  \mid (\bX_n)_i = \nu_k \Bigr\} 
\longrightarrow \bm{\Phi}(\bm{z}, \tilde{\Sigma}_k)
\end{equation}
where $\tilde{\Sigma}_k = \tilde{\Sigma}(\nu_k)$ is as defined in Eq.~\eqref{eq:lse-sigma} and $n_k$ for $k \in \{1,2,\dots,K\}$ denote the number of vertices in $\mathbf{A}$ that are assigned to block $k$.
\end{theorem}

\begin{remark}
As a special case of Theorem~\ref{thm:clt-lse-sbm}, we note that if $\mathbf{A}$ is an Erd\H{o}s-R\'{e}nyi graph on $n$ vertices with edge probability $p^2$ -- which corresponds to a random dot product graph where the latent positions are identically $p$ -- then for each fixed index $i$, the normalized Laplacian embedding satisfies
$$ n\bigl(\breve{\bX}_i - \tfrac{1}{\sqrt{n}}\bigr) \overset{\mathrm{d}}{\longrightarrow} \mathcal{N}\bigl(0, \tfrac{1 - p^2}{4p^2}\bigr).$$
Recall that $\breve{\bX}_i$ is proportional to $1/\sqrt{d_i}$ where $d_i$ is the degree of the $i$-th vertex.
On the other hand, the adjacency spectral embedding satisfies $$\sqrt{n}(\hat{\bX}_i - p) \overset{\mathrm{d}}{\longrightarrow} \mathcal{N}(0, 1 - p^2).$$ 
As another example, let $\mathbf{A}$ be sampled from a stochastic blockmodel with block probability matrix $\mathbf{B} = \bigl[\begin{smallmatrix} p^2 & pq \\ pq & q^2 \end{smallmatrix}\bigr]$ and block assignment probabilities $(\pi, 1- \pi)$. Since $\mathbf{B}$ has rank $1$, this model corresponds to a random dot product graph where the latent positions are either $p$ with probability $\pi$ or $q$ with probability $1 - \pi$. Then for each fixed index $i$, the normalized Laplacian embedding satisfies
\begin{gather}
\label{eq:er-p-q-lse1}
 n \bigl(\breve{\bX}_i - \tfrac{p}{\sqrt{n_1 p^2 + n_2 pq}}\bigr) \overset{\mathrm{d}}{\longrightarrow} \mathcal{N}\Bigl(0, \tfrac{\pi p (1 - p^2) + (1 - \pi) q(1 - pq)}{4 (\pi p + (1 - \pi)q)^3}\Bigr) \,\, \text{if $\bX_i = p$}, \\
 \label{eq:er-p-q-lse2}
n\bigl(\breve{\bX}_i - \tfrac{q}{\sqrt{n_1 pq + n_2 q^2}}\bigr) \overset{\mathrm{d}}{\longrightarrow} \mathcal{N}\Bigl(0, \tfrac{\pi p (1 - pq) + (1 - \pi) q(1 - q^2)}{4 (\pi p + (1 - \pi)q)^3}\Bigr) \,\, \text{if $\bX_i = q$}. 
\end{gather}
where $n_1$ and $n_2 = n - n_1$ are the number of vertices of $\mathbf{A}$ with latent positions $p$ and $q$. The adjacency spectral embedding, meanwhile, satisfies
\begin{gather}
\label{eq:er-p-q-ase1}
 \sqrt{n}(\hat{\bX}_i  - p) \overset{\mathrm{d}}{\longrightarrow} \mathcal{N}\Bigl(0, \tfrac{\pi p^4(1 - p^2) + (1 - \pi) pq^3(1 - pq)}{(\pi p^2 + (1 - \pi)q^2)^2}\Bigr) \,\, \text{if $\bX_i = p$},\\
 \label{eq:er-p-q-ase2}
\sqrt{n}(\hat{\bX}_i  - q) \overset{\mathrm{d}}{\longrightarrow} \mathcal{N}\Bigl(0, \tfrac{\pi p^3q(1 - pq) + (1 - \pi) q^4(1 - q^2)}{(\pi p^2 + (1 - \pi)q^2)^2}\Bigr) \,\, \text{if $\bX_i = q$}.
\end{gather}
\end{remark}

We present a sketch of the proof of Theorem~\ref{THM:LSE} in the Appendix, Section \ref{sec:Appendix}. Due to the intricacy of the proof, however, even in the Appendix we do not provide full details; we instead refer the reader to \cite{tang_lse} for the complete proof.  Moving forward, we focus on the important implications of these distributional results for subsequent inference, including a mechanism by which to assess the relative desirability of ASE and LSE, which {\em vary depending on inference task.}

\section{Implications for subsequent inference}
The previous sections are devoted to establishing the consistency and asymptotic normality of the adjacency and Laplacian spectral embeddings for the estimation of latent positions in an RDPG.  
In this section, we describe several subsequent graph inference tasks, all of which depend on this consistency: specifically, nonparametric clustering, semiparametric and nonparametric two-sample graph hypothesis testing, and multi-sample graph inference.  

\subsection{Nonparametric clustering: a comparison of ASE and LSE via Chernoff information}
\label{subsec:chernoff}
We now discuss how the limit results of Section~\ref{subsec:Distributional} and Section~\ref{subsec:lse} provides insight into subsequent inference.
In a recent pioneering work the authors of \cite{bickel_sarkar_2013} analyze, in the context of stochastic blockmodel graphs, how the choice of spectral embedding by either the adjacency matrix or the normalized Laplacian matrix impacts subsequent recovery of the block assignments. In particular, they show that a metric constructed from the average distance between the vertices of a block and its cluster centroid for the spectral embedding can be used as a surrogate measure for the performance of the subsequent inference task, i.e., the metric is a surrogate measure for the error rate in recovering the vertices to block assignments using the spectral embedding. It is shown in \cite{bickel_sarkar_2013} that for two-block stochastic blockmodels, for a large regime of parameters the normalized Laplacian spectral embedding reduces the within-block variance (occasionally by a factor of four) while preserving the between-block variance, as compared to that of the adjacency spectral embedding.
This suggests that for a large region of the parameter space for two-block stochastic blockmodels, the spectral embedding of the Laplacian is preferable to the spectral embedding of the adjacency matrix for subsequent inference. However, we observe that the metric in \cite{bickel_sarkar_2013} is intrinsically tied to the use of $K$-means as the clustering procedure: specifically, a smaller value of the metric for the Laplacian spectral embedding as compared to that for the adjacency spectral embedding only implies that clustering the Laplacian spectral embedding using $K$-means is possibly better than clustering the adjacency spectral embedding using $K$-means. 

Motivated by the above observation, in \cite{tang_lse}, we propose a metric that is {\em independent} of any specific clustering procedure, a metric that characterizes the minimum error achievable by {\em any} clustering procedure that uses only the spectral embedding for the recovery of block assignments in stochastic blockmodel graphs. For a given embedding method, the metric used in \cite{tang_lse} is based on the notion of statistical information between the limiting distributions of the blocks. Roughly speaking, smaller statistical information implies less information to discriminate between the blocks of the stochastic blockmodel. More specifically, the limit result in Section~\ref{subsec:Distributional} and Section~\ref{subsec:lse} state that, for stochastic blockmodel graphs, conditional on the block assignments the entries of the scaled eigenvectors corresponding to the few largest eigenvalues of the adjacency matrix and the normalized Laplacian matrix converge to a multivariate normal as the number of vertices increases. Furthermore, the associated covariance matrces are not spherical, so $K$-means clustering for the adjacency spectral embedding or Laplacian spectral embedding does not yield minimum error for recovering the block assignment. Nevertheless, these limiting results also facilitate comparison between the two embedding methods via the classical notion of Chernoff information \cite{chernoff_1952}. 

We now recall the notion of the Chernoff $\alpha$-divergences (for $\alpha \in (0,1))$ and Chernoff information. Let $F_0$ and $F_1$ be two absolutely continuous multivariate distributions in $\Omega = \mathbb{R}^{d}$ with density functions $f_0$ and $f_1$, respectively. Suppose that $Y_1, Y_2, \dots, Y_m$ are independent and identically distributed random variables, with $Y_i$ distributed either $F_0$ or $F_1$. We are interested in testing the simple null hypothesis $\mathbb{H}_0 \colon F = F_0$ against the simple alternative hypothesis $\mathbb{H}_1 \colon F = F_1$. A test $T$ can be viewed as a sequence of mappings $T_m \colon \Omega^{m} \mapsto \{0,1\}$ such that given $Y_1 = y_1, Y_2 = y_2, \dots, Y_m = y_m$, the test rejects $\mathbb{H}_0$ in favor of $\mathbb{H}_1$ if $T_m(y_1, y_2, \dots, y_m) = 1$; similarly, the test favors $\mathbb{H}_0$ if $T_m(y_1, y_2, \dots, y_m) = 0$. 

The Neyman-Pearson lemma states that, given $Y_1 = y_1, Y_2 = y_2, \dots, Y_m = y_m$ and a threshold $\eta_m \in \mathbb{R}$, the likelihood ratio test which rejects $\mathbb{H}_0$ in favor of $\mathbb{H}_1$ whenever
$$ \Bigl(\sum_{i=1}^{m} \log{f_0(y_i)} - \sum_{i=1}^{m} \log{f_1(y_i)} \Bigr) \leq \eta_m $$
is the most powerful test at significance level $\alpha_m = \alpha(\eta_m)$, so that the likelihood ratio test minimizes the Type II error $\beta_m$ subject to the constraint that the Type I error is at most $\alpha_m$. Assuming that $\pi \in (0,1)$ is a prior probability that $\mathbb{H}_0$ is true. Then, for a given $\alpha_m^{*} \in (0,1)$, let $\beta_m^{*} = \beta_m^{*}(\alpha_m^{*})$ be the Type II error associated with the likelihood ratio test when the Type I error is at most $\alpha_m^{*}$. The quantity $\inf_{\alpha_m^{*} \in (0,1)} \pi \alpha_m^{*} + (1 - \pi) \beta_m^{*}$ is then the Bayes risk in deciding between $\mathbb{H}_0$ and $\mathbb{H}_1$ given the $m$ independent random variables $Y_1, Y_2, \dots, Y_m$. A classical result of Chernoff \cite{chernoff_1952,chernoff_1956} states that the Bayes risk is intrinsically linked to a quantity known as the {\em Chernoff information}. More specifically, let $C(F_0, F_1)$ be the quantity
\begin{equation}
\label{eq:chernoff-defn}
\begin{split} C(F_0, F_1) & = - \log \, \Bigl[\, \inf_{t \in (0,1)} \int_{\mathbb{R}^{d}} f_0^{t}(\bm{x}) f_1^{1-t}(\bm{x}) \mathrm{d}\bm{x} \Bigr] \\
&= \sup_{t \in (0,1)} \Bigl[ - \log \int_{\mathbb{R}^{d}} f_0^{t}(\bm{x}) f_1^{1-t}(\bm{x}) \mathrm{d}\bm{x} \Bigr].
\end{split}
\end{equation}
Then we have
\begin{equation}
\label{eq:chernoff-binary}
\begin{split}
\lim_{m \rightarrow \infty} \frac{1}{m} \inf_{\alpha_m^{*} \in (0,1)} \log( \pi \alpha_m^{*} + (1 - \pi) \beta_m^{*}) & = - \, C(F_0, F_1).
\end{split}
\end{equation}
Thus $C(F_0, F_1)$, the Chernoff information between $F_0$ and $F_1$, is the {\em exponential} rate at which the Bayes error $\inf_{\alpha_m^{*} \in (0,1)} \pi \alpha_m^{*} + (1 - \pi) \beta_m^{*}$ decreases as $m \rightarrow \infty$; we note that the Chernoff information is independent of $\pi$. We also define, for a given $t \in (0,1)$ the Chernoff divergence $C_t(F_0, F_1) $ between $F_0$ and $F_1$ by
$$ C_{t}(F_0,F_1) = - \log \int_{\mathbb{R}^{d}} f_0^{t}(\bm{x}) f_1^{1-t}(\bm{x}) \mathrm{d}\bm{x}. $$
The Chernoff divergence is an example of a $f$-divergence as defined in \cite{Csizar,Ali-Shelvey}. When $t = 1/2$, $C_t(F_0,F_1)$ is the Bhattacharyya distance between $F_0$ and $F_1$. Recall that any $f$-divergence satisfies the Information Processing Lemma and is invariant with respect to invertible transformations \cite{Liese_Vadja}. Therefore, any $f$-divergence such as the Kullback-Liebler divergence can also be used to compare the two embedding methods; the Chernoff information is particularly attractive because of its explicit relationship with the Bayes risk.

The characterization of Chernoff information in Eq.~\eqref{eq:chernoff-binary} can be extended to $K + 1 \geq 2$ hypotheses. Let $F_0, F_1, \dots, F_{K}$ be distributions on $\mathbb{R}^{d}$ and suppose that $Y_1, Y_2, \dots, Y_m$ are independent and identically distributed random variables with $Y_i$ distributed $F \in \{F_0, F_1, \dots, F_K\}$. We are thus interested in determining the distribution of the $Y_i$ among the $K+1$ hypothesis $\mathbb{H}_0 \colon F = F_0, \dots, \mathbb{H}_{K} \colon F = F_K$. Suppose also that hypothesis $\mathbb{H}_k$ has {\em a priori} probability $\pi_k$. Then for any decision rule $\delta$, the risk of $\delta$ is $r(\delta) = \sum_{k} \pi_k \sum_{l \not = k} \alpha_{lk}(\delta) $ where $\alpha_{lk}(\delta)$ is the probability of accepting hypothesis $\mathbb{H}_l$ when hypothesis $\mathbb{H}_k$ is true. Then we have \cite{leang-johnson}
\begin{equation}
\label{eq:chernoff-multiple}
\inf_{\delta} \lim_{m \rightarrow \infty}  \frac{r(\delta)}{m} = - \min_{k \not = l} C(F_k, F_l).
\end{equation}
where the infimum is over all decision rules $\delta$. That is, for any $\delta$, $r(\delta)$ decreases to $0$ as $m \rightarrow \infty$ at a rate no faster than $\exp(- m \min_{k \not = l} C(F_k, F_l))$. It was also shown in \cite{leang-johnson} that the {\em Maximum A Posterior} decision rule achieves this rate. 

Finally, if $F_0 =  \mathcal{N}(\mu_0, \Sigma_0)$ and $F_1 = \mathcal{N}(\mu_1, \Sigma_1)$, then denoting by $\Sigma_t = t \Sigma_0 + (1 - t) \Sigma_1$, the Chernoff information $C(F_0, F_1)$ between $F_0$ and $F_1$ is given by
\begin{equation*}
C(F_0, F_1) = \sup_{t \in (0,1)} \Bigl(\frac{t(1 - t)}{2} (\mu_1 - \mu_2)^{\top}\Sigma_t^{-1}(\mu_1 - \mu_2) + \frac{1}{2} \log \frac{|\Sigma_t|}{|\Sigma_0|^{t} |\Sigma_1|^{1 - t}}  \Bigr).
\end{equation*}

Comparison of the performance of the Laplacian spectral embedding and the adjacency spectral embedding for recovering the block assignments now proceeds as follows. Let $\mathbf{B} \in [0,1]^{K \times K}$ and $\bm{\pi} \in \mathbb{R}^{K}$ be the matrix of block probabilities and the vector of block assignment probabilities for a $K$-block stochastic blockmodel. We shall assume that $\mathbf{B}$ is positive semidefinite. Then given an $n$ vertex instantiation of the SBM graph with parameters $(\bm{\pi}, \mathbf{B})$, for sufficiently large $n$, the large-sample optimal error rate for recovering the block assignments when adjacency spectral embedding is used as the initial embedding step can be characterized by the quantity $\rho_{\mathrm{A}} = \rho_{\mathrm{A}}(n)$  defined by
\begin{equation}
\label{eq:rho_ASE}
\rho_{\mathrm{A}} = \min_{k \not = l} \! \sup_{t \in (0,1)} \frac{1}{2} \log \frac{|\Sigma_{kl}(t)|}{|\Sigma_{k}|^{t} |\Sigma_{l}|^{1-t}} + \frac{nt(1 - t)}{2}(\nu_k - \nu_l)^{\top} \Sigma_{kl}^{-1}(t) (\nu_k - \nu_l)
\end{equation}
where $\Sigma_{kl}(t) = t \Sigma_{k} + (1 - t) \Sigma_l$; $\Sigma_k = \Sigma(\nu_k)$ and $\Sigma_l = \Sigma(\nu_l)$ are as defined in Theorem~\ref{thm:clt_orig_but_better}. We recall Eq.~\eqref{eq:chernoff-multiple}, in particular the fact that as $\rho_{\mathrm{A}}$ increases, the large-sample optimal error rate decreases. 
Similarly, the large-sample optimal error rate when Laplacian spectral embedding is used as the pre-processing step can be characterized by the quantity
$\rho_{\mathrm{L}} = \rho_{\mathrm{L}}(n)$ defined by
\begin{equation} 
\label{eq:rho_LSE}
\rho_{\mathrm{L}} = \min_{k \not = l} \sup_{t \in (0,1)} \frac{1}{2} \log \frac{|\tilde{\Sigma}_{kl}(t)|}{|\tilde{\Sigma}_{k}|^{t} |\tilde{\Sigma}_{l}|^{1-t}} + \frac{n t(1 - t)}{2} (\tilde{\nu}_k - \tilde{\nu}_l)^{\top} \tilde{\Sigma}_{kl}^{-1}(t) (\tilde{\nu}_k - \tilde{\nu}_l)
\end{equation}
where $\tilde{\Sigma}_{kl}(t) = t \tilde{\Sigma}_{k} + (1 - t) \tilde{\Sigma}_l$ with $\tilde{\Sigma}_k = \tilde{\Sigma}(\nu_k)$ and $\tilde{\Sigma}_l = \tilde{\Sigma}(\nu_l)$ as defined in Theorem~\ref{thm:clt-lse-sbm}, and
$\tilde{\nu}_k = \nu_k/(\sum_{k'} \pi_{k'} \nu_k^{\top} \nu_{k'})^{1/2}$. We emphasize that we have made the simplifying assumption that $n_k = n \pi_k$ in our expression for $\tilde{\nu}_k$ in Eq.~\eqref{eq:rho_LSE}. This is for ease of comparison between $\rho_{\mathrm{A}}$ and $\rho_{\mathrm{L}}$ in our subsequent discussion.

The ratio $\rho_{\mathrm{A}}/\rho_{\mathrm{L}}$ is a surrogate measure of the relative large-sample performance of the adjacency spectral embedding as compared to the Laplacian spectral embedding for subsequent inference, at least in the context of stochastic blockmodel graphs. That is to say, for given parameters $\bm{\pi}$ and $\mathbf{B}$, if $\rho_{\mathrm{A}}/\rho_{\mathrm{L}} > 1$, then adjacency spectral embedding is to be preferred over Laplacian spectral embedding when $n$, the number of vertices in the graph, is sufficiently large; similarly, if $\rho_{\mathrm{A}}/\rho_{\mathrm{L}} < 1$, then Laplacian spectral embedding is to be preferred over adjacency spectral embedding.

\begin{figure}[tp!]
\center
\includegraphics[width=0.7\textwidth]{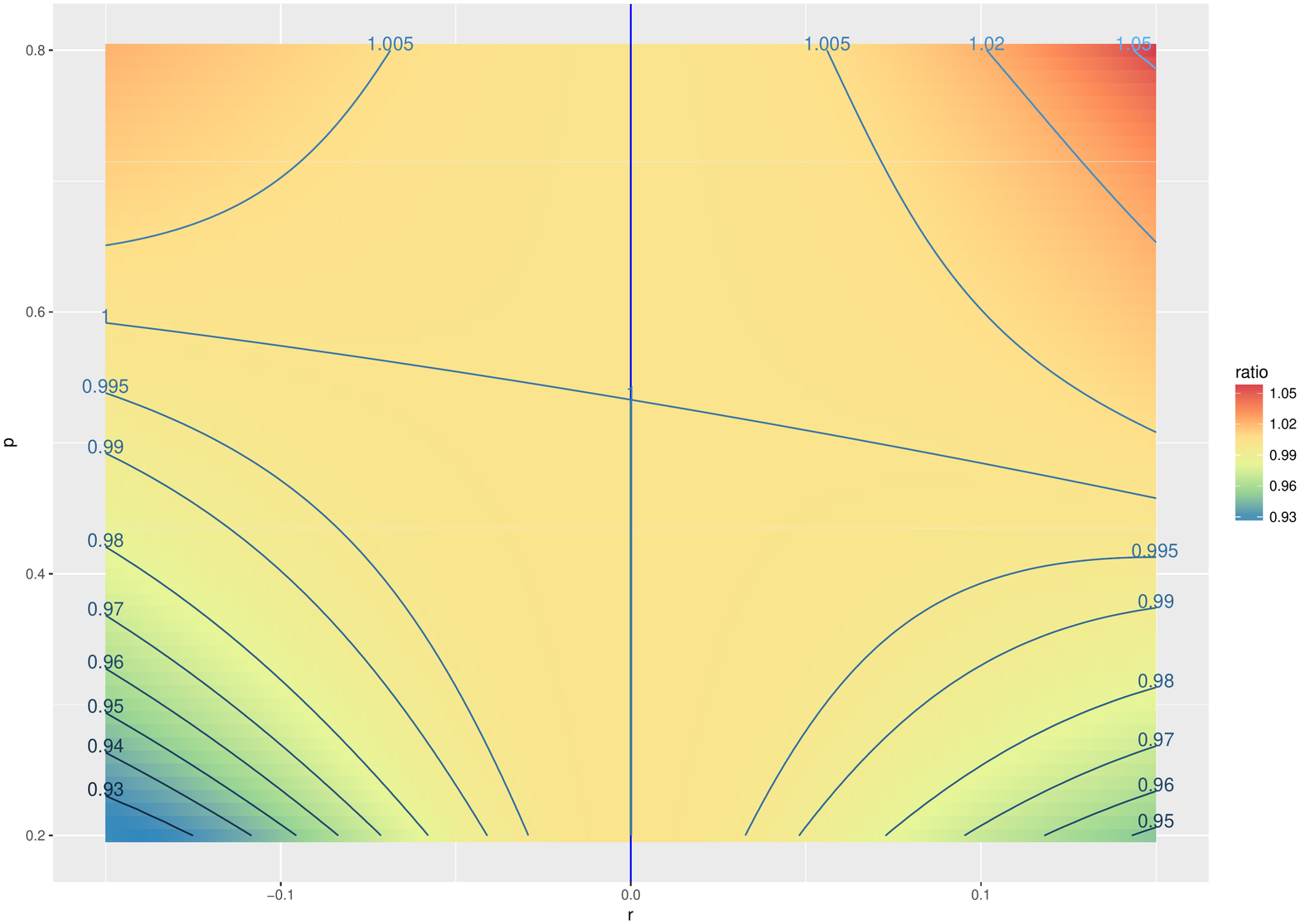}
\caption{The ratio $\rho_{\mathrm{A}}/\rho_{\mathrm{L}}$ displayed for various values of $p \in [0.2, 0.8]$ and $r = q - p \in [-0.15, 0.15]$. The labeled lines are the contour lines for $\rho_{\mathrm{A}}/\rho_{\mathrm{L}}$. Figure duplicated from \cite{tang_lse}.
}
\label{fig:ratio-plot}
\end{figure}

As an illustration of the ratio $\rho_{\mathrm{A}}/\rho_{\mathrm{L}}$, we first consider the collection of 2-block stochastic blockmodels where $\mathbf{B} = \Bigl[ \begin{smallmatrix} p^2 & pq \\ pq & q^2 \end{smallmatrix} \Bigr]$ for $p, q \in (0,1)$ and $\bm{\pi} = (\pi_1, \pi_2 )$ with $\pi_1 + \pi_2 = 1$. We note that these $\mathbf{B}$ also have rank $1$ and thus the Chernoff information can be computed explicitly. 
Then for sufficiently large $n$, $\rho_{\mathrm{A}}$ is approximately
$$ \rho_{\mathrm{A}} \approx \sup_{t \in (0,1)} \frac{nt(1 - t)}{2} (p - q)^{2} (t \sigma_1^{2} + (1 - t) \sigma_2^{2})^{-1}$$
where $\sigma_1$ and $\sigma_2$ are as specified in Eq.~\eqref{eq:er-p-q-ase1} and Eq.~\eqref{eq:er-p-q-ase2}, respectively. Simple calculations yield
$$ \rho_{\mathrm{A}} \approx \frac{n(p - q)^2 (\pi_1 p^2 + \pi_2 q^2)^2}{2\bigl(\sqrt{\pi_1 p^4 (1 - p^2) + \pi_2 p q^3(1 - pq) } + \sqrt{\pi_1 p^3 q(1 - pq) + \pi_2 q^4 (1 - q^2)}\bigr)^2}$$ for sufficiently large $n$. Similarly, denoting by $\tilde{\sigma}_1^{2}$ and $\tilde{\sigma}_2^2$ the variances specified in Eq.~\eqref{eq:er-p-q-lse1} and Eq.~\eqref{eq:er-p-q-lse2}, we have
\begin{equation*}
\begin{split}
 \rho_{\mathrm{L}} & \approx \sup_{t \in (0,1)} \frac{nt(1 - t)}{2} \Bigl(\frac{p}{\sqrt{\pi_1 p^2  + \pi_2 pq}} - \frac{q}{\sqrt{\pi_1 p q + \pi_2 q^2}}\Bigr)^{2} (t \tilde{\sigma}_1^{2} + (1 - t) \tilde{\sigma}_2^2)^{-1} \\
 & \approx \frac{2n(\sqrt{p} - \sqrt{q})^2 (\pi_1 p + \pi_2 q)^2}{\bigl(\sqrt{\pi_1 p (1 - p^2) + \pi_2 q (1 - pq)} + \sqrt{\pi_1 p (1 - pq) + \pi_2 q (1 - q^2)}\bigr)^2}  \\ &
 \approx \frac{2n(p - q)^2 (\pi_1 p + \pi_2 q)^2}{(\sqrt{p} + \sqrt{q})^2 \bigl(\sqrt{\pi_1 p (1 - p^2) + \pi_2 q (1 - pq)} + \sqrt{\pi_1 p (1 - pq) + \pi_2 q (1 - q^2)}\bigr)^2} 
 \end{split}
\end{equation*}
for sufficiently large $n$. 
Fixing $\bm{\pi} = (0.6, 0.4)$,
we compute the ratio $\rho_{\mathrm{A}}/\rho_{\mathrm{L}}$ for a range of $p$ and $q$ values, with $p \in [0.2, 0.8]$ and $q = p + r$ where $r \in [-0.15, 0.15]$. The results are plotted in Figure~\ref{fig:ratio-plot}. The $y$-axis of Figure~\ref{fig:ratio-plot} denotes the values of $p$ and the $x$ axis are the values of $r$. We see from the above figure that, in general, neither of the methods, namely adjacency spectral clustering or normalized Laplacian spectral embedding followed by clustering via Gaussian mixture models, dominates over the whole $(p,r)$ parameter space. However, in general, one can easily show that LSE is preferable over ASE whenever the block probability matrix $\mathbf{B}$ becomes sufficiently sparse. 
Determination of similarly intuitive conditions for which ASE dominates over LSE is considerably more subtle and is the topic of current research. But in general, we observe that ASE dominates over LSE whenever the entries of $\mathbf{B}$ are relatively large. 


Finally we consider the collection of stochastic blockmodels with parameters $\bm{\pi}$ and $\mathbf{B}$ where 
\begin{equation}
\label{eq:3block-example}
 \mathbf{B} = \begin{bmatrix} p & q & q \\ q & p & q \\ q & q & p \\ \end{bmatrix}, \quad p, q \in (0,1), \,\, \text{and} \,\, \bm{\pi} = (0.8, 0.1, 0.1).
\end{equation}
First we compute the ratio $\rho_{\mathrm{A}}/\rho_{\mathrm{L}}$ for $p \in [0.3, 0.9]$ and $r = q - p$ with $r \in [- 0.2, -0.01]$. The results are plotted in Figure~\ref{fig:ratio_3blocks}, with the $y$-axis of Figure~\ref{fig:ratio_3blocks} being the values of $p$ and the $x$-axis being the values of $r$. Once again we see that for the purpose of subsequent inference, neither embedding methods dominates over the whole parameter space and that LSE is still preferable to ASE for smaller values of $p$ and $q$ and that ASE is preferable to LSE for larger values of $p$ and $q$. 

\begin{figure}[htbp]
\center 
\includegraphics[width=0.8\textwidth]{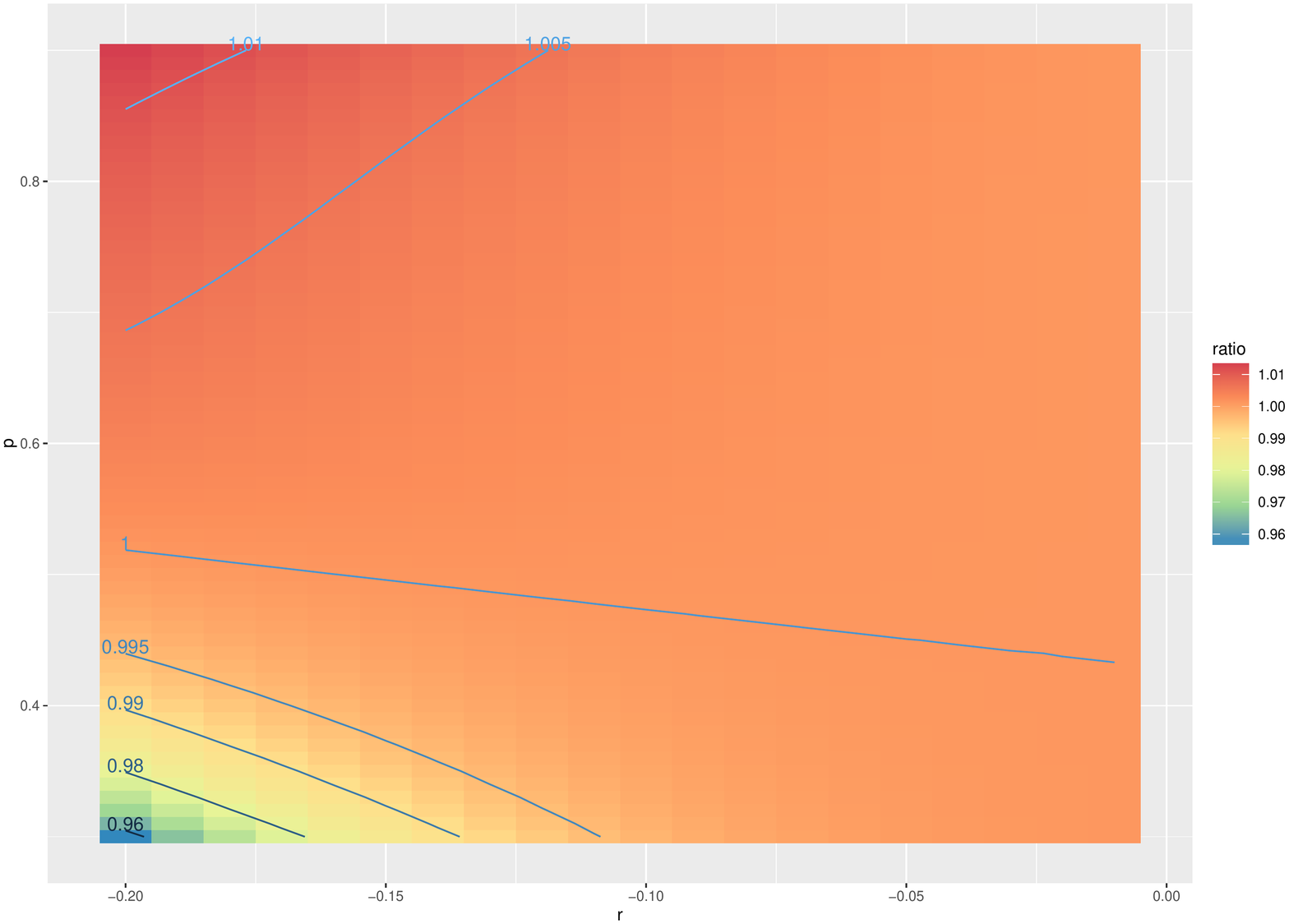}
\caption{The ratio $\rho_{A}/\rho_{L}$ displayed for various values of $p \in [0.2, 0.8]$ and $r = q - p \in [-0.2, -0.01]$ for the 3-block stochastic blockmodel of Eq.~\eqref{eq:3block-example}. The labeled lines are the contour lines for $\rho_{\mathrm{A}}/\rho_{\mathrm{L}}$. Figure duplicated from \cite{tang_lse}.
 }
\label{fig:ratio_3blocks}
\end{figure}



\subsection{Hypothesis testing}\label{subsec:Testing}
The field  of multi-sample graph inference is comparatively new, and the development of a comprehensive machinery for two-sample hypothesis
testing for random graphs is of both theoretical and practical
importance.  The test procedures in \cite{tang14:_semipar} and \cite{tang14:_nonpar}, both of which leverage the adjacency spectral embedding to test hypotheses of equality or equality in distribution for random dot product graphs, are among the only principled methodologies currently available. In both cases, the accuracy of the adjacency spectral embedding as an estimate of the latent positions is the key to constructing a test statistic.  Specifically, \cite{tang14:_semipar} gives a new and improved bond, Theorem~\ref{thm:conc_Xhat_X} below, for the Frobenius room of the difference between the
original latent positions and the estimated latent positions obtained from the embedding. This bound is then used to establish a valid and consistent test for the semiparameteric hypothesis test of equality for latent positions in a pair of vertex-aligned random dot product graphs. In the nonparametric case, \cite{tang14:_nonpar} demonstrates how the adjacency spectral embedding can be integrated with a kernel density estimator to accurately estimate the underlying distribution $F$ for in a random dot product graph with i.i.d latent positions.


To begin, we consider the problem of developing a test for the hypothesis that two random dot product
graphs on the same vertex set, with known vertex correspondence, have
the same generating latent position or have generating latent
positions that are scaled or diagonal transformations of one another.
This framework includes, as a special case, a test for whether two
stochastic blockmodels have the same or related block probability
matrices. In this two-sample testing problem, though,
the parameter dimension grows as the sample size grows. Therefore, the
problem is not precisely analogous to classical two-sample tests for,
say, the difference of two parameters belonging to some fixed
Euclidean space, in which an increase in data has no effect on the
dimension of the parameter. The problem is also not 
nonparametric, since we view our latent positions as fixed
and impose specific distributional requirements on the
data---that is, on the adjacency matrices. Indeed, we regard the
problem as semiparametric, and \cite{tang14:_semipar} adapts the traditional definition of
consistency to this setting. In particular, for the test procedure we describe, power will increase 
to one for alternatives in which the difference between the two latent
positions grows with the sample size.

Our test procedure is, at first glance, deceptively simple: given a pair of adjacency matrices $\bA$ and $\bB$ for two $d$-dimensional random dot product graphs, we generate their adjacency spectral embeddings, denoted $\Xhat$ and $\Yhat$, respectively, and compute an appropriately normalized version of the so-called {\em Procrustes fit} or {\em Procrustes distance} between the two embeddings:
$$\min_{\bW \in \mathcal{O}^{d \times d}}\|\Xhat-\Yhat \bW\|_F$$
(Recall that such a fit is necessary because of the inherent nonidentifiability of the random dot product model.)	

Understanding the limiting distribution of this test statistic is more complicated, however, and appropriately framing the set of null and alternative hypothesis for which the test is valid and consistent (i.e. a level $\alpha$-test with power converging to 1 as $n \rightarrow \infty$) is delicate.  To that end, we first state the key concentration inequality for $\min_{\bW \in \mathcal{O}^{d \times d}}\|\Xhat-\Yhat \bW\|_F$. 
	
\begin{theorem}\label{thm:conc_Xhat_X}
	Suppose $\bP=\bX \bX^{\top}$ is an $n \times n$ probability
	matrix of rank $d$.  Suppose also
	that there exists $\epsilon>0$ such that
	$\delta(\bP)>(\log{n})^{2 + \epsilon}$. Let $c>0$ be
	arbitrary but fixed. Then there exists a $n_0(c)$ and a universal
	constant $C \geq 0$ such that if $n \geq n_0$ and $n^{-c} <
	\eta<1/2$, then there exists a deterministic $\bW \in
	\mathcal{O}(d)$ such that, with probability at least $1 - 3\eta$,
	\begin{equation}\label{conc_x-xhat}
	\Bigl| \|\Xhat-\bX\bW\|_F-C(\bX) \Bigr| \leq \frac{C d
		\log{(n/\eta)}}{C(\bX) \sqrt{ \gamma^{5}(\bP) \delta(\bP)}}
	\end{equation}
	where $C(\bX)$ is a function of $\bX$ given by
	\begin{equation}\label{eq:def_of_C_semipar}
	C(\bX) = \sqrt{\mathrm{tr} \,\, 
		\SP^{-1/2} \UP^{\top} \E[(\bA -
		\bP)^{2}] \UP
		\SP^{-1/2}}
	\end{equation}
	where $\mathbb{E}[(\bA - \bP)^2]$ is taken with respect to $\bA$ and conditional on $\bX$.
\end{theorem}	
We note that the proof of this theorem consists of two pieces: it is straightforward to show that the Frobenius norm bound of Lemma \ref{thm:minh_frob} implies that
$$\|\Xhat - \bX \bW \|_{F} = \|(\bA - \bP) \UP  \SP^{-1/2} \|_{F}
+ O(d \log (n) \delta^{-1/2}(\bP) \gamma^{-5/2}(\bP))$$
To complete the theorem, then, \cite{tang14:_semipar} demonstrates a concentration inequality for
$\|(\mathbf{A} - \mathbf{P}) \mathbf{U}_{\mathbf{P}}
\mathbf{S}_{\mathbf{P}}^{-1/2} \|^{2}_{F}$, showing that
\begin{equation}
\label{eq:log_sobolev_conc_inequality}
\bigl|\|(\mathbf{A} - \mathbf{P}) \mathbf{U}_{\mathbf{P}}
\mathbf{S}_{\mathbf{P}}^{-1/2} \|^{2}_{F} - C^{2}(\mathbf{X}) \bigr| \leq
\frac{14 \sqrt{2d} \log{(n/\eta)}}{\gamma(\mathbf{P}) \sqrt{\delta(\mathbf{P})}}.
\end{equation} 
where $C(\bX)$ is as defined in \eqref{eq:def_of_C_semipar}. We do not go into the details of this concentration inequality here, but rather point the reader to \cite{tang14:_semipar}.  We observe, however, that this inequality has immediate consequences for two-sample testing for random dot product graphs. For two random dot product graphs with probability matrices $\bP=\bX \bX^{\top}$ and $\bQ=\bY \bY^{\top}$, consider the null hypothesis $\bX=\bY\bW$ for some orthogonal $\bW$. It can be shown that $\min_{\bW \in \mathcal{O}^{d \times d}}\|\Xhat-\Yhat\bW\|_F$ is the basis for a valid and consistent test.  We emphasize, though, that as the graph size $n$ increases, the $n \times d$ matrix of latent positions also increases in size. As a consequence of this, we consider the following notion of consistency in this semiparametric setting. As an aside on the notation, in this section, we consider a sequence of graphs with latent positions, all indexed by $n$; thus, as we noted in our preliminary remarks on notation, $\bX_n$ and $\Xhat_n$ refer to the {\em matrices} of true and estimated latent positions in this sequence. 
\begin{definition}\label{consistency}
	Let $(\mathbf{X}_n, \mathbf{Y}_n)_{n \in \mathbb{N}}$, be a given sequence of latent positions, where $\mathbf{X}_n$ and $\mathbf{Y}_n$ are both in $\mathbb{R}^{n \times d}$. A test statistic $T_n$ and associated
	rejection region $R_n$ to test the null hypothesis
	\begin{align*} 
	H^n_0: \,  {\bf X}_n =_{W} {\bf Y}_n \quad
	\textrm{ against } \quad H^n_a: \, {\bf X}_n \not =_{W} {\bf Y}_n 
	\end{align*}
	is a {\em consistent, asymptotically level $\alpha$ test} 
	if for any $\eta>0$, there exists $n_0 = n_0(\eta)$ such that 
	\begin{enumerate}[(i)]
		\item If $n>n_0$ and $H_a^n$ is true, then $P(T_n \in R_n)>1-\eta$
		\item If $n > n_0$ and $H_0^n$ is true, then $P(T_n \in R_n) \leq \alpha + \eta$
	\end{enumerate}
\end{definition}
With this definition of consistency, we obtain the following theorem on two-sample testing for random dot products on the same vertex set and with known vertex correspondence.
\begin{theorem}
	\label{thm:identity}
	For each fixed $n$, consider the
	hypothesis test
	\begin{equation*}
	H^{n}_0: {\bf X}_n =_{W} {\bf Y}_n \quad \textrm{ versus } \quad
	H^{n}_a: {\bf X}_n \not =_{W} {\bf Y}_n
	\end{equation*}
	where ${\bf X}_n$ and ${\bf Y}_n$ $\in \mathbb{R}^{n \times d}$ are
	matrices of latent positions for two random dot product graphs. Let
	$\hat{{\bf X}}_n$ and $\hat{{\bf Y}}_n$ be the adjacency spectral
	embeddings of ${\bf A}_n\sim \mathrm{Bernoulli}({\bf X}_n{\bf
		X}_n^{\top})$ and ${\bf B}_n \sim \mathrm{Bernoulli}({\bf Y}_n{\bf
		Y}_n^{\top})$, respectively. Define the test statistic $T_n$ as follows:
	\begin{equation}
	\label{eq:semipar_TS_def}
	T_n=\frac{\min\limits_{{\bf W} \in \mathcal{O}(d)} 
		\|\hat{{\bf X}}_n{\bf W}-\hat{{\bf Y}}_n\|_F}
	{\sqrt{d\gamma^{-1}(\mathbf{A}_n)}+ \sqrt{d \gamma^{-1}(\mathbf{B}_n)}}.
	\end{equation}
	Let $\alpha \in (0,1)$ be given. Then for all $C > 1$, if the
	rejection region is 
	$R:=\left\{t \in \mathbb{R}: t\geq C \right\}$, 
	then there exists an
	$n_1 = n_1(\alpha, C) \in \mathbb{N}$ such that for all $n \geq n_1$, the
	test procedure with $T_n$ and rejection region $R$ is an at most level
	$\alpha$ test, i.e., for all $n \geq n_1$, if $\mathbf{X}_n
	=_{W} \mathbf{Y}_n$, then 
	$ \mathbb{P}(T_n \in R) \leq \alpha.$
	Furthermore,
	consider the sequence of latent positions $\{{\bf X}_n\}$ and 
	$\{{\bf Y}_n\}$, $n \in \mathbb{N}$,
	satisfying the eigengap assumptions in Assumption~\ref{ass:max_degree_assump} and denote by $d_n$ the
	quantity $
	d_n := \min\limits_{{\bf W} \in \mathcal{O}(d)} \| {\bf X}_n{\bf W}-{\bf
		Y}_n \|$. 
	Suppose $d_n \neq 0$ for infinitely many $n$.  Let $t_1=\min\{k>0:
	d_k>0\}$ and sequentially define $t_n=\min\{k>t_{n-1}: d_k>0\}$.
	Let $b_n=d_{t_n}$.  If $\liminf b_n = \infty$, then this test
	procedure is consistent in the sense of Definition~\ref{consistency}
	over this sequence of latent positions.
\end{theorem}
\begin{remark} This result does not require that ${\bf A}_n$ and ${\bf B}_n$
	be independent for any fixed $n$, nor that the sequence of pairs
	$({\bf A}_n, {\bf B}_n)$, $n \in \mathbb{N}$, be independent.
	We note that Theorem \ref{thm:identity} is written to emphasize consistency in the sense of Definition \ref{consistency}, even in a case when, for example, the latent position sequence is such that $\mathbf{X}_n=_W\mathbf{Y}_n$ for all even $n$, but $\mathbf{X}_n$ and $\mathbf{Y}_n$ are sufficiently far apart for odd $n$. 
	In
	addition, the requirement that $\liminf b_k=\infty$ can be weakened
	somewhat.  Specifically, consistency is achieved as long
	as $$\liminf_{n \rightarrow \infty} \Bigl(\| \mathbf{X}_n\mathbf{W}
	-\mathbf{Y}_n \|_{F} - C(\mathbf{X}_n) - C(\mathbf{Y}_n)\Bigr) > 0.$$
\end{remark}
It also possible to construct analogous tests for latent positions related by scaling factors, or, in the case of the degree-corrected stochastic block model, by projection.  We summarize these below, beginning with the case of scaling.

For the scaling case, let $\mathcal{C}=\mathcal{C}(\mathbf{Y}_n)$
denote the class of all positive constants $c$ for which all the
entries of $c^2 \mathbf{Y}_n \mathbf{Y}_n^{\top}$ belong to the unit
interval. We wish to test the null hypothesis
$H_0 \colon \mathbf{X}_n =_{W} c_n \mathbf{Y}_n$ for some
$c_n\in \mathcal{C}$ against the alternative
$H_a \colon \mathbf{X}_n \not =_{W} c_n\mathbf{Y}_n$ for any
$c_n \in \mathcal{C}$. In what follows below, we will only write
$c_n>0$, but will always assume that $c_n \in \mathcal{C}$, since the
problem is ill-posed otherwise. The test statistic $T_n$ is now a
simple modification of the one used in Theorem~\ref{thm:identity}: for
this test, we compute a Procrustes distance between scaled adjacency
spectral embeddings for the two graphs. 
\begin{theorem}
	\label{thm:2}
	For each fixed $n$, consider the
	hypothesis test
	\begin{align*}
	H^{n}_0  \colon {\bf X}_n =_{W} c_n{\bf Y}_n \quad \text{for some $c_n > 0$}
	\textrm{ versus } \,\,  
	H^{n}_a  \colon {\bf X}_n \not =_{W} c_n{\bf Y}_n \quad \text{for all $c_n > 0$}
	\end{align*}
	where ${\bf X}_n$ and ${\bf Y}_n$ $\in \mathbb{R}^{n \times d}$ are
	latent positions for two random dot product graphs with adjacency
	matrices $\mathbf{A}_n$ and $\mathbf{B}_n$, respectively. 
	Define the test statistic $T_n$ as follows:
	\begin{equation}
	\label{eq:8}
	T_n=\frac{\min\limits_{{\bf W} \in \mathcal{O}(d)} 
		\|\hat{{\bf X}}_n{\bf W}/\|\Xhat_n\|_{F} - \hat{{\bf
				Y}}_n/\|\hat{\mathbf{Y}}_n\|_{F} \|_{F}}
	{2 \sqrt{d \gamma^{-1}(\mathbf{A}_n)}/\|\Xhat_n\|_{F}+ 2\sqrt{d
			\gamma^{-1}(\mathbf{B}_n)}/\|\hat{\mathbf{Y}}_n \|_{F}}.
	\end{equation}
	Let $\alpha \in (0,1)$ be given. Then for all $C
	> 1$, if the rejection region is $R:=\left\{t \in \mathbb{R}: t\geq C \right\}$, 
	then there exists an $n_1 = n_1(\alpha, C) \in \mathbb{N}$ such that
	for all $n \geq n_1$, the test procedure with $T_n$ and rejection
	region $R$ is an at most level $\alpha$ test.
	Furthermore, consider the sequence of latent position $\{{\bf X}_n\}$ and 
	$\{{\bf Y}_n\}$, $n \in \mathbb{N}$,
	satisfying Assumption~\ref{ass:max_degree_assump} and denote by $d_n$ the quantity
	\begin{equation}\label{eq:scaling_alternative}
	d_n := \frac{ \min\limits_{{\bf W} \in \mathcal{O}(d)}\| {\bf X}_n{\bf W}/\|{\bf
			X}_n\|_{F} - {\bf Y}_n/\|{\bf Y}_n\|_{F}
		\|_{F}}{1/\|\mathbf{X}_n\|_{F} +
		1/\|\mathbf{Y}_n\|_{F}} =   \frac{ \min\limits_{{\bf W} \in \mathcal{O}(d)}
		\| {\bf X}_n \|\mathbf{Y}_n \|_{F} {\bf W}
		- {\bf Y}_n \|\mathbf{X}_n \|_{F}
		\|_{F}}{\|\mathbf{X}_n\|_{F} +
		\|\mathbf{Y}_n\|_{F}}
	\end{equation}
	Suppose $d_n \neq 0$ for infinitely many $n$.  Let $t_1=\min\{k>0:
	d_k>0\}$ and sequentially define $t_n=\min\{k>t_{n-1}: d_k>0\}$.  Let
	$b_n=d_{t_n}$.  If $\liminf b_n = \infty$, then this test procedure is
	consistent in the sense of Definition~\ref{consistency} over this
	sequence of latent positions.
\end{theorem}

We next consider the case of testing whether the latent positions are
related by a diagonal transformation.  
i.e., whether $H_0
\colon \mathbf{X}_n =_{W} \mathbf{D}_n \mathbf{Y}_n$ for some
diagonal matrix $\mathbf{D}_n$. We proceed analogously to the
scaling case, above, by defining the class
$\mathcal{E}=\mathcal{E}(\mathbf{Y}_n)$ to be all positive diagonal
matrices $\mathbf{D}_n \in \mathbb{R}^{n \times n}$ such that
$\mathbf{D}_n \mathbf{Y}_n \mathbf{Y}_n^{\top} \mathbf{D}_n$ has all
entries in the unit interval.

As before, we will
always assume that $\mathbf{D}_n$ belongs to $\mathcal{E}$, even if
this assumption is not explicitly stated.  The test statistic $T_n$ in
this case is again a simple modification of the one used in
Theorem~\ref{thm:identity}. However, for technical reasons, our proof
of consistency requires an additional condition on the minimum
Euclidean norm of each row of the matrices $\mathbf{X}_n$ and
$\mathbf{Y}_n$. To avoid certain technical issues, we impose a
slightly stronger density assumption on our graphs for this test.
These assumptions can be weakened, but at the cost of interpretability.
The assumptions we make on the latent positions, which we summarize
here, are moderate restrictions on the sparsity of the graphs.
\begin{assumption}
	\label{eigengap_assump_diagonal}
	We assume that there exists $d \in \mathbb{N}$ such that for all
	$n$, $\mathbf{P}_n$ is of rank $d$.  Further, we assume that there
	exist constants $\epsilon_1>0$, $\epsilon_2 \in (0,1)$, $c_0>0$ and
	$n_0(\epsilon_1, \epsilon_2, c) \in \mathbb{N}$ such that for all $n \geq n_0$:
	\begin{align}
	\gamma(\mathbf{P}_n) \geq c_0; \qquad 
	\delta(\mathbf{P}_n) \geq (\log{n})^{2 + \epsilon_1}; \qquad
	\min_{i} \|X_i\| >
	\left(\frac{\log{n}}{\sqrt{\delta(\mathbf{P}_n)}}\right)^{1 -
		\epsilon_2}
	\end{align}
\end{assumption}
We then have the following result.
\begin{theorem}
	\label{thm:1}
	For each fixed $n$, consider the
	hypothesis test
	\begin{align*}
	H^{n}_0  \colon {\bf X}_n =_{W} \mathbf{D}_n {\bf Y}_n \quad
	\text{for some $\mathbf{D}_n \in \mathcal{E}$}\,\, 
	\textrm{ versus }  
	\,\, H^{n}_a  \colon {\bf X}_n \not =_{W} \mathbf{D}_n{\bf Y}_n \quad \text{for
		any $\mathbf{D}_n \in \mathcal{E}$}
	\end{align*}
	where ${\bf X}_n$ and ${\bf Y}_n$ $\in \mathbb{R}^{n \times d}$ are
	matrices of latent positions for two random dot product graphs. 
	For any matrix $\mathbf{Z} \in \mathbb{R}^{n \times d}$, let
	$\mathcal{D}(\mathbf{Z})$ be the diagonal matrix whose diagonal
	entries are the Euclidean norm of the rows of $\mathbf{Z}$ and let
	$\mathcal{P}(\mathbf{Z})$ be the matrix whose rows are the projection of the rows of
	$\mathbf{Z}$ onto the unit sphere. 
	We define the test statistic as follows:
	\begin{equation}
	\label{eq:15}
	T_n=\frac{\min\limits_{{\bf W} \in \mathcal{O}(d)} 
		\|\mathcal{P}(\hat{\mathbf{X}}_n) {\bf W} -
		\mathcal{P}(\hat{\mathbf{Y}}_n) \|_{F}}
	{2 \sqrt{d \gamma^{-1}(\mathbf{A})} \|\mathcal{D}^{-1}(\Xhat_n)\|+
		2 \sqrt{d \gamma^{-1}(\mathbf{B}_n)}\|\mathcal{D}^{-1}(\hat{\mathbf{Y}}_n)\|}.
	\end{equation}
	where we write
	$\mathcal{D}^{-1}(\mathbf{Z})$ for
	$(\mathcal{D}(\mathbf{Z}))^{-1}$. Note that
	$\|\mathcal{D}^{-1}(\mathbf{Z})\| = 1/(\min_{i} \|Z_i\|)$. 
	
	Let $\alpha \in (0,1)$ be given. Then for all $C
	> 1$, if the rejection region is 
	$R:=\left\{t \in \mathbb{R}: t\geq C \right\},$
	then there exists an $n_1 = n_1(\alpha, C) \in \mathbb{N}$ such that
	for all $n \geq n_1$, the test procedure with $T_n$ and rejection
	region $R$ is an at most level-$\alpha$ test.
	Furthermore, consider the sequence of latent position $\{{\bf X}_n\}$ and $\{{\bf Y}_n\}$, $n \in \mathbb{N}$,
	satisfying 
	Assumption~\ref{eigengap_assump_diagonal} and denote by $d_n$ the quantity
	\begin{equation}
	\label{eq:12}
	d_n := 
	\frac{ \min\limits_{{\bf W} \in \mathcal{O}(d)}\| \mathcal{P}({\bf
			X}_n) {\bf W} -
		\mathcal{P}({\bf Y}_n)
		\|_{F}}{\|\mathcal{D}^{-1}(\mathbf{X})\| +
		\|\mathcal{D}^{-1}(\mathbf{Y})\|} = D_{\mathcal{P}}(\mathbf{X}_n, \mathbf{Y}_n) 
	\end{equation}
	Suppose $d_n \neq 0$ for infinitely many $n$.  Let $t_1=\min\{k>0:
	d_k>0\}$ and sequentially define $t_n=\min\{k>t_{n-1}: d_k>0\}$.  Let
	$b_n=d_{t_n}$. If $\liminf b_n = \infty$, then this test procedure is
	consistent in the sense of Definition~\ref{consistency} over this
	sequence of latent positions.
\end{theorem}
This collection of semiparametric tests has numerous applications in graph comparison; in Section \ref{sec:Applications} , we describe its use in connectomics and brain scan data. We stress, though, that the Procrustes transformations are rather cumbersome, and they limit our ability to generalize these procedures to graph comparisons involving more than two graphs.  As a consequence, it can be useful to consider {\em joint} or {\em omnibus} embeddings, in which adjacency matrices for multiple graphs on the same vertex set are jontly embedded into a single (larger-dimensional) space, but {\em with a distinct representation for each graph}.  For an illuminating joint graph inference study on the {\em C. elegans} connectome that addresses somewhat different questions from semiparametric testing, see \cite{chen_worm}. 
Simultaneously embedding multiple graphs into a shared space
allows comparison of graphs without the need to
perform pairwise alignments of graph embeddings.
Further, a distinct representation of each graph
renders the omnibus embedding especially useful for
subsequent comparative graph inference.  

\subsection{Omnibus embedding}\label{subsec:omnibus}
In \cite{levin_omni_2017}, we show that an omnibus embedding---that is, an embedding of multiple graphs into a single shared space---can yield consistent estimates of underlying latent positions. Moreover, like the adjacency spectral embedding for a single graph, the rows of this omnibus embedding, suitably-scaled, are asymptotically normally distributed.  As might be anticipated, the use of multiple independent graphs generated from the same latent positions, as opposed just a single graph, yields a reduction in variance for the estimated latent positions, and since the omnibus embedding provides a distinct representation for each graph, subsequently averaging these estimates reduces the variance further still. Finally, the omnibus embedding allows us to compare graphs without cumbersome Procrustes alignments.

To construct the omnibus embedding, we consider a collection of $m$ random dot product graphs,
all with the same the same generating latent positions.
This motivates the following definition:
\begin{definition}
	[Joint Random Dot Product Graph]
	\label{def:JRDPG}
	Let $F$ be a $d$-dimensional inner product distribution on $\R^d$.
	We say that random graphs $\bA^{(1)},\bA^{(2)},\dots,\bA^{(m)}$
	are distributed as a \emph{joint random dot product graph (JRDPG)}
	and write $(\bA^{(1)},\bA^{(2)},\dots,\bA^{(m)},\bX) \sim \JRDPG(F,n,m)$
	if $\bX = [\bX_1, \bX_2,\dots,\bX_n]^{\top} \in \R^{n \times d}$ has its (transposed)
	rows distributed i.i.d. as $\bX_i \sim F$, and we have 
	marginal distributions $(\bA^{(k)},\bX) \sim \RDPG(F,n)$
	for each $k=1,2,\dots,m$.
	That is, the $\bA^{(k)}$ are conditionally independent
	given $\bX$, with edges independently distributed as
	$\bA^{(k)}_{i,j} \sim \Bern( (\bX\bX^{\top})_{ij} )$ for all $1 \le i < j \le n$
	and all $k \in [m]$.
\end{definition}
Given a set of $m$ adjacency matrices distributed as
$$(\bA^{(1)},\bA^{(2)},\dots,\bA^{(m)},\bX) \sim \JRDPG(F,n,m)$$
for distribution $F$ on $\R^d$,
a natural inference task is to recover the $n$ latent positions
$\bX_1,\bX_2,\dots,\bX_n \in \R^d$ shared by the vertices of the $m$ graphs.
To estimate the underlying latent postions from these $m$ graphs, \cite{runze_law_large_graphs} provides justification for the estimate
$\Xbar = \ASE( \Abar, d )$, where $\Abar$ is the sample mean of the
adjacency matrices $\bA^{(1)},\bA^{(2)},\dots,\bA^{(m)}$.
However, $\Xbar$ is ill-suited to any task that requires
comparing latent positions across the $m$ graphs,
since the $\Xbar$ estimate collapses the $m$ graphs into a single
set of $n$ latent positions.
This motivates the \emph{omnibus embedding},
which still yields a single spectral decomposition, but with a separate $d$-dimensional representation for each of the $m$ graphs.
This makes the omnibus embedding useful for {\em simultaneous}
inference across all $m$ observed graphs.
\begin{definition}[Omnibus embedding] \label{def:omni_def}
	Let $\bA^{(1)},\bA^{(2)},\dots,\bA^{(m)} \in \R^{n \times n}$
	be (possibly weighted) adjacency matrices
	of a collection of $m$ undirected graphs.
	We define the $mn$-by-$mn$ omnibus matrix
	of $\bA^{(1)}, \bA^{(2)}, \dots, \bA^{(m)}$ by
	\begin{equation} \label{eq:omnidef}
	\bM =
	\begin{bmatrix}
	\bA^{(1)} & \frac{1}{2}(\bA^{(1)} + \bA^{(2)}) & \dots & \frac{1}{2}(\bA^{(1)} + \bA^{(m)}) \\
	\frac{1}{2}(\bA^{(2)} + \bA^{(1)}_2) & \bA^{(2)} & \dots & \frac{1}{2}(\bA^{(2)} + \bA^{(m)}) \\
	\vdots & \vdots & \ddots & \vdots \\
	\frac{1}{2}(\bA^{(m)} + \bA^{(1)}) & \frac{1}{2}(\bA^{(m)} + \bA^{(2)})
	& \dots & \bA^{(m)} \end{bmatrix},
	\end{equation}
	and the $d$-dimensional \emph{omnibus embedding} of
	$\bA^{(1)},\bA^{(2)},\dots,\bA^{(m)}$
	is the adjacency spectral embedding of $\bM$:
	$$ \OMNI(\bA^{(1)},\bA^{(2)},\dots,\bA^{(m)},d)= \ASE( \bM, d ). $$
\end{definition}
where $ASE$ is the $d$-dimensional adjacency spectral embedding.
Under the JRDPG, the omnibus matrix has expected value
$$ \E \bM = \Ptilde = \bJ_m \otimes \bP = \UPt \SPt \UPt^{\top} $$
for $\UPt \in \R^{mn \times d}$ having $d$ orthonormal columns
and $\SPt \in \R^{d \times d}$ diagonal.
Since $\bM$ is a reasonable estimate for $\Ptilde = \E \bM$,
the matrix $\Zhat = \OMNI(\bA^{(1)},\bA^{(2)},\dots,\bA^{(m)},d)$
is a natural estimate of the $mn$ latent positions
collected in the matrix
$\bZ = [\bX^{\top} \, \bX^{\top}\, \dots \,\bX^{\top}]^{\top} \in \R^{mn \times d}$.
Here again, as in Remark~\ref{rem:nonid}, $\Zhat$ only recovers
the true latent positions $\bZ$ up to an orthogonal rotation.
The matrix
\begin{equation} \label{eq:Zstruct}
\Zstar = \begin{bmatrix} \Xstar \\ \Xstar \\ \vdots \\ \Xstar \end{bmatrix}
= \UPt \SPt^{1/2} \in \R^{mn \times d},
\end{equation}
provides a reasonable canonical choice of latent positions,
so that $\bZ = \Zstar \bW$ for some suitably-chosen orthogonal matrix
$\bW \in \R^{d \times d}$; again, just as for a single random dot product graph, spectral embedding of the omnibus matrix is a consistent estimator for the latent positions (up to rotation).

Below, we state precise results on consistency and asymptotic normality of the embedding of the omnibus matrix $\bM$.  The proofs are similar to, but somewhat more involved than, the aforementioned analogues for the adjacency spectral embedding for one graph. We also demonstrate from simulations that the omnibus embedding can be successfully leveraged for subsequent inference, specifically two-sample testing.

First, Lemma~\ref{lem:omni2toinf} shows that the omnibus embedding
provides uniformly consistent estimates of the true latent positions,
up to an orthogonal transformation,
roughly analogous to Lemma 5 in \cite{lyzinski13:_perfec}.
Lemma~\ref{lem:omni2toinf}
shows consistency of the omnibus embedding under the $\tti$ norm,
implying that all $mn$ of the estimated latent positions 
are near (a rotation of) their corresponding true positions.
\begin{lemma}\label{lem:omni2toinf}
	With $\Ptilde$, $\bM$, $\UM$, and $\UPt$ defined as above, there exists
	an orthogonal matrix $\Wtilde \in \R^{d \times d}$ such that
	with high probability,
	\begin{equation}\label{eq:omni2toinf_actualbound}
	\|\UM \SM^{1/2}-\UPt \SPt^{1/2} \Wtilde \|_{\tti}
	\le \frac{Cm^{1/2} \log mn }{\sqrt{n}} .
	\end{equation}
\end{lemma}

As with the adjacency spectral embedding, we once again can show the asymptotic normality of the individual rows of the omnibus embedding. Note that the covariance matrix does change with $m$, and for $m$ large, this results in a nontrivial variance reduction.

\begin{theorem} \label{thm:main}
	Let $(\bA^{(1)},\bA^{(2)},\dots,\bA^{(m)},\bX) \sim \JRDPG(F,n,m)$ for some
	$d$-dimensional inner product distribution $F$ and let $\bM$ denote
	the omnibus matrix as in \eqref{eq:omnidef}. Let
	$\bZ = \Zstar \bW$ with $\Zstar$ as defined in Equation~\eqref{eq:Zstruct},
	with estimate
	$\Zhat = \OMNI(\bA^{(1)},\bA^{(2)},\dots,\bA^{(m)},d)$.
	Let $h = m(s-1) + i$ for $i \in [n],s \in [m]$, so that $\Zhat_h$
	denotes the estimated latent position
	of the $i$-th vertex in the $s$-th graph $\bA^{(s)}$.
	That is, $\Zhat_h$ is the
	column vector formed by transposing the $h$-th row of the matrix
	$\Zhat = \UM \SM^{1/2} = \OMNI(\bA^{(1)},\bA^{(2)},\dots,\bA^{(m)},d)$.
	Let $\Phi(\bx,\bSigma)$ denote the cumulative distribution function of a (multivariate)
	Gaussian with mean zero and covariance matrix $\bSigma$,
	evaluated at $\bx \in \R^d$.
	There exists a sequence of orthogonal $d$-by-$d$ matrices
	$( \Wntilde )_{n=1}^\infty$ such that for all $\bx \in \R^d$,
	$$ \lim_{n \rightarrow \infty}
	\Pr\left[ n^{1/2} \left( \Zhat \Wntilde - \bZ \right)_h
	\le \bx \right]
	= \int_{\supp F} \Phi\left(\bx, \bSigma(\by) \right) dF(\by), $$
	where
	$\bSigma(\by) = (m+3)\bDelta^{-1} \Sigmatilde(\by) \bDelta^{-1}/(4m), $
	$\bDelta = \mathbb{E}[\bX_1 \bX_1^{\top}]$ and
	$$\Sigmatilde(\by)
	= \E\left[ (\by^{\top} \bX_1 - ( \by^{\top} \bX_1)^2 ) \bX_1 \bX_1^{\top} \right].$$
\end{theorem}

Next, we summarize from \cite{levin_omni_2017} experiments on synthetic data
exploring the efficacy of the omnibus embedding described above.
If we merely wish to estimate the latent positions $\bX$
of a set of $m$ graphs
$(\bA^{(1)},\bA^{(2)},\dots,\bA^{(m)},\bX) \sim \JRDPG(F,n,m)$,
the estimate $\Xbar = \ASE( \sum_{i=1}^m \bA^{(i)}/m, d )$,
the embedding of the sample mean of the adjacency matrices
performs well asymptotically \cite{runze_law_large_graphs}.
Indeed, all else equal,
the embedding $\Xbar$ is preferable to the omnibus embedding
if only because it requires an eigendecomposition
of an $n$-by-$n$ matrix rather
than the much larger $mn$-by-$mn$ omnibus matrix.

\begin{figure}[t!]
	\centering
	\includegraphics[width=0.6\columnwidth]{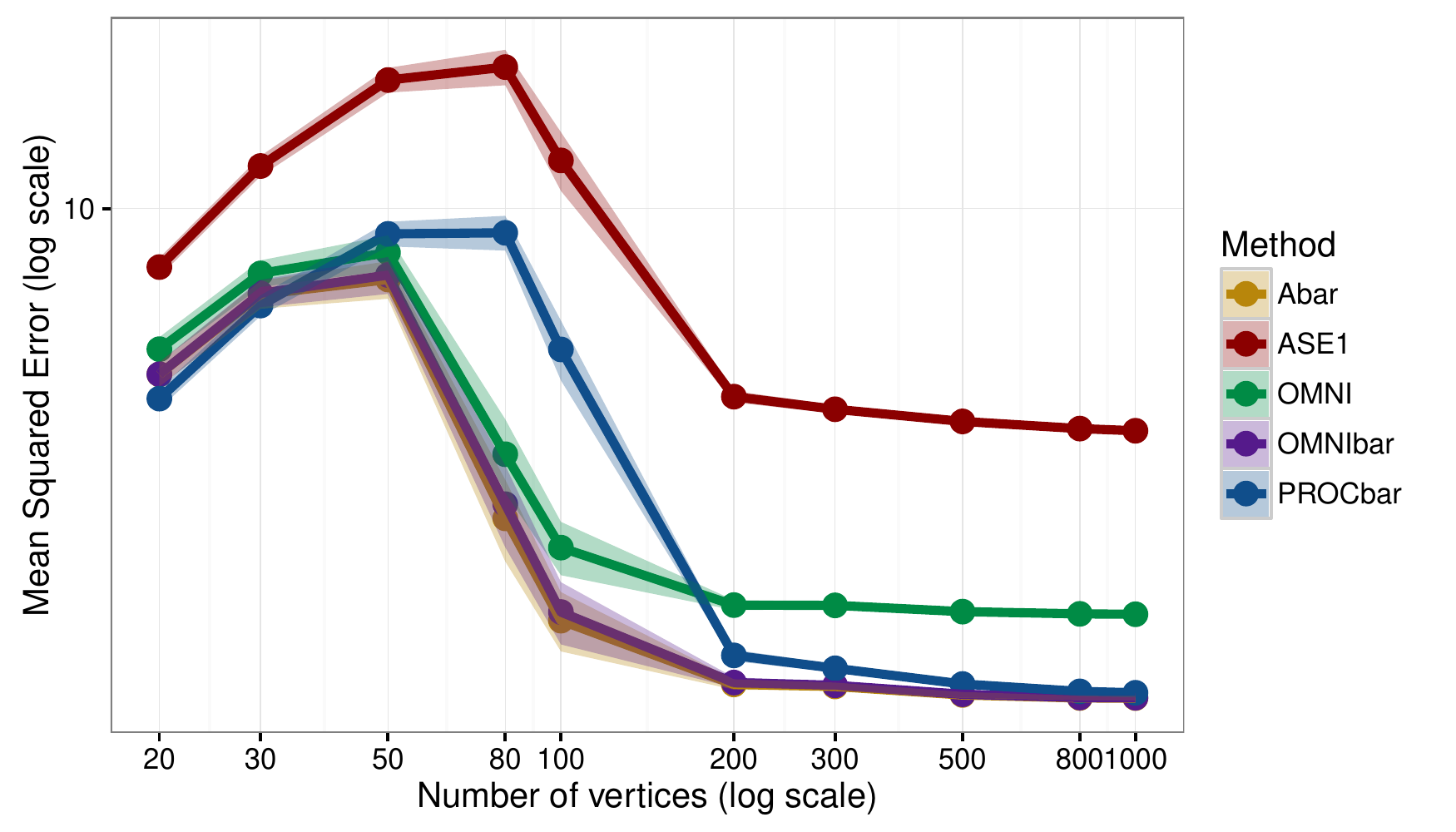}
	\caption{Mean squared error (MSE) in recovery of latent positions (up to rotation) in a 2-graph joint RDPG model as a function of the number of vertices.  Performance of ASE applied to a single graph (red), ASE embedding of the mean graph (gold), the Procrustes-based pairwise embedding (blue), the omnibus embedding (green) and the mean omnibus embedding (purple). Each point is the mean of 50 trials; error bars indicate $\pm 2(SE)$. Mean omnibus embedding (OMNIbar) is competitive with $\ASE(\Abar,d)$; Procrustes alignment estimation is notably inferior to the other two-graph techniques for graphs of size between 80 and 200 vertices (note that the gap appears to persist at larger graph sizes, though it shrinks). Figure duplicated from \cite{levin_omni_2017}.}
	\label{fig:compareMSE}
\end{figure}

Of course, the omnibus embedding can still be used to
to estimate the latent positions, potentially at the cost of
increased variance.
Figure \ref{fig:compareMSE} compares the mean-squared error of various
techniques for estimating the latent positions for a random dot product graph.
The figure plots the (empirical) mean squared error in recovering the
latent positions of a $3$-dimensional
JRDPG as a function of the number of vertices $n$.
Each point in the plot is the empirical mean of 50 independent trials; in each trial, the latent positions are drawn i.i.d. from a
Dirichlet with parameter $[1,\,1,\,1]^{\top} \in \R^{3}$.
Once the latent positions are so obtained, we independently generate two random dot product graphs,
$\bA^{(1)},\bA^{(2)} \in \R^{n \times n}$ with these latent positions.
The figure is interpreted as follows:
\begin{enumerate}
	\item {\bf ASE1 (red)}: we embed only one of the two observed graphs,
	and use only the ASE of that graph to estimate the latent positions
	in $\bX$, ignoring entirely the information present in
	$\bA^{(2)}$. This condition serves as a baseline for how much
	additional information is provided by the second graph $\bA^{(2)}$.
	\item {\bf Abar (gold)}: we embed the average of the two graphs,
	$\Abar = (\bA^{(1)} + \bA^{(2)})/2$ as $\Xhat = \ASE( \Abar, 3 )$.
	\item {\bf OMNI (green)}: We apply the omnibus embedding to obtain
	$\Zhat = \ASE(\bM,3)$,
	where $\bM$ is as in Equation~\eqref{eq:omnidef}.
	We then use only the first $n$ rows of
	$\Zhat \in \R^{2n \times d}$ as our estimate of $\bX$.
	This embedding incorporates information available
	in both graphs $\bA^{(1)}$ and $\bA^{(2)}$, but does not
weight them equally, since the first rows of $\Zhat$ are based
	primarily on the information contained in $\bA^{(1)}$.
	\item {\bf OMNIbar (purple)}: We again apply the omnibus embedding to obtain
	estimated latent positions
	$\Zhat = \ASE(\bM,3)$, but this time we use all available
	information by averaging the first $n$ rows and the second $n$ rows
	of $\Zhat$.
	\item {\bf PROCbar (blue)}: We separately embed the graphs
	$\bA^{(1)}$ and $\bA^{(2)}$, perform a Procrustes alignment between the two embeddings, and average the aligned embeddings to obtain
	our final estimate of the latent positions.
\end{enumerate}
First, let us note that ASE applied to a single graph (red)
lags all other methods, as expected, since this discards all information from the second graph.
For very small graphs, the dearth of signal is such that no method will
recover the latent positions accurately.

Crucially, however, we see that the OMNIbar estimate (purple) performs nearly
identically to the Abar estimate (gold), the natural choice among spectral methods for the estimation latent positions
The Procrustes estimate (in blue)
provides a two-graph analogue of ASE (red):
it combines two ASE estimates via Procrustes alignment,
but does not enforce an {\em a priori} alignment of the estimated latent positions.
As predicted by the results in \cite{lyzinski13:_perfec} and \cite{tang14:_semipar},
the Procrustes estimate is competitive with the Abar (gold)
estimate for suitably large graphs.
The OMNI estimate (in green) serves, in a sense, as an intermediate, since it uses information available from both graphs,
but in contrast to Procrustes (blue), OMNIbar (purple)
and Abar (gold), it does not make complete use of the information
available in the second graph.
For this reason, it is noteworthy that the OMNI estimate
outperforms the Procrustes estimate for graphs of 80-100 vertices.
That is, for certain graph sizes,
the omnibus estimate appears to more optimally leverage the information in both graphs
than the Procrustes estimate does,
despite the fact that the information in the second graph has comparatively little
influence on the OMNI embedding.

The omnibus embedding can also be applied to testing
the semiparametric hypothesis that two observed graphs are drawn from the
same underlying latent positions.
Consider a collection of latent positions
$\bX_1,\bX_2,\dots,\bX_n,\bY_1,\bY_2,\dots,\bY_n \in \R^d$.
Let the graph $G_1$ with adjacency matrix $\bA^{(1)}$ have edges distributed
independently as
$ \bA^{(1)}_{ij} \sim \Bern( \bX_i^{\top} \bX_j )$.
Similarly, let $G_2$ have adjacency matrix $\bA^{(2)}$ with edges
distributed independently as
$ \bA^{(2)}_{ij} \sim \Bern( \bY_i^{\top} \bY_j )$.
The omnibus embedding provides a natural test of
the null hypothesis \eqref{eq:H0}
by comparing the first $n$ and last $n$ embeddings of the omnibus matrix
$$ \bM = \begin{bmatrix} \bA^{(1)} & (\bA^{(1)} + \bA^{(2)})/2 \\
(\bA^{(1)} + \bA^{(2)})/2 & \bA^{(2)}
\end{bmatrix}. $$
Intuitively, when $H_0$ holds,
the distributional result in Theorem~\ref{thm:main} holds,
and the $i$-th and $(n+i)$-th rows of $\OMNI(\bA^{(1)},\bA^{(2)},d)$
are equidistributed (though they are not independent).
On the other hand, when $H_0$ fails to hold, there exists at least one
$i \in [n]$ for which the $i$-th and $(n+i)$-th rows of $\bM$ are \emph{not}
identically distributed, and thus the corresponding embeddings are
also distributionally distinct.
This suggests a test that compares the first $n$ rows of
$\OMNI(\bA^{(1)},\bA^{(2)},d)$
against the last $n$ rows (see below for details).
Here, we empirically explore the power this test against its
Procrustes-based alternative from \cite{tang14:_semipar}.

We draw $\bX_1,\bX_2,\dots,\bX_n \in \R^3$ i.i.d. according to a
Dirichlet distribution $F$ with parameter
$\alphavec = [1, 1, 1]^{\top}$.
With $\bX$ defined as the matrix
$ \bX = [\bX_1 \bX_2 \dots \bX_n]^{\top} \in \R^{n \times 3}, $
let graph $G_1$ have adjacency matrix $\bA^{(1)}$, where
$ \bA^{(1)}_{ij} \sim \Bern( (\bX \bX^{\top})_{ij} )$.
We generate a second graph $G_2$ by first
drawing random points $\bZ_1,\bZ_2,\dots,\bZ_n \iid F$.
Selecting a set of indices $I \subset [n]$ of size $k < n$ uniformly at
random from among all such $\binom{n}{k}$ sets,
we let $G_2$ have latent positions
$$ \bY_i = \begin{cases} \bZ_i & \mbox{ if } i \in I \\
\bX_i & \mbox{ otherwise. } \end{cases} $$
With $\bY$ the matrix
$\bY = [\bY_1, \bY_2, \dots, \bY_n]^{\top} \in \R^{n \times 3}, $
we generate graph $G_2$ with adjacency matrix $\bA^{(2)}$, where
$ \bA^{(2)}_{ij} \sim \Bern( (\bY \bY^{\top})_{ij} ).$
We wish to test
\begin{equation} \label{eq:ptsnull}
H_0 : \bX = \bY.
\end{equation}
Consider two different tests, one based on
a Procrustes alignment of the adjacency spectral embeddings of $G_1$ and $G_2$
\cite{tang14:_semipar}
and the other based on the omnibus embedding.
Both approaches are based on estimates of the latent positions
of the two graphs.
In both cases we use a test statistic of the form
$ T = \sum_{i=1}^n \| \Xhat_i - \Yhat_i \|_F^2, $
and accept or reject based on a Monte Carlo estimate of the
critical value of $T$ under the null hypothesis,
in which $\bX_i = \bY_i$ for all $i \in [n]$.
In each trial, we use $500$ Monte Carlo iterates to estimate the
distribution of $T$.

We note that in the experiments presented here,
we assume that the latent positions
$\bX_1,\bX_2,\dots,\bX_n$ of graph $G_1$ are known for sampling purposes,
so that the matrix $\bP = \E \bA^{(1)}$ is known exactly, rather than
estimated from the observed adjacency matrix $\bA^{(1)}$.
This allows us to sample from the true null distribution.
As proved in \cite{lyzinski13:_perfec},
the estimated latent positions $\Xhat_1 = \ASE(\bA^{(1)})$
and $\Xhat_2 = \ASE( \bA^{(2)} )$ recover the true latent positions
$\bX_1$ and $\bX_2$ (up to rotation) to arbitrary accuracy
in $(2,\infty)$-norm for suitably large $n$~\cite{lyzinski13:_perfec}.
Without using this known matrix $\bP$, we would require that our matrices
have tens of thousands of vertices before the variance associated with
estimating the latent positions would no longer overwhelm the signal present
in the few altered latent positions.

Three major factors influence the complexity of testing
the null hypothesis in Equation \eqref{eq:ptsnull}:
the number of vertices $n$,
the number of changed latent positions $k = |I|$,
and the distances $\|\bX_i - \bY_i\|_F$ between the latent positions.
The three plots in Figure \ref{fig:trueP:power} illustrate
the first two of these three factors.
These three plots show the power of two different approaches to testing
the null hypothesis \eqref{eq:ptsnull} for different sized graphs
and for different values of $k$, the number of altered latent positions.
In all three conditions, both methods improve as the number of vertices
increases, as expected, especially since we do not require
estimation of the underlying expected matrix $\bP$ for Monte Carlo
estimation of the null distribution of the test statistic.
We see that when only one vertex is changed, neither method has power much above $0.25$.
However, in the case of $k = 5$ and $k = 10$, is it clear that the
omnibus-based test achieves higher power than the Procrustes-based
test, especially in the range of 30 to 250 vertices. A more detailed examination of the relative impact of these factors in testing is given in \cite{levin_omni_2017}.

\begin{figure}[t!]
	\centering
	\subfloat[]{ \includegraphics[width=0.325\columnwidth]{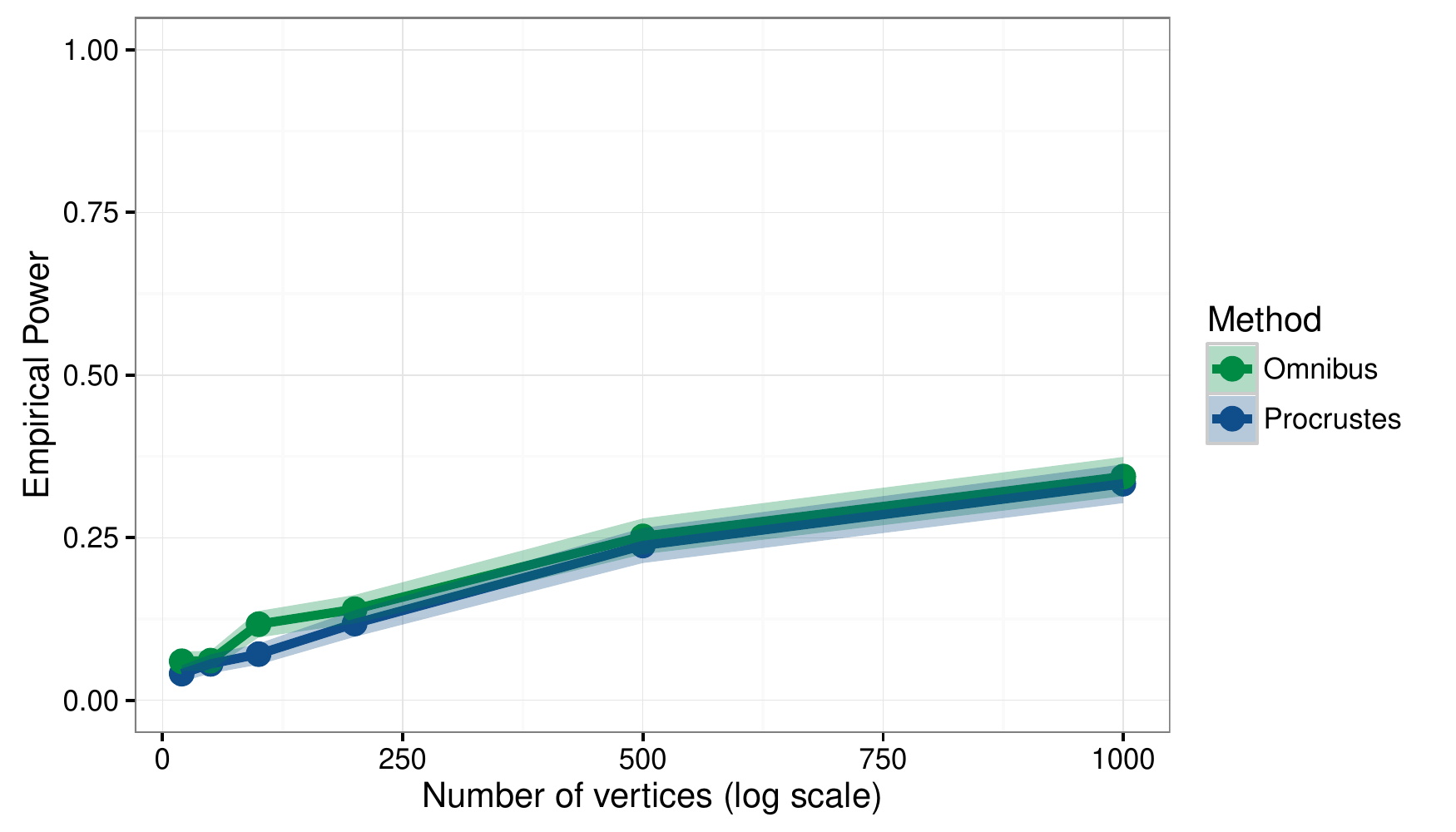} }
	\subfloat[]{ \includegraphics[width=0.325\columnwidth]{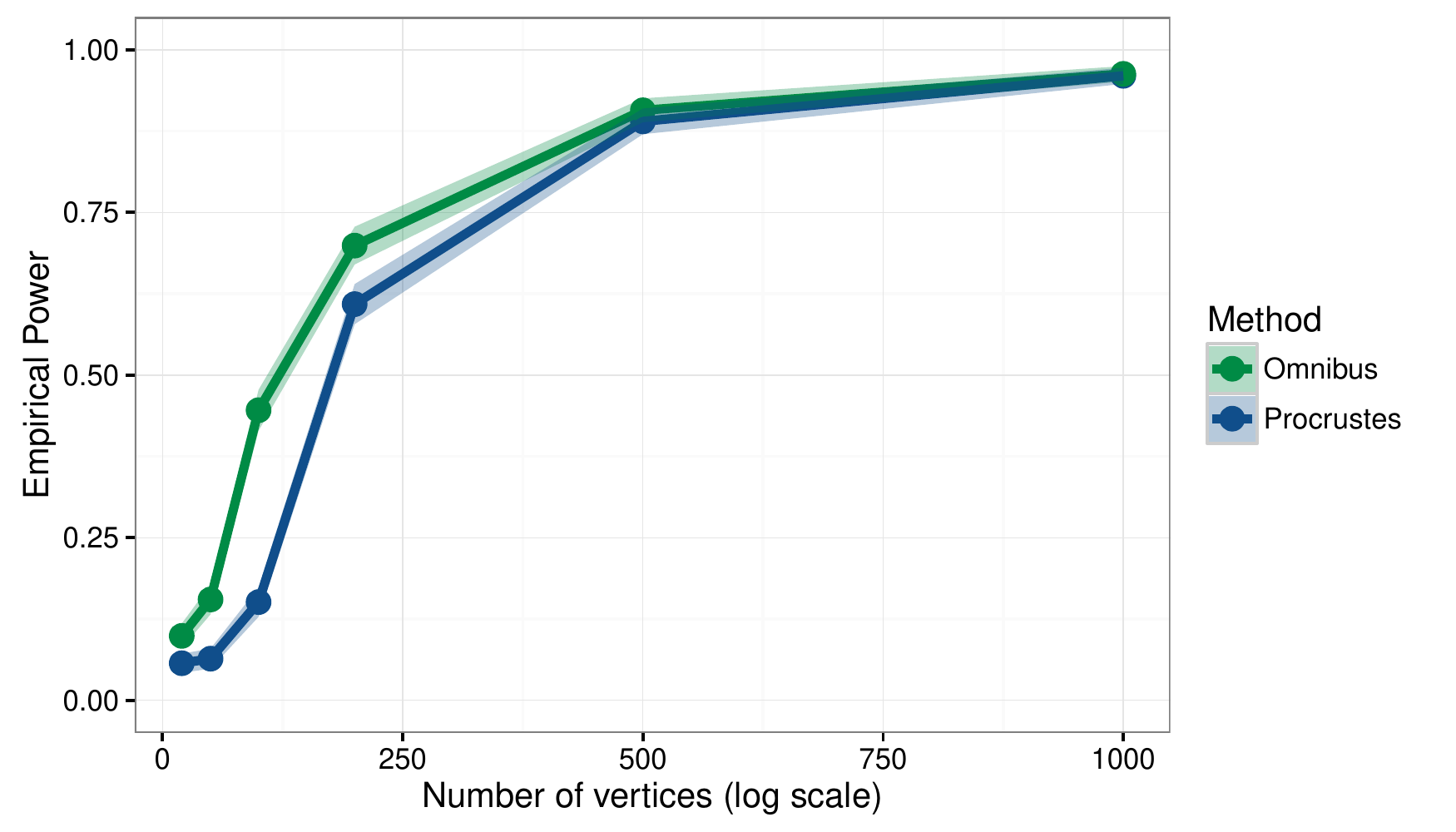} }
	\subfloat[]{ \includegraphics[width=0.325\columnwidth]{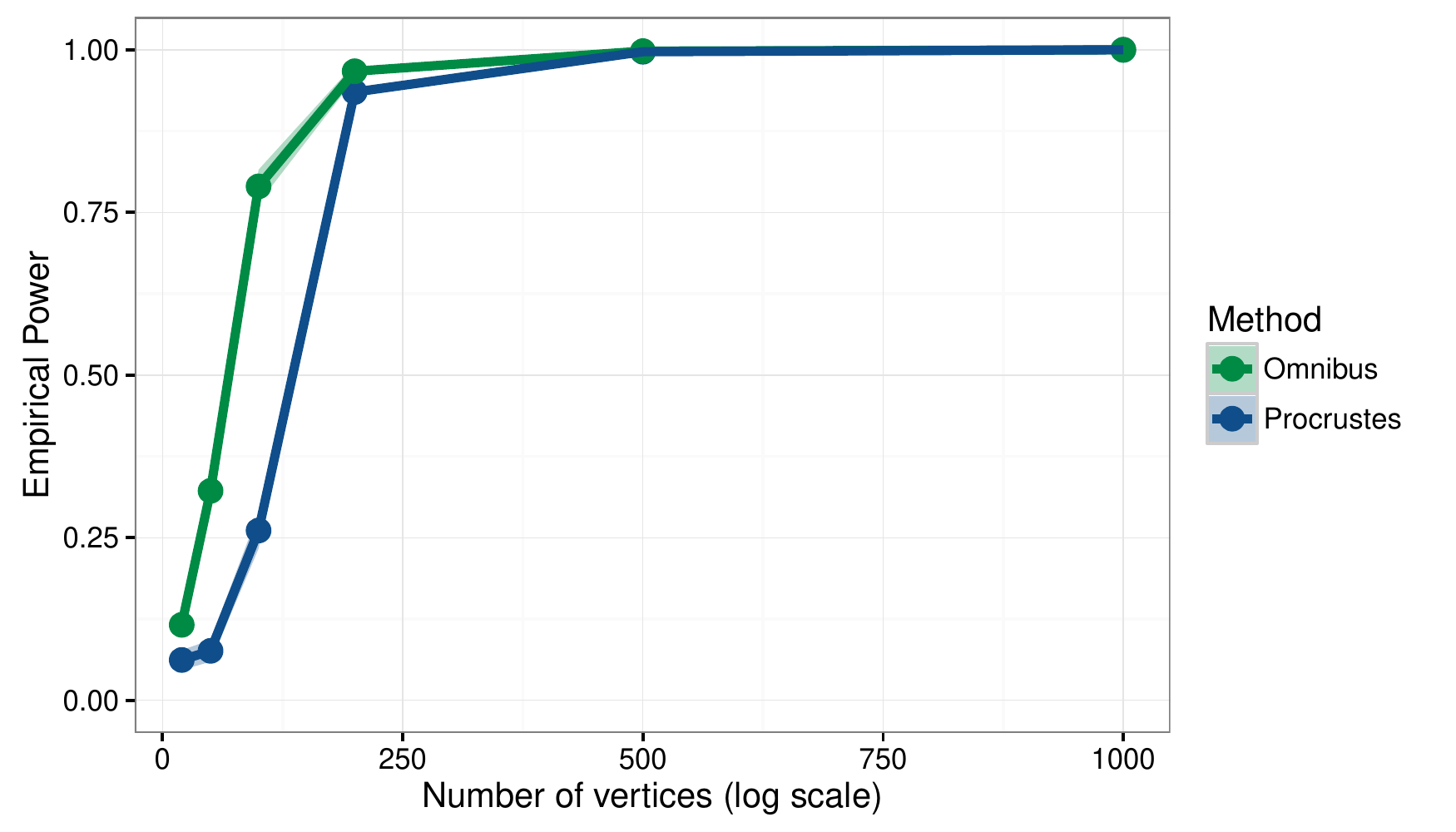} }
	\caption{Power of the ASE-based (blue) and omnibus-based (green)
		tests to detect when the two graphs being testing differ in
		(a) one, (b) five, and (c) ten of their latent positions.
		Each point is the proportion of 1000 trials for which the given
		technique correctly rejected the null hypothesis,
		and error bars denote two standard errors of this empirical mean
		in either direction. Figure duplicated from \cite{levin_omni_2017}. }
	\label{fig:trueP:power}
\end{figure}

In sum, our omnibus embedding provides a natural mechanism for the
simultaneous embedding of multiple graphs into a single vector space.
This eliminates the need for multiple Procrustes alignments,
which were required in previously-explored
approaches to multiple-graph testing \cite{tang14:_semipar}.
Recall the two-graph hypothesis testing framework of \cite{tang14:_semipar},
each graph is embedded separately, yielding estimates $\bX_1$ and $\bX_2$, and under the null hypothesis of equality of latent positions $bX$ up to some rotation, 
%
Procrustes alignment is required: the test statistic is 
\begin{equation} \label{eq:procmin}
\min_{\bW \in \calO_d} \| \Xhat_1 - \Xhat_2 \bW \|_F,
\end{equation}
and under the null hypothesis, a suitable rescaling of this converges as $n \rightarrow \infty$.
The effect of this Procrustes alignment on subsequent inference
is ill-understood; it has the potential to
introduce variance, and our simulations results
suggest that it negatively impacts performance in both estimation
and testing settings.
Furthermore, when the matrix $\bP = \bX \bX^{\top}$
does not have distinct eigenvalues
(i.e., is not uniquely diagonalizable), this Procrustes step is unavoidable,
since the difference $\|\Xhat_1 - \Xhat_2\|_F$ need not converge at all.

In contrast, our omnibus embedding builds an alignment of the graphs
into its very structure. To see this, consider, for simplicity, the $m=2$ case.
Let $\bX \in \R^{n \times d}$ be the matrix whose rows are the latent positions
of both graphs $G_1$ and $G_2$, and let $\bM \in \R^{2n \times 2n}$ be their
omnibus matrix.
Then 
\begin{equation*}
\E \bM = \Ptilde = \begin{bmatrix} \bP & \bP \\ \bP & \bP \end{bmatrix}
= \begin{bmatrix} \bX \\ \bX \end{bmatrix}
\begin{bmatrix} \bX \\ \bX \end{bmatrix}^{\top}.
\end{equation*}
Suppose now that we wish to factorize $\Ptilde$ as
$$ \Ptilde = \begin{bmatrix} \bX  \\ \bX \bW^* \end{bmatrix}
\begin{bmatrix} \bX  \\ \bX \bW^* \end{bmatrix}^{\top}
= \begin{bmatrix} \bP & \bX (\bW^*)^{\top} \bX^{\top} \\
\bX \bW^* \bX^{\top} & \bP \end{bmatrix}. $$
That is, we want to consider graphs $G_1$ and $G_2$ as being generated from
the same latent positions,
but in one case, say, under a  \emph{different} rotation.
This possibility necessitates the Procrustes alignment in the case of
separately-embedded graphs.
In the case of the omnibus matrix,
the structure of the $\Ptilde$ matrix implies that $\bW^*=\bI_d$.
In contrast to the Procrustes alignment,
the omnnibus matrix incorporates an alignment {\em a priori}.
Simulations show that the omnibus embedding
outperforms the Procrustes-based test for equality of latent positions,
especially in the case of moderately-sized graphs.

To further illustrate the utility of this omnibus embedding, consider the case of testing whether three different random dot product graphs have the same generating latent positions. The omnibus embedding gives us a {\em single} canonical representation of all three graphs: Let $\Xhat^O_1$, $\Xhat^O_2$, and $\Xhat^O_3$ be the estimates for the three latent position matrices generated from the omnibus embedding.
To test whether any two of these random graphs have the same generating latent positions, we merely have to compare the Frobenius norms of their differences, as opposed to computing three separate Procrustes alignments.
In the latter case, in effect, we do not have a canonical choice of coordinates in which to compare our graphs simultaneously.

\subsection{Nonparametric graph estimation and testing}
\label{subsec:nonpar}
The semiparametric and omnibus testing procedures we describe both focus on the estimation of the latent positions themselves. But a very natural concern, for a random dot product graph with distribution $F$, is the estimation of the {\em distribution} $F$.  We next address how the adjacency spectral embedding, judiciously integrated with kernel density estimation, can be used for nonparametric estimation and testing in random graphs.

Throughout our discussion on nonparametric estimation, we shall always assume that the distributions of the latent positions
satisfy the following distinct eigenvalues assumption. The assumption
implies that the estimates of the latent position obtained by the
adjacency spectral embedding will,
in the limit, be uniquely determined. 
\begin{assumption}\label{ass:rank_F}
  The distribution $F$ for the latent positions
  $X_1, X_2, \dots, \sim F$ is such that the second moment matrix
  $\mathbb{E}[X_1 X_1^{\top}]$ has $d$ distinct eigenvalues and $d$ is known.
\end{assumption}
We realize that Assumption \ref{ass:rank_F} is restrictive -- in
particular, it is not satisfied by the stochastic block model with
$K>2$ blocks of equal size and edge probabilities $p$ within
communities and $q$ between communities -- it is a necessary technical condition for us to obtain 
the limiting results of Theorem~\ref{eq:lpg}. 
The motivation behind this assumption is as follows: the matrix
$\mathbb{E}[X_1 X_1^{\top}]$ is of rank $d$ with $d$ known so
that given a graph $\mathbf{A} \sim \mathrm{RDPG}(F)$, one can
construct the adjacency spectral embedding of $\mathbf{A}$ into the
``right'' Euclidean space. The requirement that
$\mathbb{E}[X_1 X_1^{\top}]$ has $d$ distinct eigenvalues is---once again---due to the
intrinsic property of non-identifiability of random dot
product graphs. As always, for any random dot product graph
$\mathbf{A}$, the latent position $\mathbf{X}$ associated with
$\mathbf{A}$ can only be estimated up to some true but unknown
orthogonal transformation. Because we are concerned with two-sample
hypothesis testing, we must guard against the scenario in which we
have two graphs $\mathbf{A}$ and $\mathbf{B}$ with latent positions
$\mathbf{X} = \{X_i\}_{i=1}^{n} \overset{\mathrm{i.i.d}}{\sim} F$ and
$\mathbf{Y} = \{Y_k\}_{k=1}^{m} \overset{\mathrm{i.i.d}}{\sim} F$ but
whose estimates $\hat{\mathbf{X}}$ and $\hat{\mathbf{Y}}$ lie in
different, incommensurate subspaces of $\mathbb{R}^{d}$. That is to
say, the estimates $\hat{\mathbf{X}}$ and $\hat{\mathbf{Y}}$ satisfy
$\hat{\mathbf{X}} \approx \mathbf{X} \mathbf{W}_1$ and
$\hat{\mathbf{Y}} \approx \mathbf{Y} \mathbf{W}_2$, but
$\|\mathbf{W}_1 - \mathbf{W}_2\|_{F}$ does not converge to $0$ as $n,m
\rightarrow \infty$. See also \cite{fishkind15:_incom_phenom} for
exposition of a related so-called ``incommensurability phenomenon."


Our main point of departure for this subsection compared to Section~\ref{subsec:Testing} is the assumption
that, given a sequence of pairs of random dot product graphs with adjacency matrices $\mathbf{A}_n$ and $\mathbf{B}_n$, 
the rows of the latent positions $\mathbf{X}_n$ and
$\mathbf{Y}_n$ are independent samples from some fixed distributions $F$ and $G$,
respectively. The corresponding tests are therefore tests of equality between $F$ and
$G$. More formally, we consider the following two-sample nonparametric testing
problems for random dot product graphs. Let $F$ and $G$ be two inner product distributions.
Given $\mathbf{A} \sim \mathrm{RDPG}(F)$ and $\mathbf{B} \sim \mathrm{RDPG}(G)$, we
consider the tests:
\begin{enumerate}
\item{\em (Equality, up to orthogonal transformation)}
  \begin{align*}
  H_{0} \colon F \upVdash G
  \quad \text{against} \quad H_{A} \colon F  \nupVdash G,
  \end{align*}
  where $F \upVdash G$ denotes that there exists a unitary operator
  $U$ on $\mathbb{R}^{d}$ such that $F = G \circ U$ and $F \nupVdash
  G$ denotes that $F \not = G \circ U$ for any unitary operator $U$ on
  $\mathbb{R}^{d}$.
\item {\em (Equality, up to scaling)}
  \begin{align*}
H_{0} \colon F \upVdash G \circ c \quad \text{for
    some $c > 0$} \quad  
  \text{against} \quad H_{A} \colon F  \nupVdash G \circ c \quad
  \text{for any  $c > 0$},
  \end{align*}
  where $Y \sim F \circ c$ if $cY \sim F$.
\item {\em (Equality, up to projection)}   \begin{align*}
\quad H_{0} \colon F \circ \pi^{-1} \upVdash G
\circ \pi^{-1} \quad 
  \text{against} \quad H_{A} \colon F \circ \pi^{-1} \nupVdash G \circ \pi^{-1} ,
  \end{align*}
  where $\pi$ is the
  projection $x \mapsto x/\|x\|$; hence $Y \sim F \circ
  \pi^{-1}$ if $\pi^{-1}(Y) \sim F$. 
\end{enumerate}
We note that the above null hypotheses are nested; $F \upVdash G$
implies $F \upVdash G \circ c $ for $c = 1$ while $F \upVdash G
\circ c$ for some $c > 0$ implies $F \circ \pi^{-1} \upVdash G \circ
\pi^{-1}$. 

We shall address the above hypothesis testing problem by combining the framework of adjacency spectral embedding and 
the kernel-based hypothesis testing framework of \cite{gretton12:_kernel_two_sampl_test}. 
The testing procedure in \cite{gretton12:_kernel_two_sampl_test} is based on the following notion of the maximum mean discrepancy between
distributions. 
Let $\Omega$ be a
compact metric space and
$\kappa \, \colon\, \Omega \times \Omega \mapsto \mathbb{R}$ a
continuous, symmetric, and positive definite kernel on
$\Omega$. Denote by $\mathcal{H}$ the reproducing kernel Hilbert space
associated with $\kappa$. Now let $F$ be a probability distribution on
$\Omega$. Under mild conditions on $\kappa$, the map $\mu[F]$ defined
by
\begin{equation*}
  \mu[F] :=  \int_{\Omega} \kappa(\omega, \cdot) \, \mathrm{d} F(\omega)
\end{equation*}
belongs to $\mathcal{H}$.
Now, for given probability distributions $F$ and $G$ on $\Omega$, the
{\em maximum mean discrepancy} between $F$ and $G$ with respect to
$\mathcal{H}$ is the measure
\begin{equation*}
  \mathrm{MMD}(F, G; \mathcal{H}) := \|\mu[F] - \mu[G]
    \|_{\mathcal{H}}. 
\end{equation*}
We now summarize some important properties of the maximum mean
discrepancy from \citep{gretton12:_kernel_two_sampl_test}. 
\begin{theorem}
  \label{thm:mmd_unbiased_limiting}
  Let $\kappa
  \, \colon \,
  \mathcal{X} \times \mathcal{X} \mapsto \mathbb{R}$ be a positive definite
  kernel and denote by $\mathcal{H}$ the reproducing kernel Hilbert space
  associated with $\kappa$. Let $F$ and $G$ be probability distributions on $\Omega$; $X$
  and $X'$ independent random variables with distribution $F$, $Y$ and
  $Y'$ independent random variables with distribution $G$, and $X$ is
  independent of $Y$. Then
  \begin{equation}
    \label{eq:5}
    \begin{split}
    \| \mu[F] - \mu[G] \|^{2}_{\mathcal{H}} &= \sup_{h \in \mathcal{H}
    \colon \|h\|_{\mathcal{H}} \leq 1} |\mathbb{E}_{F}[h] -
    \mathbb{E}_{G}[h]|^{2} \\ &= \mathbb{E}[\kappa(X,X')] - 2 \mathbb{E}[\kappa(X,Y)]
    + \mathbb{E}[\kappa(Y,Y')].
    \end{split}
  \end{equation}
  Given $\mathbf{X} = \{X_i\}_{i=1}^{n}$ and $\mathbf{Y}
  = \{Y_k\}_{k=1}^{m}$ with $\{X_i\}
  \overset{\mathrm{i.i.d}}{\sim} F$ and $\{Y_i\}
  \overset{\mathrm{i.i.d}}{\sim} G$, 
  the quantity $U_{n,m}({\bf X}, {\bf Y})$ defined by
  \begin{equation}
    \label{eq:10}
    \begin{split}
     U_{n,m}({\bf X}, {\bf Y})
 &= \frac{1}{n(n-1)} 
     \sum_{j\not = i} \kappa(X_i,X_j)
    - \frac{2}{mn} \sum_{i=1}^{n} \sum_{k=1}^{m} \kappa(X_i, Y_k) \\ &+
    \frac{1}{m(m-1)} \sum_{l \not = k} \kappa(Y_k, Y_l)
    \end{split}
  \end{equation}
  is an {unbiased consistent estimate} of $\|\mu[F] -
  \mu[G]\|_{\mathcal{H}}^{2}$. Denote by $\tilde{\kappa}$ the kernel
  \begin{equation*}
    \begin{split}
    \tilde{\kappa}(x,y) &= \kappa(x,y) - \mathbb{E}_{z}
    \kappa(x, z) - \mathbb{E}_{z'} \kappa(z', y) +
    \mathbb{E}_{z,z'} \kappa(z,z')
     \end{split}
  \end{equation*}
  where the expectation is taken with respect to $z, z' \sim F$. 
Suppose that $\tfrac{m}{m+n} \rightarrow
  \rho \in (0,1)$ as $m, n \rightarrow \infty$. Then under the null
  hypothesis of $F = G$, 
  \begin{equation}
    \label{eq:mmd-X}
    (m+n) U_{n,m}(\mathbf{X}, \mathbf{Y})
    \overset{d}{\longrightarrow} \frac{1}{\rho(1 - \rho)} \sum_{l=1}^{\infty}
    \lambda_{l} (\chi^{2}_{1l} - 1)
  \end{equation}
  where $\{\chi^{2}_{1l}\}_{l=1}^\infty$ is a sequence of independent $\chi^{2}$
  random variables with one degree of freedom, and $\{\lambda_{l}\}$
  are the eigenvalues of the integral operator $\mathcal{I}_{F,
    \tilde{\kappa}}:\mathcal{H}\mapsto \mathcal{H}$ defined as
  \begin{equation*}
    I_{F, \tilde{\kappa}}(\phi)(x)=\int_{\Omega} \phi(y)\tilde{\kappa}(x,y) dF(y).
    \end{equation*}
    Finally, if $\kappa$ is a universal or
    characteristic kernel \citep{sriperumbudur11:_univer_charac_kernel_rkhs_embed_measur,
    steinwart01:_suppor_vector_machin}, then $\mu$
  is an injective map, i.e., $\mu[F] = \mu[G]$ if and only if $F = G$.
\end{theorem}

\begin{remark}
  A kernel $\kappa \colon \mathcal{X} \times \mathcal{X} \mapsto
  \mathbb{R}$ is universal if $\kappa$ is a continuous function of
  both its arguments and if the reproducing kernel Hilbert space
  $\mathcal{H}$ induced by $\kappa$ is dense in the space of
  continuous functions on $\mathcal{X}$ with respect to the supremum
  norm. Let $\mathcal{M}$ be a family of Borel probability measures on
  $\mathcal{X}$. A kernel $\kappa$ is characteristic for $\mathcal{M}$
  if the map $\mu \in \mathcal{M} \mapsto \int \kappa(\cdot, z) \mu(
  dz)$ is injective. If $\kappa$ is universal, then $\kappa$ is
  characteristic for any $\mathcal{M}$
  \citep{sriperumbudur11:_univer_charac_kernel_rkhs_embed_measur}. As
  an example, let $\mathcal{X}$ be a finite dimensional Euclidean
  space and define, for any $q \in (0,2)$, $k_{q}(x,y) =
  \tfrac{1}{2}(\|x\|^{q} + \|y\|^{q} - \|x - y\|^{q})$. The kernels
  $k_{q}$ are then characteristic for the collection of probability
  distributions with finite second moments
  \citep{lyons11:_distan,sejdinovic13:_equiv_rkhs}. In addition, by
  Eq.~\eqref{eq:5}, the maximum mean discrepancy with reproducing
  kernel $k_{q}$ can be written as
\begin{equation*}
  \mathrm{MMD}^{2}(F,Q;k_{q}) = 2 \mathbb{E} \|X - Y\|^{q} - 
  \mathbb{E}\|X - X'\|^{q} - \mathbb{E} \|Y - Y'\|^{q}.
\end{equation*}
where $X, X'$ are
independent with distribution $F$, $Y, Y'$ are independent with
distribution $G$, and $X, Y$ are independent.  This coincides with the
notion of the energy distances of \citep{szekely13:_energ}, or, when $q
= 1$, a special case of the one-dimensional interpoint comparisons
of \citep{maa96:_reduc}. Finally, we note that $(m+n)U_{n,m}(\mathbf{X}, \mathbf{Y})$ under the null hypothesis of
  $F = G$ in Theorem~\ref{thm:mmd_unbiased_limiting} depends
  on the $\{\lambda_l\}$ which, in turn, depend on the distribution
  $F$; thus the limiting distribution is not distribution-free.
  Moreover the eigenvalues $\{\lambda_l\}$ can, at best, be estimated;
  for finite $n$, they cannot be explicitly determined when $F$ is
  unknown. In practice, generally the critical values are estimated
  through a bootstrap resampling or permutation test.
\end{remark}

We focus on the nonparametric two-sample hypothesis test of $\mathbb{H}_0 \colon F \upVdash G$ against $\mathbb{H}_A \colon F \nupVdash G$.
For our purposes, we shall assume
henceforth that $\kappa$ is a twice continuously-differentiable radial
kernel and that $\kappa$ is also universal. 
To justify this assumption on our kernel, we point out that in Theorem~\ref{thm:mmd_unbiased_ase} below, we show that the test
statistic $U_{n,m}(\hat{\mathbf{X}}, \hat{\mathbf{Y}})$ based on the
estimated latent positions converges to the corresponding statistic
$U_{n,m}(\mathbf{X}, \mathbf{Y})$ for the true but unknown latent
positions. Due to the non-identifiability of the random dot product
graph under unitary transformation, \emph{any} estimate of the latent
positions is close, only up to an appropriate orthogonal transformation, to
$\mathbf{X}$ and $\mathbf{Y}$. For a radial kernel $\kappa$, this implies the approximations
$\kappa(\hat{X}_i, \hat{X}_j) \approx \kappa(X_i, X_j)$,
$\kappa(\hat{Y}_k,\hat{Y}_l) \approx \kappa(Y_k, Y_l)$ and the
convergence of $U_{n,m}(\hat{\mathbf{X}}, \hat{\mathbf{Y}})$ to
$U_{n,m}(\mathbf{X}, \mathbf{Y})$. If $\kappa$ is
not a radial kernel, the above approximations might not hold and
$U_{n,m}(\hat{\mathbf{X}}, \hat{\mathbf{Y}})$ need not converge to
$U_{n,m}(\mathbf{X}, \mathbf{Y})$. The assumption that
$\kappa$ is twice continuously-differentiable is technical. Finally, the
assumption that $\kappa$ is universal allows the test procedure to be
consistent against a large class of alternatives.

 \begin{theorem}
  \label{thm:mmd_unbiased_ase}
  Let $(\mathbf{X}, \mathbf{A}) \sim \mathrm{RDPG}(F)$ and
  $(\mathbf{Y}, \mathbf{B}) \sim \mathrm{RDPG}(G)$ be independent
  random dot product graphs with latent position distributions $F$ and
  $G$. Furthermore, suppose that both $F$ and $G$ satisfies the
  distinct eigenvalues condition in Assumption~\ref{ass:rank_F}.
  Consider the hypothesis test
  \begin{align*}
    H_{0} \colon F \upVdash G \quad 
    \text{against} \quad H_{A} \colon F  \nupVdash G.
  \end{align*}
  Denote by $\hat{\mathbf{X}} = \{\hat{X}_1, \dots, \hat{X}_n\}$
  and $\hat{\mathbf{Y}} = \{\hat{Y}_1, \dots, \hat{Y}_m\}$ the
  adjacency spectral embedding of $\mathbf{A}$ and
  $\mathbf{B}$, respectively. 
  Let $\kappa$ be a twice continuously-differentiable radial kernel and $U_{n,m}(\hat{\mathbf{X}},\hat{\mathbf{Y}})$ be defined as
\begin{equation*}
  \begin{split}
    U_{n,m}(\hat{\mathbf{X}}, \hat{\mathbf{Y}}) &= \frac{1}{n(n-1)}
    \sum_{j \not = i}
    \kappa(\hat{X}_i,
    \hat{X}_j) - \frac{2}{mn} \sum_{i=1}^{n}
    \sum_{k=1}^{m} \kappa(\hat{X}_i,
    \hat{Y}_k) + \frac{1}{m(m-1)} \sum_{l \not = k} \kappa(\hat{Y}_k, \hat{Y}_l).
  \end{split}
\end{equation*}
  Let $\mathbf{W}_1$
  and $\mathbf{W}_{2}$
  be $d \times d$ orthogonal matrices in the eigendecomposition
  $\mathbf{W}_1 \mathbf{S}_1 \mathbf{W}_{1}^{\top} = \mathbf{X}^{\top}
  \mathbf{X}$, $\mathbf{W}_{2} \mathbf{S}_{2} \mathbf{W}_{2} =
  \mathbf{Y}^{\top} \mathbf{Y}$, respectively.  
Suppose that $m, n \rightarrow \infty$
  and $m/(m+n) \rightarrow \rho \in (0,1)$. Then under the null
  hypothesis of $F \upVdash G$, the sequence of matrices $\mathbf{W}_{n,m} = \mathbf{W}_{2} \mathbf{W}_{1}^{\top}$ satisfies
  \begin{equation}
    \label{eq:conv_mmdXhat_null}
    (m+n) (U_{n,m}(\hat{\mathbf{X}}, \hat{\mathbf{Y}}) -
    U_{n,m}(\mathbf{X}, \mathbf{Y} \mathbf{W}_{n,m})) \overset{\mathrm{a.s.}}{\longrightarrow} 0.
  \end{equation}
  Under the alternative hypothesis of $F \nupVdash G$, the sequence of
  matrices ${\bf W}_{n,m} $ satisfies
 \begin{equation}
   \label{eq:conv_mmdXhat_alt}
   \frac{m+n}{\log^2{\!(m+n)}} 
   (U_{n,m}(\hat{\mathbf{X}}, \hat{\mathbf{Y}}) - U_{n,m}(\mathbf{X},
   \mathbf{Y} \mathbf{W}_{n,m})) \overset{\mathrm{a.s.}}{\longrightarrow} 0.
  \end{equation}
\end{theorem}
Eq.\eqref{eq:conv_mmdXhat_null} and Eq.\eqref{eq:conv_mmdXhat_alt}
state that the test statistic $U_{n,m}(\hat{\mathbf{X}},
\hat{\mathbf{Y}})$ using the {\em estimated} latent positions is
almost identical to the statistic $U_{n,m}(\mathbf{X}, \mathbf{Y} \mathbf{W}_{n,m})$
using the true latent positions, under
both the null and alternative hypothesis. If we assume that $\kappa$ is a universal
kernel, then $U_{n,m}(\mathbf{X}, \mathbf{Y} \mathbf{W}_{n,m})$ converges
to $0$ under the null and converges to a positive number under the
alternative. The test statistic $U_{n,m}(\hat{\mathbf{X}},
\hat{\mathbf{Y}})$ therefore yields a test procedure that is
consistent against any alternative, provided that both $F$ and $G$
satisfy Assumption~\ref{ass:rank_F}, namely that the second moment
matrices have $d$ distinct eigenvalues. 

We next consider the case of testing the hypothesis that the
distributions $F$ and $G$ are equal up to scaling or equal up to projection. For the test of equality up to scaling, 
  \begin{align*}
\quad H_{0} \colon F \upVdash G \circ c \quad \text{for
    some $c > 0$}
  \quad \text{against} \quad H_{A} \colon F  \nupVdash G \circ c \quad
  \text{for any  $c > 0$},
  \end{align*}
  where $Y \sim F \circ c$ if $cY \sim F$, 
we modified 
Theorem~\ref{thm:mmd_unbiased_ase} by first scaling the adjacency
spectral embeddings by the norm of the empirical means before
computing the kernel test statistic. 
In particular, let 
\begin{align*}
  \hat{s}_{X}=n^{-1/2}\|\hat{\mathbf{X}}\|_{F}, \quad \hat{s}_{Y}=m^{-1/2}
    \|\hat{\mathbf{Y}}\|_{F}, \quad 
      s_{X}= n^{-1/2} \|\mathbf{X}\|_{F},\quad s_{Y} =
    m^{-1/2} \|\mathbf{Y}\|_{F},
  \end{align*}
  then the conclusions of Theorem~\ref{thm:mmd_unbiased_ase} hold
  when we replace $U_{n,m}(\hat{\mathbf{X}}, \hat{\mathbf{Y}})$ and $U_{n,m}(\mathbf{X}, \mathbf{Y} \mathbf{W}_{n,m})$ with 
  $U_{n,m}(\hat{\mathbf{X}}/\hat{s}_{X},\hat{\mathbf{Y}}/\hat{s}_{Y})$ and 
  $U_{n,m}(\mathbf{X}/s_{X},\mathbf{Y} \mathbf{W}_{n,m}/s_{Y})$, respectively.  
  
  For the test of equality up to projection, 
  \begin{align*}
\quad H_{0} \colon F \circ \pi^{-1} \upVdash G
\circ \pi^{-1} \quad \text{against} \quad H_{A} \colon F \circ \pi^{-1} \nupVdash G \circ \pi^{-1} ,
  \end{align*}
  where $\pi$ is the
  projection $x \mapsto x/\|x\|$ that maps $x$ onto the unit sphere
  in $\mathbb{R}^{d}$, 
  the conclusions of
Theorem~\ref{thm:mmd_unbiased_ase} hold when we replace 
$U_{n,m}(\hat{\mathbf{X}}, \hat{\mathbf{Y}})$ and $U_{n,m}(\mathbf{X}, \mathbf{Y} \mathbf{W}_{n,m})$ with
$U_{n,m}(\pi(\hat{\mathbf{X}}), \pi(\hat{\mathbf{Y}}))$ and $U_{n,m}(\pi(\mathbf{X}),\pi(\mathbf{Y})
\mathbf{W}_{n,m})$, respectively, provided that 
that $0$ is not an atom of either $F$ or $G$ i.e., $F(\{0\}) = G(\{0\}) = 0$. 
The assumption that $0$ is not an atom is necessary,
because otherwise the problem is possibly ill-posed: specifically, $\pi(0)$ is
undefined. To contextualize the test of equality up to projection,
consider the very specific case of the degree-corrected stochastic
blockmodel \cite{karrer2011stochastic}. A degree-corrected stochastic
blockmodel can be regarded as a random dot product graph whose latent
position $X_v$ for an arbitrary vertex $v$ is of the form $X_v =
\theta_v \nu_{v}$ where $\nu_v$ is sampled from a mixture of point
masses and $\theta_v$ (the degree-correction factor) is sampled from a distribution on
$(0,1]$. Thus, given two degree-corrected stochastic blockmodel
graphs, equality up to projection tests whether the
underlying mixtures of point masses (that is, the distribution of the
$\nu_v$) are the same modulo the distribution of the degree-correction
factors $\theta_v$.



\section{Applications}\label{sec:Applications}
We begin this section by first presenting an application of the two-sample semiparametric test procedure in Section~\ref{subsec:Testing}, demonstrating how it
can be applied to compare data from a collection of neural images.
\subsection{Semiparametric testing for brain scan data}\label{subsec:semipar_cci}
We consider neural imaging graphs obtained from the test-retest
diffusion MRI and magnetization-prepared rapid acquisition gradient echo (MPRAGE) data of \cite{landman_KKI}. The raw data consist of 42 images: namely, one pair of neural
images from each of 21 subjects. These images are generated, in part, for the
purpose of evaluating scan-rescan reproducibility of the
MPRAGE image protocol. Table 5 from \cite{landman_KKI} indicates that the
variability of MPRAGE is quite small; specifically, the cortical gray
matter, cortical white matter, ventricular cerebrospinal fluid,
thalamus, putamen, caudate, cerebellar gray matter, cerebellar white
matter, and brainstem were identified with mean volume-wise
reproducibility of $3.5\%$, with the largest variability being that of
the ventricular cerebrospinal fluid at $11\%$.  

We use the MIGRAINE pipeline of \cite{gray_migraine} to convert these scans into spatially-aligned graphs, i.e., graphs in which each vertex corrresponds to a particular voxel in a reference coordinate system to which the image is registered. We first consider a
collection of small graphs on seventy vertices that are generated from an atlas of
seventy brain regions and the fibers connecting them.  Given these
graphs, we proceed to investigate the similarities and dissimilarities
between the scans.  We first embed each graph into
$\mathbb{R}^{4}$. We then test the hypothesis of equality up to
rotation between the graphs. Since Theorem~\ref{thm:identity} is a large-sample result, the rejection region specified therein might be excessively conservative for the graphs on $n = 70$ vertices in our current context. We remedy this issue by using the rejection region and $p$-values reported by the parametric bootstrapping procedure presented in Algorithm~\ref{bootstrap-simple}.

\begin{algorithm}[tb]
	\caption{Bootstrapping procedure for the test $\mathbb{H}_0 \colon \mathbf{X} =_{W} \mathbf{Y}$.} 
	\label{bootstrap-simple} 
	\begin{algorithmic}[1]
		\Procedure{Bootstrap}{$\mathbf{X}, T, bs$} \Comment{Returns the p-value associated with $T$.}
		\State $d \gets \mathrm{ncol}(\mathbf{X})$; \quad $\mathcal{S}_{X} \gets \emptyset$ \Comment{Set $d$ to be the number of columns of $\mathbf{X}$.}
		\For{$b \gets 1 \colon bs$}
		\State $\mathbf{A}_b \gets \mathrm{RDPG}(\hat{\mathbf{X}}); \quad \mathbf{B}_b \gets \mathrm{RDPG}(\hat{\mathbf{X}})$
		\State $\hat{\mathbf{X}}_b \gets \mathrm{ASE}(\mathbf{A}_b, d); \quad \hat{\mathbf{Y}}_b \gets \mathrm{ASE}(\mathbf{B}_b, d)$
		\State $T_b \gets \min_{\mathbf{W}} \| \hat{\mathbf{X}}_{b} - \hat{\mathbf{Y}}_{b} \mathbf{W} \|_{F}; \qquad \mathcal{S}_{X} \gets \mathcal{S}_X \cup T_b$    
		\EndFor
		\State \Return $p \gets (|\{s \in \mathcal{S}_X \colon s \geq T\}| + 0.5)/bs$ \Comment{Continuity correction.} 
		\EndProcedure \\
		\State $\hat{\mathbf{X}} \gets \mathrm{ASE}(\mathbf{A}, d)$; $\hat{\mathbf{Y}} \gets \mathrm{ASE}(\mathbf{B}, d)$ \Comment{The embedding dimension $d$ is assumed given.}
		\State $T \gets \min_{\mathbf{W}} \|\hat{\mathbf{X}} - \hat{\mathbf{Y}} \mathbf{W} \|_{F}$
		\State $p_{X} \gets \mathrm{Bootstrap}(\hat{\mathbf{X}}, T, bs)$ \Comment{The number of bootstrap samples $bs$ is assumed given.}
		\State $p_{Y} \gets \mathrm{Bootstrap}(\hat{\mathbf{Y}}, T, bs)$ 
		\State $p = \max\{p_X, p_Y\}$ \Comment{Returns the maximum of the two p-values.}
	\end{algorithmic}
\end{algorithm}

The pairwise comparisons between the $42$ graphs are
presented in Figure~\ref{fig:kki-small-identity}. Figure~\ref{fig:kki-small-identity} indicates that, in general, the
test procedure fails to reject the null hypothesis when the two graphs
are for the same subject.  This is consistent with the reproducibility
finding of \cite{landman_KKI}. Furthermore, this outcome is also
intuitively plausible; in addition to failing to reject when two scans
are from the same subject, we also frequently {\em do} reject the null
hypothesis when the two graphs are from scans of different
subjects. Note that our analysis is purely exploratory; as such, we do
not grapple with issues of multiple comparisons here.

\begin{figure}[h]
	\centering
	\includegraphics[width=0.65\textwidth]{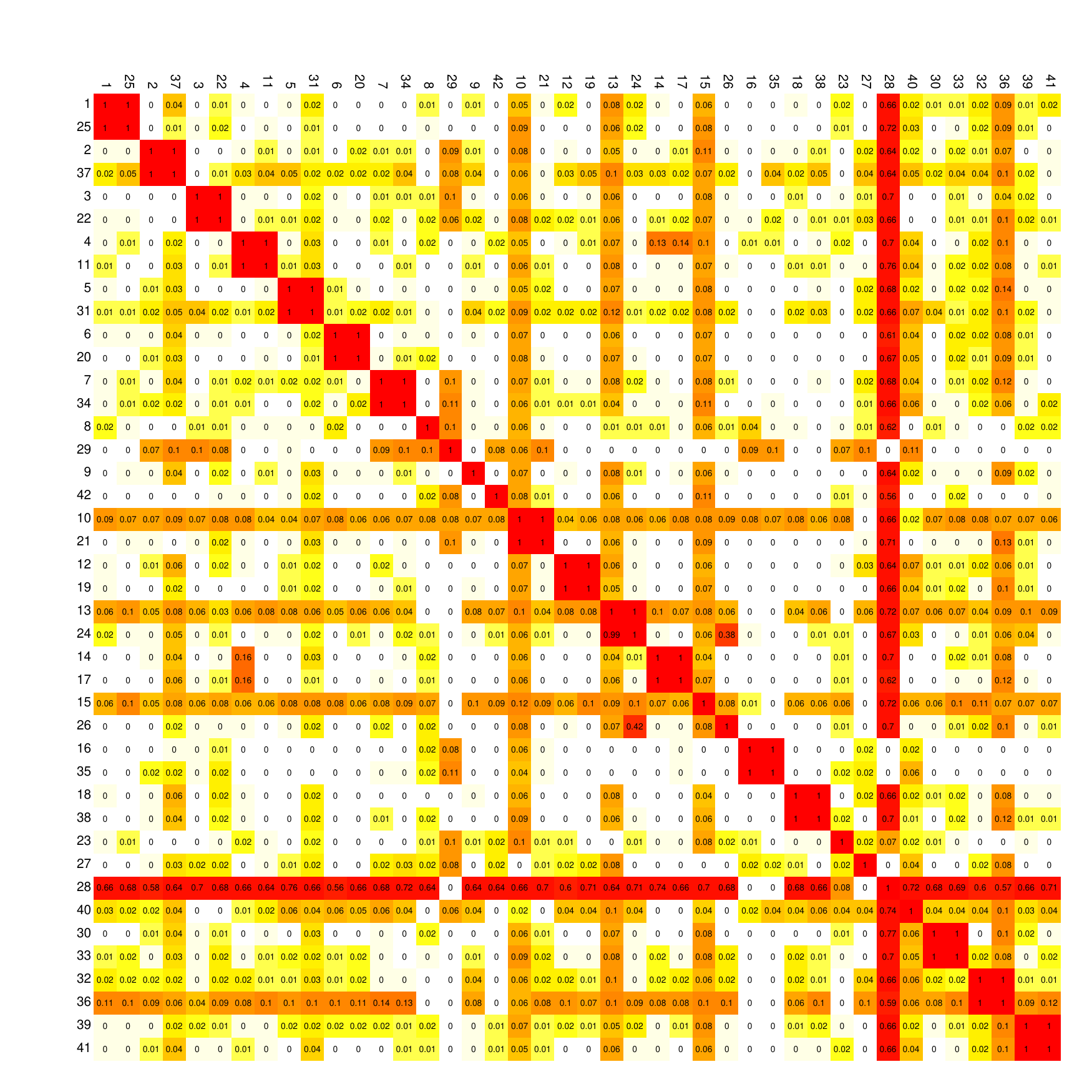}
	\caption{Matrix of p-values (uncorrected) for testing the hypothesis
		$\mathbb{H}_0 \colon \mathbf{X} =_{W} \mathbf{Y}$ for the $42
		\times 41/2$ pairs of graphs generated from the KKI test-retest
		dataset \cite{landman_KKI}. The labels had been arranged so
		that the pair $(2i-1,2i)$ correspond to scans from the same
		subject. The $p$-values are color coded to vary in intensity from
		white ($p$-value of $0$) to dark red ($p$-value of $1$). Figure duplicated from \cite{tang14:_semipar}.}
	\label{fig:kki-small-identity}
\end{figure}

Finally, we note that similar results also hold when we consider the large graphs generated from
these test-retest data through the MIGRAINE pipeline. In particular, for each magnetic resonance scan, the MIGRAINE
pipeline can generates graphs with upto $10^7$ vertices and $10^{10}$
edges with the vertices of all the graphs aligned. Bootstrapping the test statistics for these large
graphs present some practical difficulties; one procedure proposed in \cite{tang14:_semipar} is based on bootstrapping disjoint induced subgraphs of the original graphs using the bootstrapping procedure in Algorithm~\ref{bootstrap-simple} and combining the resulting $p$-values using Fisher's combined probability tests \cite{mosteller48:_quest_answer}. 

\subsection{Community detection and classification in hierarchical models}\label{subsec:HSBM}
In disciplines as diverse as social
network analysis and neuroscience, many large graphs are believed to
be composed of loosely connected communities and/orsmaller graph primitives, whose
structure is more amenable to analysis.
We emphasize that while the problem of community detection is very well-studied and there are an abundance of community detection algorithms, these algorithms have focused 
mostly on uncovering the subgraphs. Recently, however, the characterization and further
{\em classification} of these subgraphs into stochastically similar motifs
has emerged as an important area of ongoing research
The nonparametric two-sample hypothesis testing procedure in Section~\ref{subsec:nonpar} can be used in conjunction with spectral community detection algorithms to yield a robust, scalable, integrated methodology for {\em community detection} and
{\em community comparison} in graphs \cite{lyzinski15_HSBM}.

The notion of {\em hierarchical stochastic block model}---namely, a graph consisting of densely connected subcommunities which are themselves stochastic block models, with this structure iterated repeatedly---is precisely formulated in \cite{lyzinski15_HSBM}.  In that work, a novel angle-based clustering method is introduced, and this clustering method allows us to isolate appropriate subgraphs. We emphasize that the angle-based clustering procedure in \cite{lyzinski15_HSBM} is designed to identify a particular affinity structure within our hierarchical block model graph. Figure \ref{fig:hsbm_cluster} illustrates how  an angle-based clustering may differ from a $k$-means clustering.
\begin{figure}[tph]
	\centering
	\includegraphics[width=0.25\textwidth]{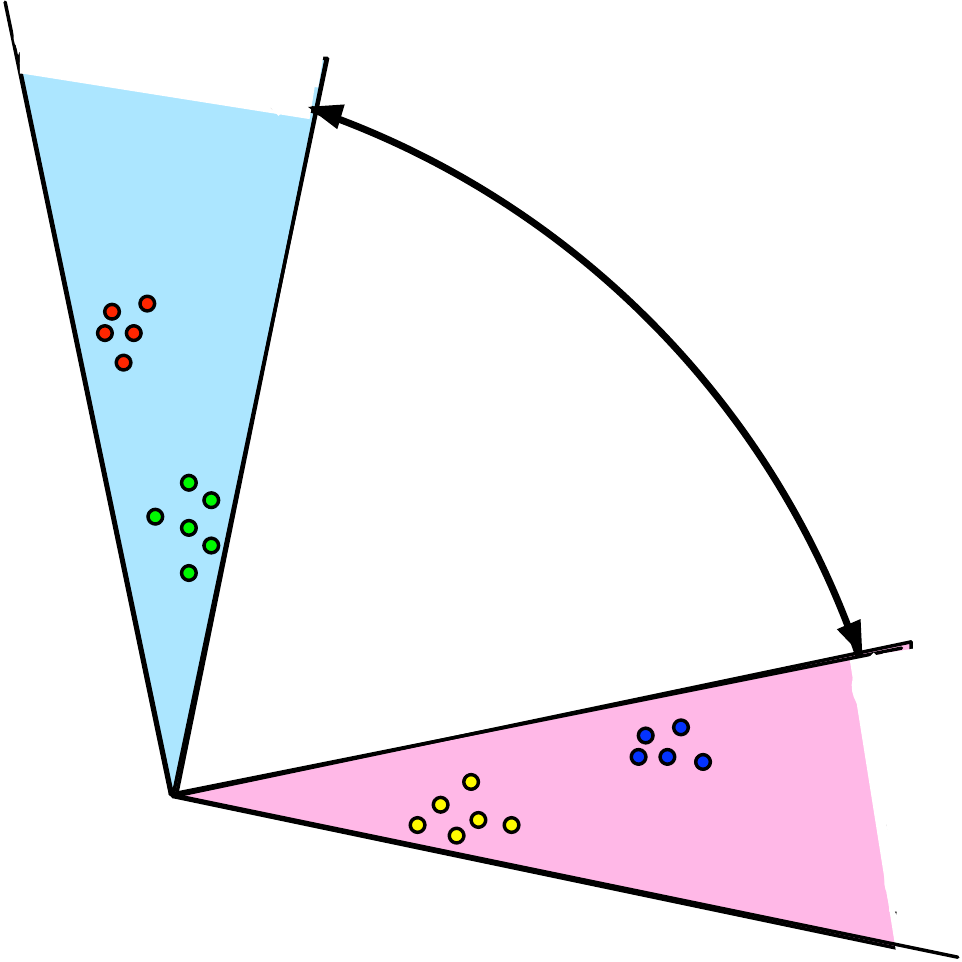}
	\caption{Subgraphs vs. angle-based clustering:  Note that if the fraction of points in the pink cone is sufficiently large, $K$-means clustering (with $K=2)$ will not cluster the vertices in the hierarchical SBM affinity model of \cite{lyzinski15_HSBM} into the appropriate subgraphs. Figure duplicated from \cite{lyzinski15_HSBM}. }
	\label{fig:hsbm_cluster}
\end{figure}

Our overall community detection algorithm is summarized in Algorithm~\ref{alg:main}.
As an illustrative example of this methodology, we present an analysis of the
communities in the Friendster social network. The Friendster social network
contains roughly $60$ million users and $2$ billion
connections/edges.
In addition, there are roughly $1$ million communities at the local scale.
Because we expect the social interactions in these communities to inform the function of the different communities, we expect to observe distributional repetition among the graphs associated with these communities.

\begin{algorithm}[t!]
  \begin{algorithmic}
    \State \textbf{Input}: Adjacency matrix $A\in
    \{0,1\}^{n\times n}$ for a latent position random graph.
    \State \textbf{Output}: Subgraphs and characterization of their dissimilarity
\While {Cluster size exceeds threshold}
\State {\em Step 1}: 
Compute the adjacency spectral embedding $\hat{\mathbf{X}}$ of $A$ into $\mathbb{R}^{D}$;
   \State {\em Step 2}: Cluster $\hat{\mathbf{X}}$ to obtain subgraphs $\hat{H}_1,
    \cdots, \hat{H}_R$ using a novel angle-based clustering procedure given in \cite{lyzinski15_HSBM}.
    \State {\em Step 3}: For each $i\in[R],$ compute the adjacency
    spectral embedding for each subgraph
    $\hat{H}_i$ into $\mathbb{R}^{d}$, obtaining $\hat{\mathbf{X}}_{\hat{H}_i}$; 
    \State \textit{Step 4}: Compute $\widehat S:=[U_{\hat n_r,\hat n_s}(\hat{\mathbf{X}}_{\hat{H}_r}, \hat{\mathbf{X}}_{\hat{H}_s})]$, where $U$ is the test statistic in Theorem~\ref{thm:mmd_unbiased_ase}, producing a pairwise dissimilarity matrix on induced subgraphs;
\State \textit{Step 5}: Cluster induced subgraphs into motifs using the dissimilarities given in $\widehat S$; e.g., use a hierarchical clustering algorithm to cluster the rows of $\widehat{S}$ or the matrix of associated $p$-values.
\State \textit{Step 6}: Recurse on a representative subgraph from each motif (e.g., the largest subgraph), embedding into $\mathbb{R}^d$ in Step 1 (not $\mathbb{R}^D$);
\EndWhile
\end{algorithmic}
\caption{Detecting hierarchical structure for graphs}
\label{alg:main}
\end{algorithm} 
 
Implementing Algorithm \ref{alg:main} on the very large Friendster
graph presents several computational challenges and model selection quagmires. 
To overcome the computational challenge
in scalability, we use the specialized SSD-based graph processing
engine \texttt{FlashGraph} \cite{zheng_flashgraph}, which is
designed to analyze graphs with billions of nodes.  With
\texttt{FlashGraph}, we adjacency spectral embed the Friendster
adjacency matrix into $\mathbb{R}^{14}$---where $\widehat D=14$ is
chosen using universal singular value thresholding on the partial SCREE plot \cite{chatterjee2015}. 
We next cluster the embedded points into $\widehat R=15$ large-scale/coarse-grained clusters ranging
in size from $10^6$ to 15.6 million vertices (note that to alleviate sparsity concerns, we projected the embedding onto the sphere before clustering); After re-embedding the induced subgraphs associated with these $15$ clusters, we use a
linear time estimate of the test statistic $U$ to compute $\widehat
S$, the matrix of estimated pairwise dissimilarities among the
subgraphs. 
See Figure
\ref{fig:friend} for a heat map depicting
$\widehat{S}\in\mathbb{R}^{15\times 15}$.  In the heat map, the
similarity of the communities is represented on the spectrum between
white and red, with white representing highly similar communities and
red representing highly dissimilar communities.
From the figure, we can see clear repetition in the subgraph distributions; for example, we see a  repeated motif including subgraphs $\{\hat{H}_5, \hat{H}_4,\hat{H}_3,\hat{H}_2\}$ and a clear motif including subgraphs $\{\hat{H}_{10},\hat{H}_{12},\hat{H}_9\}$.

Formalizing the motif detection step, we employ hierarchical
clustering to cluster $\widehat S$ into motifs; see Figure
\ref{fig:friend} for the corresponding hierarchical clustering
dendrogram, which suggests that our algorithm does in fact uncover
repeated motif structure at the coarse-grained level in the Friendster
graph.  While it may be difficult to draw meaningful inference from
repeated motifs at the scale of hundreds of thousands to millions of
vertices, if these motifs are capturing a common HSBM structure within
the subgraphs in the motif, then we can employ our algorithm
recursively on each motif to tease out further hierarchical structure.
\begin{figure}[t!]
  \centering
  \includegraphics[width=0.4\textwidth]{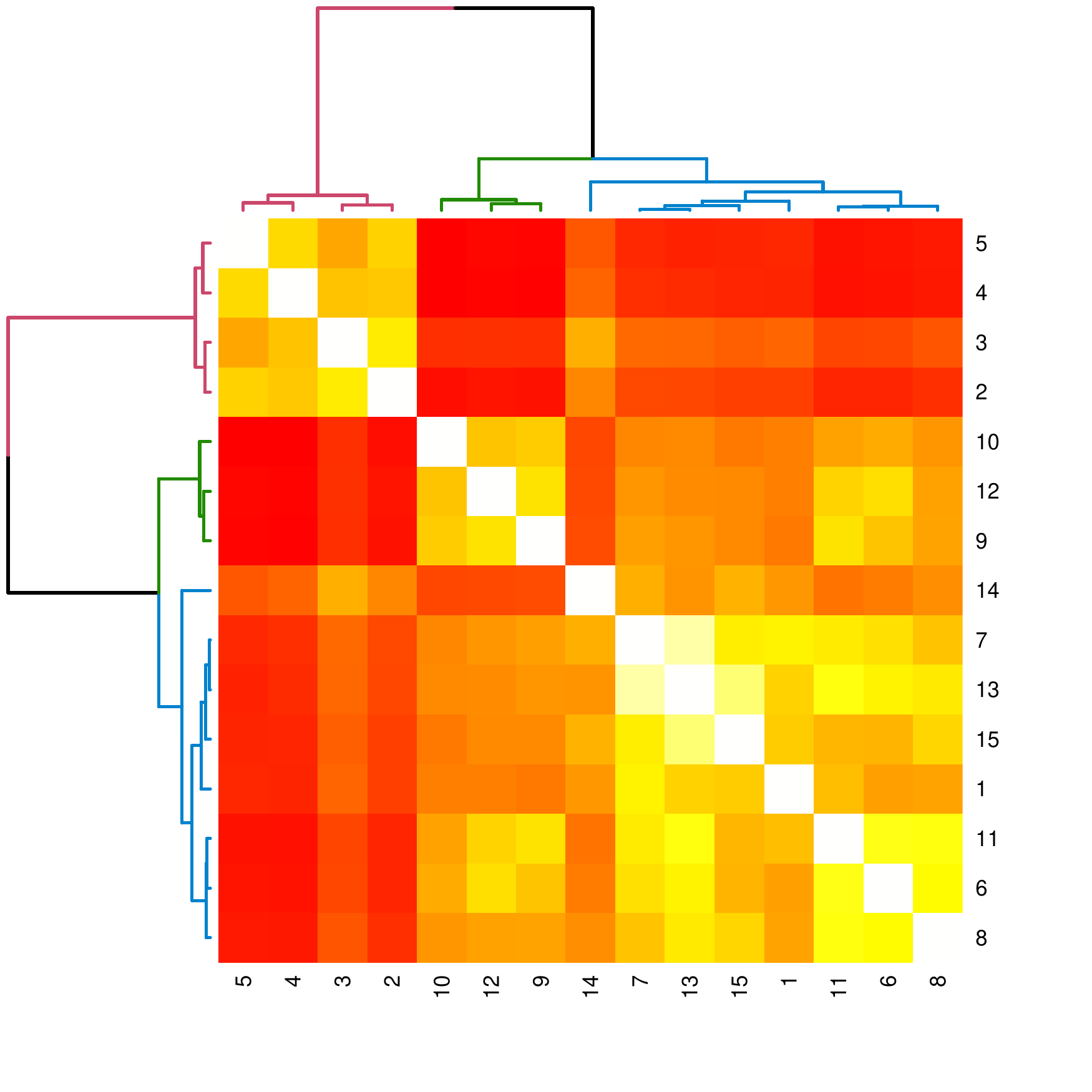}
\caption{Heat map depiction of the level one Friendster estimated dissimilarity matrix $\widehat{S}\in\mathbb{R}^{15\times 15}$.  In the heat map, the similarity of the communities is represented on the spectrum between white and red, with white representing highly similar communities and red representing highly dissimilar communities.  
  In addition, we cluster $\widehat S$ using hierarchical clustering and display the associated hierarchical clustering dendrogram. Figure duplicated from \cite{lyzinski15_HSBM}.}
  \label{fig:friend}
\end{figure}

Exploring this further, we consider three subgraphs $\{\hat{
H}_2,\hat{H}_{8},\hat{H}_{15}\}$, two of which are in the same
motif (8 and 15) and both differing significantly from subgraph 2
according to $\widehat S.$ 
We embed these subgraphs into
$\mathbb{R}^{26}$ ($26$ were chosen once again using the universal singular value thresholding of \cite{chatterjee2015})
perform a Procrustes alignment of the vertex sets of the three
subgraphs, cluster each into $4$ clusters ($4$ chosen
to optimize silhouette width in $k$-means clustering), and estimate both the block connection probability matrices, 
$$\hat P_2=\begin{bmatrix}
0.000045& 0.00080& 0.00056& 0.00047\\
0.00080& 0.025& 0.0096& 0.0072\\
0.00057& 0.0096& 0.012& 0.0067\\
0.00047& 0.0072& 0.0067& 0.023
\end{bmatrix},$$
$$\hat P_8=\begin{bmatrix}
0.0000022& 0.000031& 0.000071& 0.000087\\
0.000031& 0.0097& 0.00046& 0.00020\\
0.000071& 0.00046& 0.0072& 0.0030\\
0.000087& 0.00020& 0.003& 0.016
\end{bmatrix},$$
$$\hat
P_{15}=\begin{bmatrix}
0.0000055& 0.00011& 0.000081& 0.000074\\
0.00011& 0.014& 0.0016& 0.00031\\
0.000081& 0.0016 &0.0065& 0.0022\\
0.000074& 0.00031& 0.0022& 0.019
\end{bmatrix},$$
 and the block membership probabilities $\hat \pi_2,\,\hat \pi_8,\,\hat
\pi_{15},$ for each of the three graphs.  We calculate
\begin{align*}
\|\hat P_2-\hat P_8\|_F=0.033; &\hspace{3mm}  \|\hat \pi_2-\hat \pi_8\|=0.043 ;\\
\|\hat P_8-\hat P_{15}\|_F=0.0058;  &\hspace{3mm}   \|\hat \pi_8-\hat \pi_{15}\|=0.0010; \\
\|\hat P_2-\hat P_{15}\|_F= 0.027; &\hspace{3mm}  \|\hat \pi_2-\hat \pi_{15}\|= 0.043;
\end{align*}
which suggests that the repeated structure our algorithm
uncovers is {\it SBM substructure}, thus ensuring that we can proceed to
apply our algorithm recursively to the subsequent motifs.

As a final point, we recursively apply Algorithm \ref{alg:main} to the
subgraph $\hat{H}_{11}$.  We first embed the graph into
$\mathbb{R}^{26}$ (again, with $26$ chosen via universal singular value thresholding). We then cluster the vertices into
$\widehat R=13$ large-scale clusters of sizes ranging from 500K to
2.7M vertices.  We then use a linear time estimate of the test
statistic $T$ to compute $\widehat S$ (see Figure \ref{fig:friend2}),
and note that there appear to be clear repeated motifs (for example,
subgraphs 8 and 12) among the $\widehat H$'s.  We run hierarchical
clustering to cluster the $13$ subgraphs, and note that the associated
dendrogram---as shown in Figure \ref{fig:friend2}---shows that our
algorithm again uncovered some repeated level-$2$ structure in the
Friendster network.  We can, of course, recursively apply our
algorithm still further to tease out the motif structure at
increasingly fine-grained scale.

Ideally, when recursively running Algorithm \ref{alg:main}, we would
like to simultaneously embed and cluster all subgraphs in the motif.
In addition to potentially reducing embedding variance, 
being able to efficiently simultaneously embed all the subgraphs in a
motif could greatly increase algorithmic scalability in large networks
with a very large number of communities at local-scale.  In order to
do this, we need to understand the nature of the repeated structure
within the motifs.  This repeated structure can inform an estimation
of a motif average (an averaging of the subgraphs within the motif),
which can then be embedded into an appropriate Euclidean space in lieu
of embedding all of the subgraphs in the motif separately.  However,
this averaging presents several novel challenges, as these subgraphs
may be of very different orders and may be errorfully obtained, which
could lead to compounded errors in the averaging step.  
 \begin{figure*}[t!]
  \centering
  \includegraphics[width=0.55\textwidth]{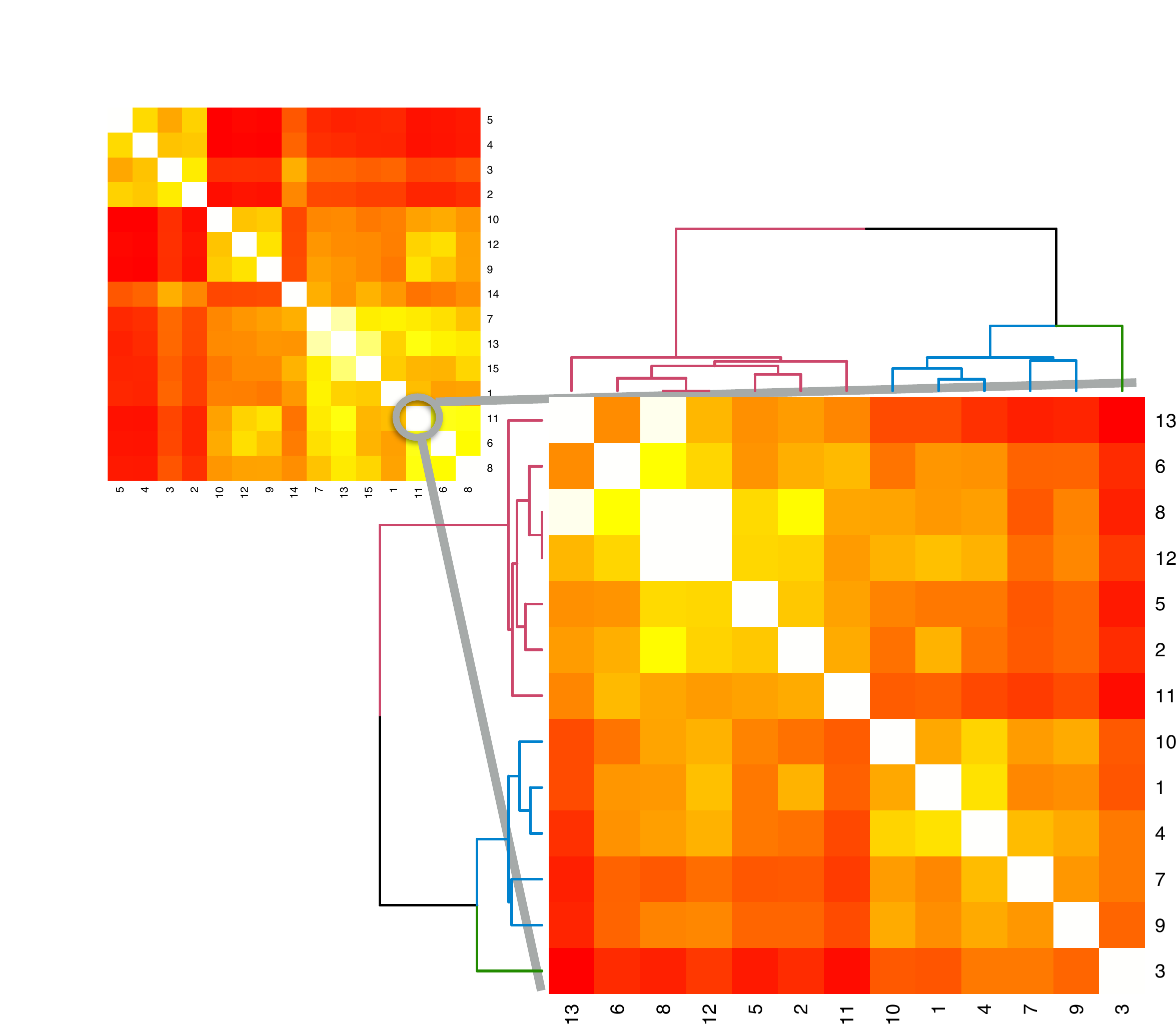}
  \caption{Heat map depiction of the level two Friendster estimated dissimilarity matrix $\widehat{S}\in\mathbb{R}^{13\times 13}$ of $\widehat H_{11}$.  In the heat map, the similarity of the communities is represented on the spectrum between white and red, with white representing highly similar communities and red representing highly dissimilar communities.  
  In addition, we cluster $\widehat S$ using hierarchical clustering and display the associated hierarchical clustering dendrogram.}
  \label{fig:friend2}
\end{figure*}
\subsection{Structure discovery in the {\em Drosophila} connectome} \label{subsec:MBStructure}
In this subsection, we address a cutting-edge application of our techniques to neuroscience:  structure discovery in the larval Drosophila connectome, comprehensively described in \cite{MBStructure}, and from which significant portions are reproduced here, with permission. This is a first-of-its-kind exploratory data analysis of a newly-available wiring diagram, and although the connectome graph we analyze is directed, weighted, and also of unknown embedding dimension, our statistical techniques can nonetheless be adapted to this setting. 

Specifically, we introduce the {\it latent structure model} (LSM) for network modeling and inference. The LSM is a generalization of the stochastic block model (SBM) in that the latent positions are allowed to lie on a lower-dimensional curve, and this model is amenable to semiparametric Gaussian mixture modeling (GMM) applied to the adjacency spectral embedding (ASE).
The resulting {\it connectome code}
derived via semiparametric GMM composed with ASE, which we denote, in shorthand, by $GMM \circ ASE$,
captures latent connectome structure
and elucidates biologically relevant neuronal properties.

HHMI Janelia has
recently reconstructed the complete wiring diagram of the higher order parallel fiber system for associative learning in the larval {\it Drosophila} brain, 
the mushroom body (MB).
Memories are thought to be stored as functional and structural changes in connections between neurons, 
but the complete circuit architecture of a higher-order learning center involved in memory formation or storage
has not been known in any organism---until now, that is.  Our MB connectome
was obtained via serial section transmission electron microscopy of an entire larval {\it Drosophila} nervous system
\cite{ohyama2015multilevel,schneider2016quantitative}.
This connectome contains the entirety of MB intrinsic neurons, called Kenyon cells, and all of their pre- and post-synaptic partners
\cite{EichlerSubmitted}. 

We consider the right hemisphere MB. The connectome consists of four distinct types of neurons
-- Kenyon Cells (KC), Input Neurons (MBIN), Output Neurons (MBON), Projection Neurons (PN) --
with directed connectivity illustrated in Figure \ref{fig:Fig1}.
There are $n=213$ neurons\footnote{
	There are 13 isolates, all are KC; removing these isolates makes the (directed) graph one (weakly, but not strongly) connected component with 213 vertices and 7536 directed edges.},
with
$n_{KC} = 100$,
$n_{MBIN} = 21$,
$n_{MBON} = 29$, and
$n_{PN} = 63$.
Figure \ref{fig:Fig2} displays the observed MB connectome as an adjacency matrix.
Note that, in accordance with Figure \ref{fig:Fig1},
Figure \ref{fig:Fig2} shows data (edges) in only eight of the 16 blocks.

\begin{figure}[htbp]
	\centering
	\includegraphics[width=0.65\textwidth]{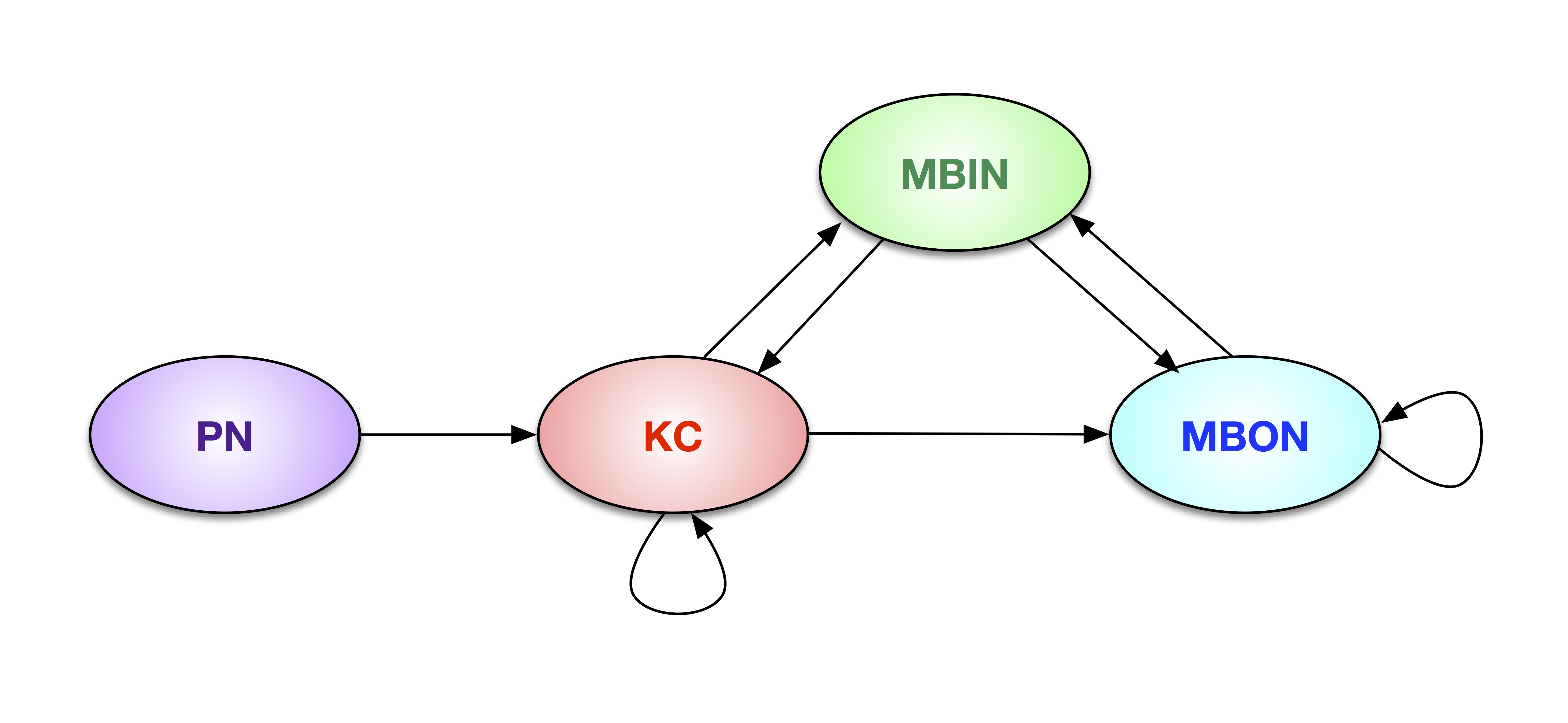}
	\caption{\label{fig:Fig1} Illustration of the larval {\it Drosophila} mushroom body connectome as a directed graph on four neuron types. Figure duplicated from \cite{MBStructure}.}
\end{figure}

\begin{figure}[htbp]
	\centering
	\includegraphics[width=0.65\textwidth]{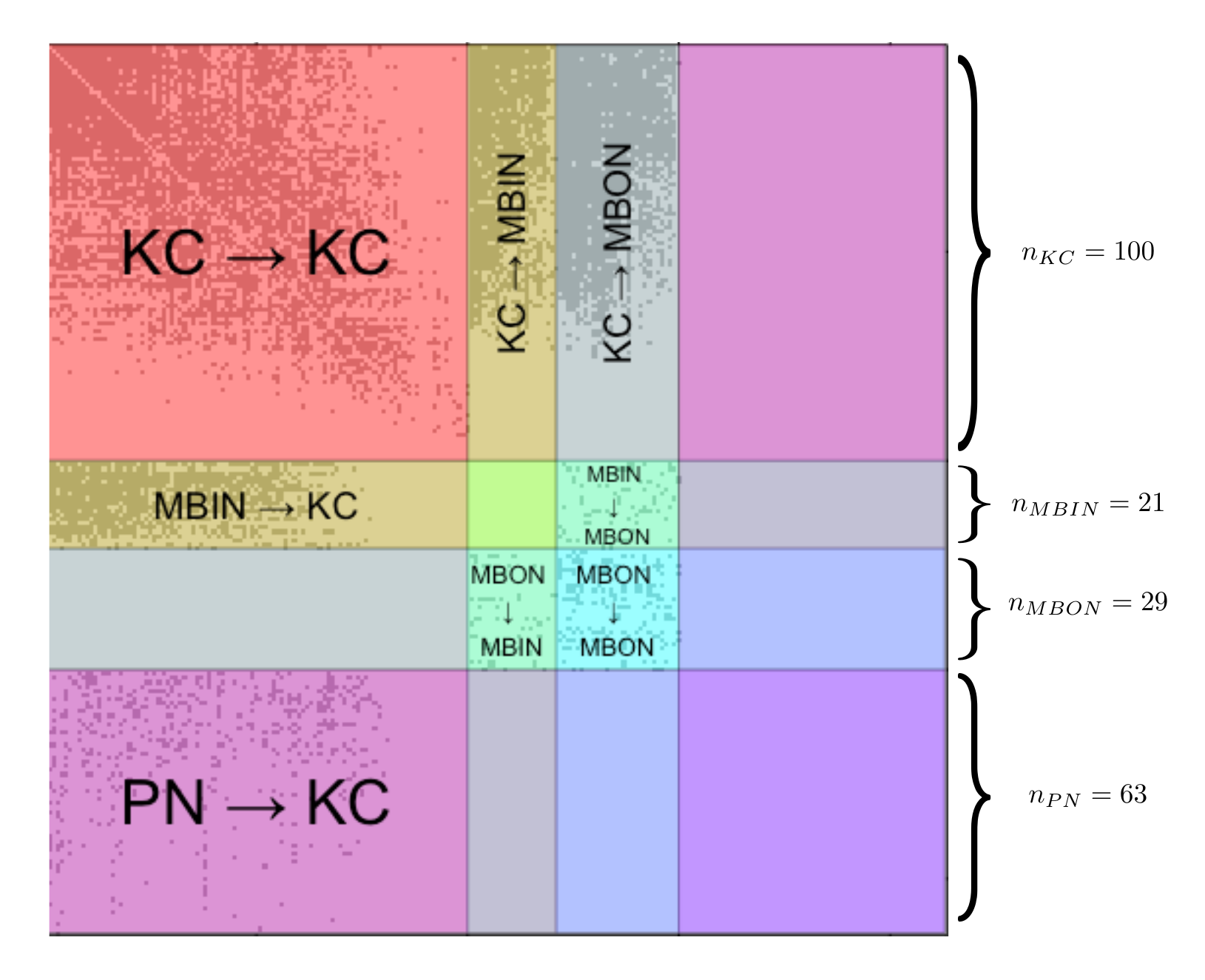}
	\caption{\label{fig:Fig2} Observed data for our MB connectome as a directed adjacency matrix on four neuron types with 213 vertices 
		($n_{KC} = 100$,
		$n_{MBIN} = 21$,
		$n_{MBON} = 29$, and
		$n_{PN} = 63$)
		and 7536 directed edges.
		(This data matrix is available at
		$<$\url{http://www.cis.jhu.edu/~parky/MBstructure.html}$>$.)
	Figure duplicated from \cite{MBStructure}.}
	\end{figure}
	
	Due to its undeniable four-neuron-type connectivity structure,
	we might think of our MB connectome,
	to first order,
	as an observation from a $(K=4)$-block directed stochastic block model (SBM) \cite{WangWong} on $n$ vertices.
	This model is parameterized by ({\it i}) a block membership probability vector $\rho = [\rho_1,\cdots,\rho_K]$
	such that $\rho_k \ge 0$ for all $k$ and $\sum_k \rho_k = 1$
	and ({\it ii}) a $K \times K$ block connectivity probability matrix $B$ with entries $B_{k_1,k_2} \in [0,1]$
	governing the probability of directed edges from vertices in block $k_1$ to vertices in block $k_2$.
	For this model of our MB connectome we have
	\[B = \left[ \begin{array}{cccc}
	B_{11} & B_{12} & B_{13} & 0 \\
	B_{21} & 0      & B_{23} & 0 \\
	0      & B_{32} & B_{33} & 0 \\
	B_{41} & 0 & 0 & 0 \end{array} \right]\] 
	where the $0$ in the $B_{31}$ entry, for example,
	indicates that there are no directed connections from any MBON neuron to any KC neuron (as seen in Figures \ref{fig:Fig1} and \ref{fig:Fig2}).
	
	Theoretical results and methodological advances suggest that Gaussian mixture modeling
	(see, for example, \cite{mclust2012})
	composed with adjacency spectral embedding,
	denoted $\mclustase$, can be instructive in analysis of the (directed) SBM.

	Since this graph is directed, we adapt our embedding technique just slightly. Given $d \geq 1$, the adjacency spectral embedding (ASE) of a {\em directed} graph 
	on $n$ vertices
	employs the singular value decomposition 
	to represent the $n \times n$ adjacency matrix via ${\bf A} = \bigl[{\bf U} \mid {\bf U}^{\perp} \bigr] \bigl[{\bf S} \bigoplus {\bf S}^{\perp}\bigr]
	\bigl[{\bf V} \mid {\bf V}^{\perp}]^{\top}$
	where $\mathbf{S}$ is the $d \times d$ diagonal matrix of the $d$ largest singular values and ${\bf U}$ and ${\bf V}$ are the matrix of corresponding left and right singular vectors, and we embed the graph as $n$ points in $\mathbb{R}^{2d}$ via the concatenation
	$$\Xhat = \left[ {\bU} {\bS}^{1/2} \, \mid \, {\bV} {\bS} ^{1/2} \right] \in \mathbb{R}^{n \times 2d}.$$
	(The scaled left-singular vectors $\bU \bS^{1/2}$ are interpreted as the ``out-vector'' representation of the digraph,
	modeling vertices' propensity to originate directed edges;
	similarly, $\bV \bS^{1/2}$ are interpreted as the ``in-vectors''.)
	Gaussian mixture modeling (GMM)
	then fits a $K$-component $2d$-dimensional Gaussian mixture model
	to the points $\Xhat_1,\cdots,\Xhat_n$ given by the rows of $\Xhat$.
	If the graph is a stochastic block model, then as we have discussed previously in Section \ref{subsec:Distributional}, clustering the rows of the adjacency spectral embedding via Gaussian mixture modeling gives us consistent estimates for the latent positions.
	
	Applying this procedure to our MB connectome
	yields the clustered embedding depicted via the pairs plot presented in Figure~\ref{fig:MBpairs12},
	with the associated cluster confusion matrix with respect to true neuron types presented in Table \ref{tab:mclust6}.
	The clusters are clearly coincident with the four true neuron types.
	(For ease of illustration, Figure \ref{fig:MBpairs12} presents just the Out1 vs.\ Out2 subspace.)
	

	\begin{figure}[htbp]
		\centering
		\includegraphics[width=0.65\textwidth]{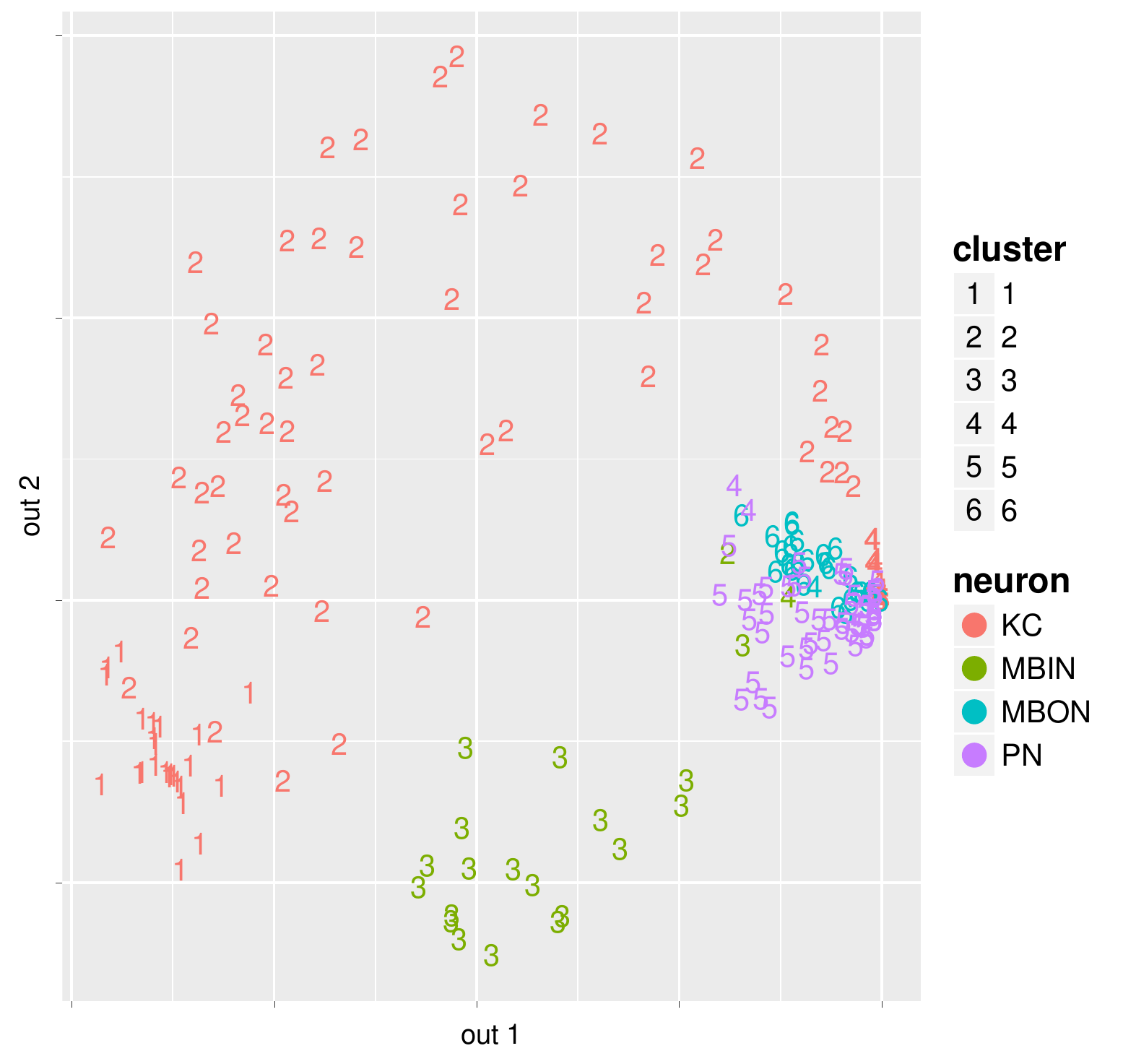}
		\caption{\label{fig:MBpairs12}
			Plot for the clustered embedding of our MB connectome 
			in the Out1 vs.\ Out2 dimensions.
			For ease of illustration,
			we present embedding results in this two-dimensional subspace.
			Recall that this is a two-dimensional visualization of six-dimensional structure. Figure duplicated from \cite{MBStructure}.
		}
	\end{figure}
	
	There are two model selection problems inherent in spectral clustering in general,
	and in obtaining our clustered embedding (Figure \ref{fig:MBpairs12}) in particular:
	choice of embedding dimension ($\dhat$), and choice of mixture complexity ($\Khat$).
	A ubiquitous and principled method for choosing the number of dimensions in eigendecompositions and singular value decompositions
	is to examine the so-called {\em scree plot}
	and look for ``elbows'' or ``knees'' defining the cut-off between the top signal dimensions and the noise dimensions.
	Identifying a ``best'' method is, in general, impossible, as the bias-variance tradeoff demonstrates that,
	for small $n$, subsequent inference may be optimized by choosing a dimension {\it smaller than} the true signal dimension;
	see Section 3 of \cite{JainDuinMao} for a clear and concise illustration of this phenomenon. 
	There are a plethora of variations for automating this singular value thresholding (SVT);
	Section 2.8 of \cite{Jackson} provides a comprehensive discussion in the context of principal components,
	and 
	\cite{chatterjee2015}
	provides a theoretically-justified (but perhaps practically suspect, for small $n$) universal SVT.
	Using the profile likelihood SVT method of \cite{zhu06:_autom} 
	yields a cut-off at three singular values, as depicted in Figure \ref{fig:MBscree}.
	Because this is a directed graph, we have both left- and right-singular vectors for each vertex;
	thus the SVT choice of three singular values results in $\dhat=6$.
	
	Similarly, a ubiquitous and principled method for choosing the number of clusters in,
	for example, Gaussian mixture models,
	is to maximize a fitness criterion penalized by model complexity.
Common approaches include
Akaike Information Criterion (AIC) \cite{akaike1974new},
Bayesian Information Criterion (BIC) \cite{BIC},
and Minimum Description Length (MDL) \cite{MDL},
to name a few.
Again, identifying a ``best'' method is, in general, impossible, as the bias-variance tradeoff demonstrates that,
for small $n$, inference performance may be optimized by choosing a number of clusters {\it smaller than} the true cluster complexity. The MCLUST algorithm \cite{mclust2012}, as implemented in \texttt{R}, and its associated BIC
applied to our MB connectome embedded via ASE into $\mathbb{R}^{\dhat=6}$,
is maximized at six clusters,
as depicted in Figure \ref{fig:MBbic}, and hence $\Khat=6$.
\begin{figure}[htbp]
	\centering
	\includegraphics[width=0.370\textwidth]{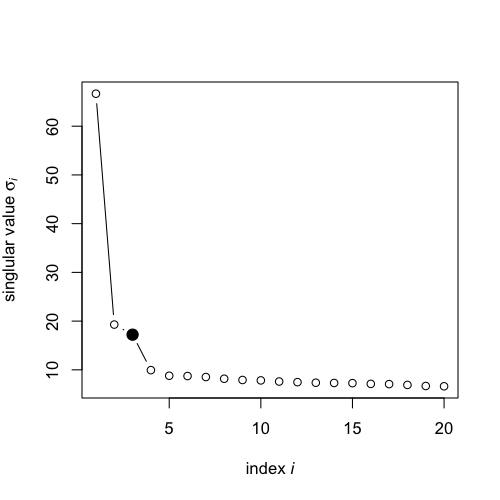}
	\caption{\label{fig:MBscree}
		Model Selection: embedding dimension $\dhat=6$ --
		the top 3 singular values and their associated left- and right-singular vectors --
		is chosen by SVT. Figure duplicated from \cite{MBStructure}.}
\end{figure}

\begin{figure}[htbp]
	\centering
	\includegraphics[width=0.715\textwidth]{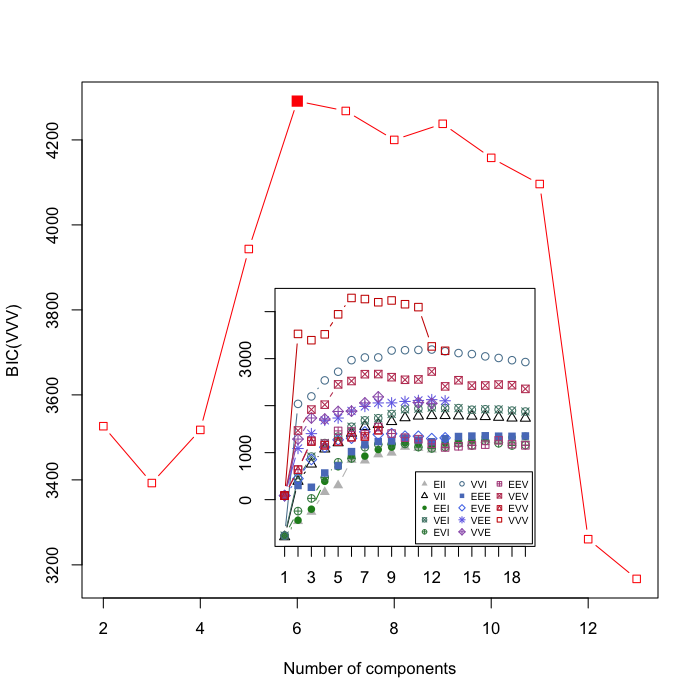}
	\caption{\label{fig:MBbic}
		Model Selection: mixture complexity $\Khat=6$ is chosen by BIC.
		(The inset shows that the main curve -- 
		BIC for dimensions 2 through 13 for MCLUST's most general covariance model, in red --
		dominates all other dimensions and all other models.) Figure duplicated from \cite{MBStructure}.
	}
\end{figure}

\begin{table}[htbp]
	\centering
	\begin{tabular}{lrrrrrrr}
		& & 1 & 2 & 3 & 4 & 5 & 6 \\
		\text{KC}   & & {\bf 25} & {\bf 57} &  0 & {\bf 16} &  2 &  0 \\
		\text{MBIN} & &  0 &  1 & {\bf 19} &  1 &  0 &  0 \\
		\text{MBON} & &  0 &  0 &  0 &  1 &  0 & {\bf 28} \\
		\text{PN}   & &  0 &  0 &  0 &  2 & {\bf 61} &  0 \\
	\end{tabular}
	\caption{\label{tab:mclust6}
		$\mclustase$ for our MB connectome yields $\Khat=6$ clusters.
		The clusters are clearly coherent with but not perfectly aligned with
		the four true neuron types, as presented in this confusion matrix.
	}
\end{table}

While BIC chooses $\Khat=6$ clusters,
it is natural to ask whether 
the distribution of KC across multiple clusters
is an artifact of insufficiently parsimonious model selection.
However, choosing four or five clusters not only (substantially) decreases BIC,
but in fact leaves KC distributed across multiple clusters while splitting and/or merging other neuron types.
In the direction of less parsimony,
Figure \ref{fig:MBbic} suggests that any choice from 7 to 11 clusters is competitive, in terms of BIC, with the maximizer $\Khat=6$.
Moreover, any of these choices only slightly decreases BIC,
while leaving PN, MBIN, and MBON clustered (mostly) singularly and (mostly) purely
and distributing KC across more clusters.
Tables
\ref{tab:mclust4},
\ref{tab:mclust5}, and
\ref{tab:mclust7}
show cluster confusion matrices for other choices of $K$.

\begin{table}[htbp]
	\centering
	\begin{tabular}{lrrrrr}
		\hline
		& & 1 & 2 & 3 & 4 \\
		\hline
		\text{KC}    & & 26 & 	56 & 	16 & 	2 \\
		\text{MBIN}  & &  0 & 	20 & 	1 & 	0 \\
		\text{MBON}  & &  0 & 	28 & 	1 & 	0 \\
		\text{PN}    & &  0 & 	0 & 	16 & 	47 \\
	\end{tabular}
	\caption{\label{tab:mclust4}
		Cluster confusion matrix for $\mclustase$ with 4 clusters.
		Choosing four or five clusters not only (substantially) decreases BIC (compared to $\Khat=6$),
		but in fact leaves KC distributed across multiple clusters
		while splitting and/or merging other neuron types.}
\end{table}

\begin{table}[htbp]
	\centering
	\begin{tabular}{lrrrrrr}
		\hline
		& & 1 & 2 & 3 & 4 & 5\\
		\hline
		\text{KC}    & & 26 &	56 &	16 &	2 &	0 \\
		\text{MBIN}  & & 0 &	20 &	1 &	0 &	0 \\
		\text{MBON}  & & 0 &	0 &	1 &	0 &	28 \\
		\text{PN}    & & 0 &	0 &	16 &	47 &	0 \\
	\end{tabular}
	\caption{\label{tab:mclust5}
		Cluster confusion matrix for $\mclustase$ with 5 clusters.
		Choosing four or five clusters not only (substantially) decreases BIC (compared to $\Khat=6$),
		but in fact leaves KC distributed across multiple clusters
		while splitting and/or merging other neuron types.}
\end{table}

\begin{table}[htbp]
	\centering
	\begin{tabular}{lrrrrrrrr}
		\hline
		& & 1 & 2 & 3 & 4 & 5 & 6 & 7\\
		\hline
		\text{KC}    & & 25 &	42 &	15 &	0 &	16 &	2 &	0 \\
		\text{MBIN}  & & 0 &	0 &	1 &	19 &	1 &	0 &	0 \\
		\text{MBON}  & & 0 &	0 &	0 &	0 &	1 &	0 &	28 \\
		\text{PN}    & & 0 &	0 &	0 &	0 &	2 &	61 &	0 \\
	\end{tabular}
	\caption{\label{tab:mclust7}
		Cluster confusion matrix for $\mclustase$ with 7 clusters.
		Any choice from 7 to 11 clusters only slightly decreases BIC (compared to $\Khat=6$),
		while leaving PN, MBIN, and MBON clustered (mostly) singularly and (mostly) purely
		and distributing KC across more clusters.}
\end{table}

\begin{figure}[htbp]
	\centering
	\includegraphics[width=0.65\textwidth]{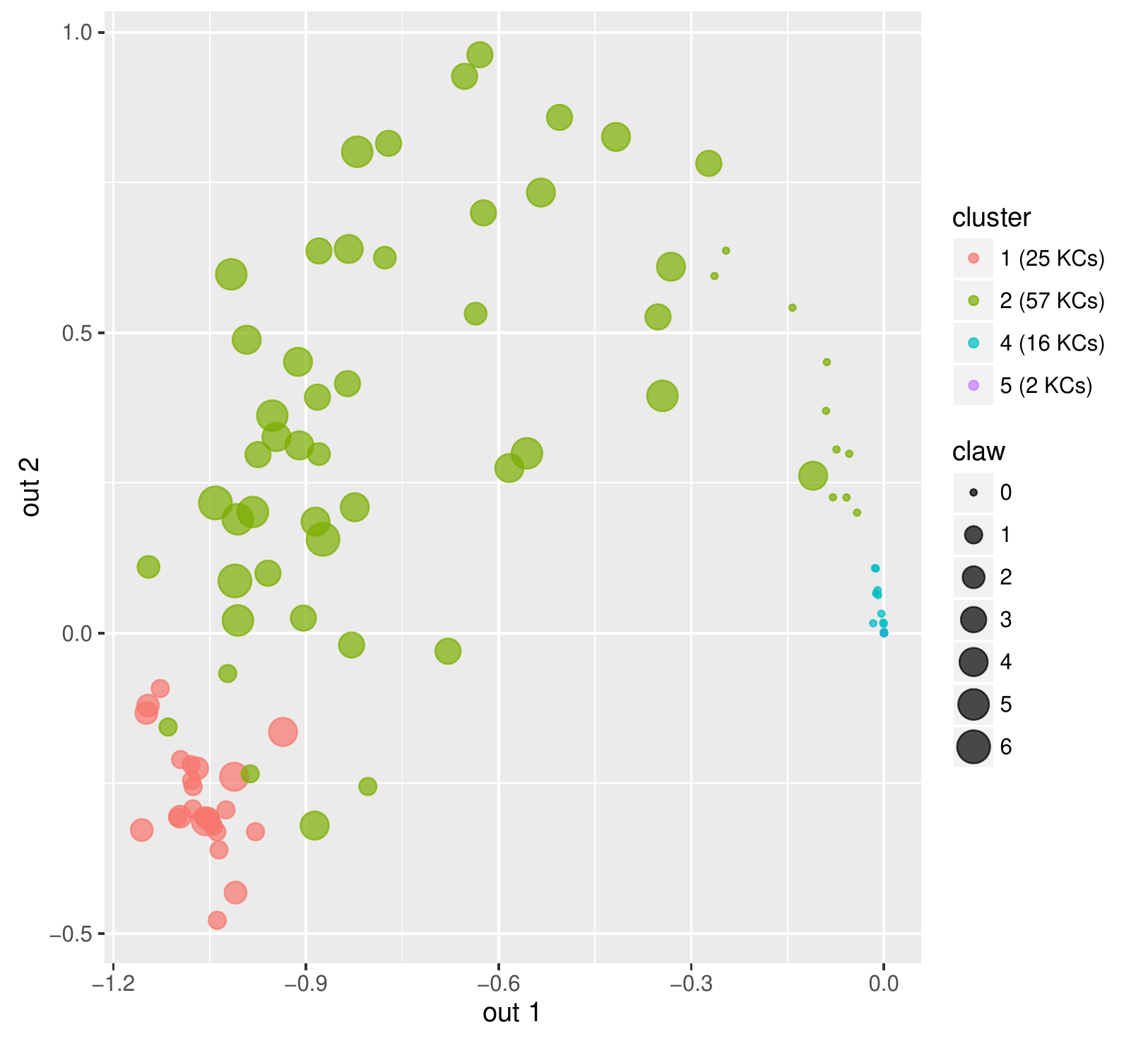}
	\caption{\label{fig:MBKCageclaw}
		The multiple clusters for the KC neurons are capturing neuron age.
		Depicted are the first two dimensions for the KC neuron out-vectors,
		with color representing $\Khat=6$ cluster membership --
		recall from Table \ref{tab:mclust6} that the $n_{KC}=100$ KCs are distributed across multiple clusters,
		with 25 neurons in cluster \#1, 
		57 in \#2,
		0 in \#3,
		16 in \#4,
		2 in \#5, and
		0 in \#6.
		The size of the dots represent the number of claws associated with the neurons.
		We see from the scatter plot that the embedded KC neurons arc
		from oldest (one-claw, lower left, cluster 1, in red),
		up and younger (more claws) through cluster 2 in green, and
		back down to youngest (zero-claw, clusters 4 and 5).
		See also Table \ref{tab:MBKCageclaw}.
		Recall that this is a two-dimensional visualization of six-dimensional structure. Figure duplicated from \cite{MBStructure}.
	}
\end{figure}

\begin{table}[htbp]
	\centering
	\begin{tabular}{lrrrrrrr}
		\hline
		cluster                             & & 1 & 2 & 3 & 4 & 5 & 6 \\
		\hline
		\text{\#KCs}                           & & 25 & 57 &  0 & 16 &  2 &  0 \\
		\hline
		\text{claw: 1 (oldest)}            & &  15 &   4 &  --- &  0 &  0 &  --- \\
		\text{claw: 2}                     & &   7 &   4 &  --- &  0 &  0 &  --- \\
		\text{claw: 3}                     & &   0 &  15 &  --- &  0 &  0 &  --- \\
		\text{claw: 4}                     & &   3 &  13 &  --- &  0 &  0 &  --- \\
		\text{claw: 5}                     & &   0 &   8 &  --- &  0 &  0 &  --- \\
		\text{claw: 6}                     & &   0 &   3 &  --- &  0 &  0 &  --- \\
		\text{claw: 0 (youngest)}          & &   0 &  10 &  --- & 16 &  2 &  --- \\
	\end{tabular}
	\caption{\label{tab:MBKCageclaw}
		The multiple clusters for the KC neurons are capturing neuron age via the number of claws associated with the neuron.
		We see from the $\Khat=6$ clustering table, for the $n_{KC}=100$ KC neurons, that 
		cluster 1 captures predominantly older neurons,
		cluster 2 captures both old and young neurons, and
		clusters 4 and 5 capture only the youngest neurons.
		See also Figure \ref{fig:MBKCageclaw}.
	}
\end{table}
We see that our spectral clustering of the MB connectome via $\mclustase$,
with principled model selection for choosing embedding dimension and mixture complexity,
yields meaningful results:
a single Gaussian cluster for each of MBIN, MBON, and PN,
and multiple clusters for KC.
That is, we have one substantial revision to
Figure \ref{fig:Fig1}'s illustration of the larval {\it Drosophila} mushroom body connectome as a directed graph on four neuron types:
significant {\em substructure} associated with the KC neurons. Indeed, this 
hints at the possibility of a continuous, rather than discrete, structure for the KC. The paper \cite{EichlerSubmitted}
describes so-called ``claws'' associated with each KC neuron,
and posits that 
KCs with only 1 claw are the oldest, followed in decreasing age by multi-claw KCs (from 2 to 6 claws), with finally the youngest KCs being those with 0 claws. 

Figure \ref{fig:MBKCageclaw} and Table \ref{tab:MBKCageclaw}
use this additional neuronal information to show that the multiple clusters for the KC neurons are capturing neuron age
-- and in a seemingly coherent geometry.
Indeed, precisely because the clusters for the KC neurons
are capturing neuron age 
-- a continuous vertex attribute --
again, in a seemingly coherent geometry, 
we define a ``latent structure model'' (LSM), a generalization of the SBM,
together with a principled semiparametric spectral clustering methodology $\smclustase$ associated thereto.
Specifically,
we fit a continuous curve to (the KC subset of) the data in latent space
and show that traversal of this curve corresponds monotonically to neuron age.
To make this precise, we begin with a directed stochastic block model:
\begin{definition}[Directed Stochastic Blockmodel (SBM)]
	Let $d_{\mathrm{out}} = d_{\mathrm{in}}$, with $d=d_{\mathrm{out}} + d_{\mathrm{in}}$.
	We say that an $n$ vertex graph $(\bA,\bX)\sim\mathrm{RDPG}(F)$
	is a directed stochastic blockmodel (SBM) with $K$ blocks if
	the distribution $F$ is a mixture of $K$ point masses,
	$$dF=\sum_{k=1}^K \rho_k \delta_{x_k},$$
	with block membership probability vector $\vec{\rho}$ in the unit $(K-1)$-simplex 
	and distinct latent positions given by $\bm{\nu}=[\nu_1,\nu_2,\ldots,\nu_K]^\top\in\mathbb{R}^{K\times d}$.
	The first $d_{\mathrm{out}}$ entries of each latent position $\nu_k$ are the out-vectors, denoted $\xi_k \in \mathbb{R}^{d_{\mathrm{out}}}$,
	and the remaining $d_{\mathrm{in}}$ elements are the in-vectors $\zeta_k$.
	We write $G\sim SBM(n,\vec{\rho},\bm{\xi} \bm{\zeta}^\top),$
	and  we refer to $\bm{\xi} \bm{\zeta}^\top\in\mathbb{R}^{K \times K}$ as the block connectivity probability matrix for the model.
\end{definition}
We model the MB connectome as a four-component latent structure model (LSM),
where LSM denotes the ``generalized SBM'' where each ``block''
may be generalized from point mass latent position distribution
to latent position distribution with support on some curve
(with the "block" curves disconnected, as (of course) are SBM's point masses).
So LSM does have block structure, albeit not as simple as an SBM;
and LSM will exhibit clustering, albeit just as transparently as an SBMs. As such, it is similar to other generalizations of SBMs, including the degree-corrected and hierarchical variants.

\begin{definition}[Directed Latent Structure Model (LSM)] 
	Let $d_{\mathrm{out}} = d_{\mathrm{in}}$, and 
	let $F$ be a distribution on a set $\mathcal{X} = \mathcal{Y} \times \mathcal{Z} \subset \mathbb{R}^{d_{\mathrm{out}}} \times \mathbb{R}^{d_{\mathrm{in}}}$
	such that $\langle y,z \rangle\in[0,1]$ for all $y\in\mathcal{Y}$ and $z\in\mathcal{Z}$.
	We say that an $n$ vertex graph $(\mathbf{A},\mathbf{X})\sim\mathrm{RDPG}(F)$
	is a directed latent structure model (LSM) with $K$ ``structure components'' if
	the support of distribution $F$ is a mixture of $K$ (disjoint) curves,
	$$dF=\sum_{k=1}^K \rho_k dF_k(x),$$
	with block membership probability vector $\vec{\rho}$ in the unit $(K-1)$-simplex
	and $F_k$ supported on $\mC_k$ and $\mC_1,\cdots,\mC_K$ disjoint.
	%
	We write $G\sim LSM(n,\vec{\rho},(F_1,\cdots,F_K))$.
\end{definition}
We now investigate our MB connectome as an LSM 
with latent positions $X_i \overset{\mathrm{i.i.d.}}{\sim} F$
where $F$ is no longer a mixture of four point masses with one point mass per neuron type
but instead $\mathrm{supp}(F)$ is three points and a continuous curve $\CKC$.

Motivated by approximate normality of the adjacency spectral embedding of an RDPG, we consider estimating $F$ via a semiparametric Gaussian mixture model for the $\Xhat_i$'s.
Let $H$ be a probability measure on a parameter space $\Theta \subset \mathbb{R}^d \times S_{d \times d}$, 
where $S_{d \times d}$ is the space of $d$-dimensional covariance matrices,
and let $\{\varphi(\cdot;\theta) : \theta \in \Theta \}$ be a family of normal densities.
Then the function given by
$$\alpha(\cdot;H) = \int_{\Theta} \varphi(\cdot;\theta) dH(\theta)$$
is a semiparametric GMM.
$H \in \mathcal{M}$ is referred to as the mixing distribution of the mixture,
where $\mathcal{M}$ is the class of all probability measures on $\Theta$.
If $H$ consists of a finite number of atoms,
then $\alpha(\cdot;H)$ is a finite normal mixture model with means, variances and proportions determined by the locations and weights of the point masses, and \cite{lindsay1983} provides theory for maximum likelihood estimation (MLE) in this context.

Thus
(ignoring covariances for presentation simplicity, so that $\theta \in \Re^d$ is the component mean vector)
we see that the central limit theorem suggests that we estimate the probability density function of the embedded MB connectome
$\Xhat_1,\cdots,\Xhat_{n=213}$,
under the LSM assumption,
as the semiparametric GMM $\alpha(\cdot;H)$
with $\Theta=\mathbb{R}^{6}$
and
where $H=F$ 
is supported by three points and a continuous curve $\CKC$.
Note that in the general case, where $\Theta$ includes both means and covariance matrices, we have $H = H_{F,n}$.
The central limit theorem for the adjacency spectral embedding provides a large-sample approximation for $H_{F,n}$,
and provides a mean-covariance constraint so that
if we knew the latent position distribution $F$ we would have no extra degrees of freedom
(though perhaps a more challenging MLE optimization problem).
As it is, we do our fitting in the general case, with simplifying constraints on the covariance structure associated with $\CKC$.

Our MLE 
(continuing to ignore covariances for presentation simplicity)
is given by
$$d\hat{H}(\theta) = \sum_{k=1}^3 \rhohat_k I\{\theta = \thetahat_k\} + \left(1 - \sum_{k=1}^3 \rhohat_k\right) \rhohat_{KC}(\theta) I\{\theta \in \Chat\}$$
where $\thetahat_1, \thetahat_2, \thetahat_3$ are given by the means of the $\mclustase$ Gaussian mixture components for MBIN, MBON, and PN,
and $\Chat \subset \mathbb{R}^d$ is a one-dimensional curve.
Figure \ref{fig:YPYQ19} displays the MLE results from an EM optimization
for the curve $\Chat$ constrained to be quadratic, as detailed in the Appendix.
(Model testing for $\CKC$ in $\mathbb{R}^6$ does yield quadratic:
testing the null hypothesis of linear against the alternative of quadratic yields clear rejection ($p < 0.001$),
while there is insufficient evidence to favor $H_A$ cubic over $H_0$ quadratic ($p \approx 0.1$).)
\begin{figure}[htbp]
	\centering
	\includegraphics[width=0.65\textwidth]{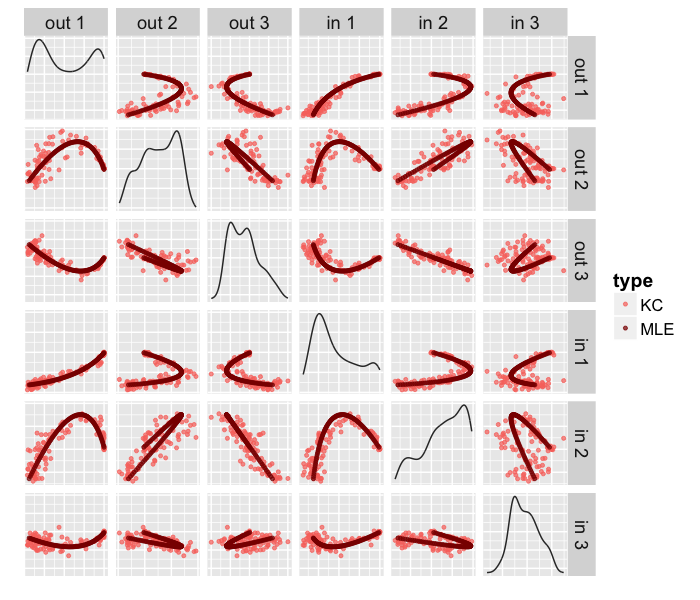}
	\caption{\label{fig:YPYQ19}
		Semiparametric MLE $\Chat$ for the KC latent-space curve in $\mathbb{R}^6$. Figure duplicated from \cite{MBStructure}.
	}
\end{figure}
That is, (continuing to ignore covariances for presentation simplicity)
our structure discovery 
via $\smclustase$ 
yields an $\mathbb{R}^6$ latent position estimate for the MB connectome
-- a {\it connectome code} for the larval {\it Drosophila} mushroom body --
as a semiparametric Gaussian mixture of three point masses
and a continuous parameterized curve $\Chat$; 
the three Gaussians correspond to three of the four neuron types,
and the curve corresponds to the fourth neuron type (KC) with the parameterization capturing neuron age (
see Figure \ref{fig:MBsemiparfig}). We note that \cite{EichlerSubmitted} suggests distance-to-neruopile $\delta_i$
-- the distance to the MB neuropile from the bundle entry point of each KC neuron $i$ --
as a proxy for neuron age, and analyzes this distance in terms of number of claws for neuron $i$ (see Figure \ref{fig:KE1}).
We now demonstrate that
the correlation of this distance with the KC neurons' projection onto the parameterized curve $\Chat$
is highly significant -- this semiparametric spectral model captures neuroscientifically important structure in the connectome.
To wit,
we project each KC neuron's embedding onto our parameterized $\Chat$
and study the relationship between the projection's position on the curve, $t_i$, and the neuron's age through the distance proxy $\delta_i$ (see Figures \ref{fig:YPYQ21} and \ref{fig:YPYQ22}).
We find significant correlation of $\delta_i$ with $t_i$ -- 
Spearman's $s = -0.271$, 
Kendall's $\tau = -0.205$, 
Pearson's $\rho = -0.304$,
with $p < 0.01$ in each case
-- demonstrating that our semiparametric spectral modeling captures biologically relevant neuronal properties.

\begin{figure}[htbp]
	\centering
	\includegraphics[width=0.55\textwidth]{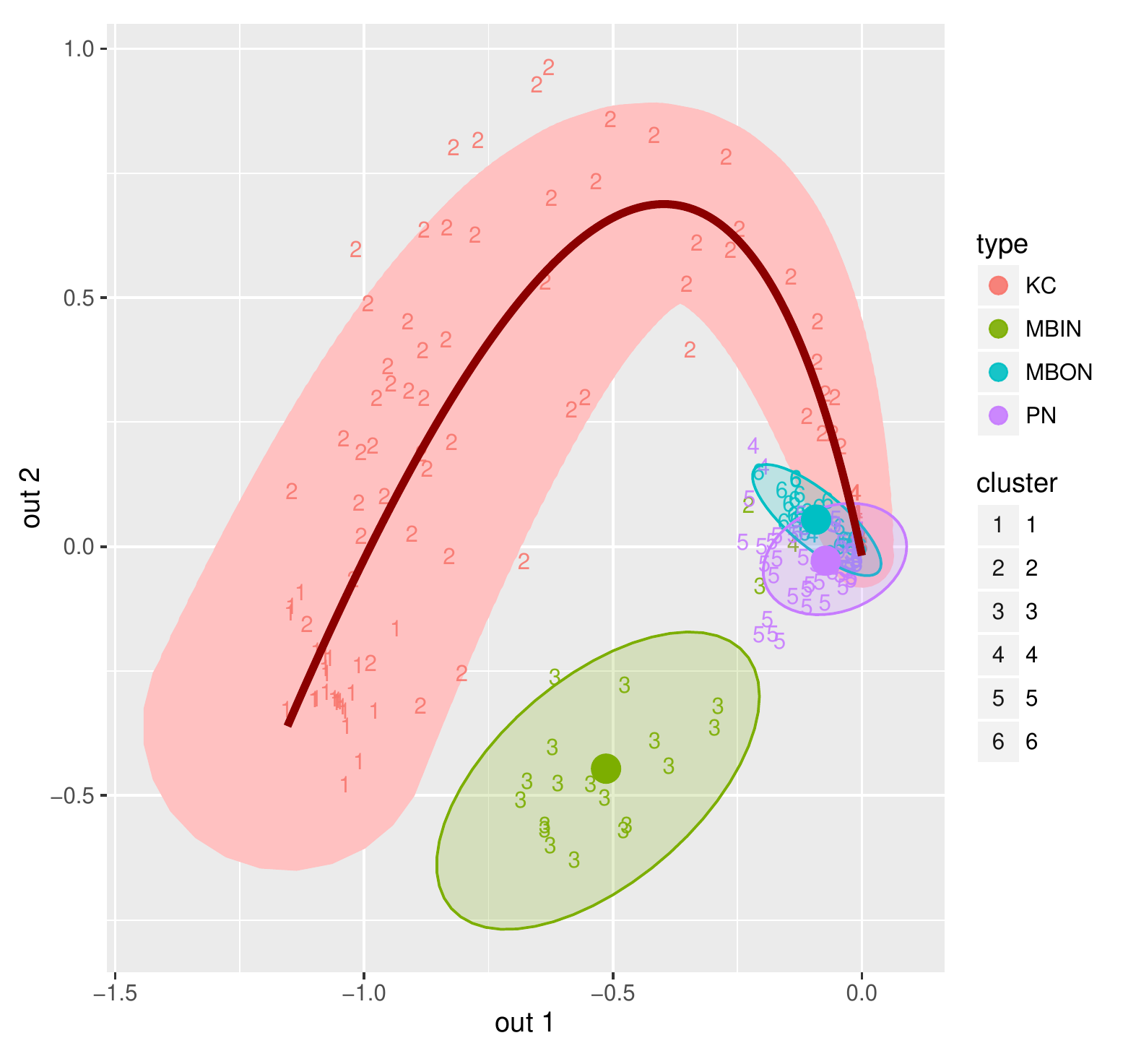}
	\caption{\label{fig:MBsemiparfig}
		Semiparametric spectral latent space estimate of our MB connectome as three Gaussians and a KC curve:
		colors distinguish the four neuron types and
		numbers distinguish the original $\Khat=6$ clusters.
		Recall that this is a two-dimensional visualization of six-dimensional structure. Figure duplicated from \cite{MBStructure}.
	}
\end{figure}

\begin{figure}[htbp]
	\centering
	\includegraphics[width=0.60\textwidth]{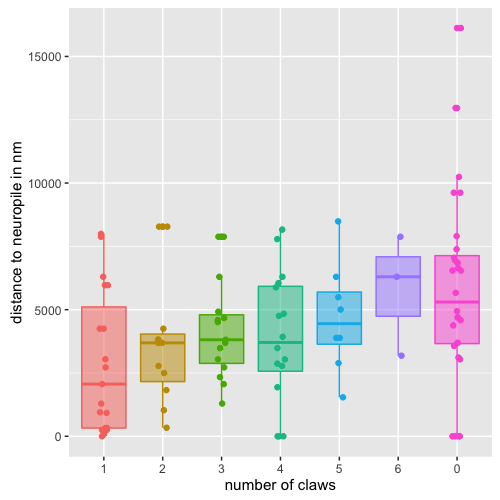}
	\caption{\label{fig:KE1}
		Relationship between number of claws and distance $\delta_i$ (a proxy for age) for the KC neurons,
		from \cite{EichlerSubmitted}. Figure duplicated from \cite{MBStructure}.
	}
\end{figure}

\begin{figure}[htbp]
	\centering
	\includegraphics[width=0.63\textwidth]{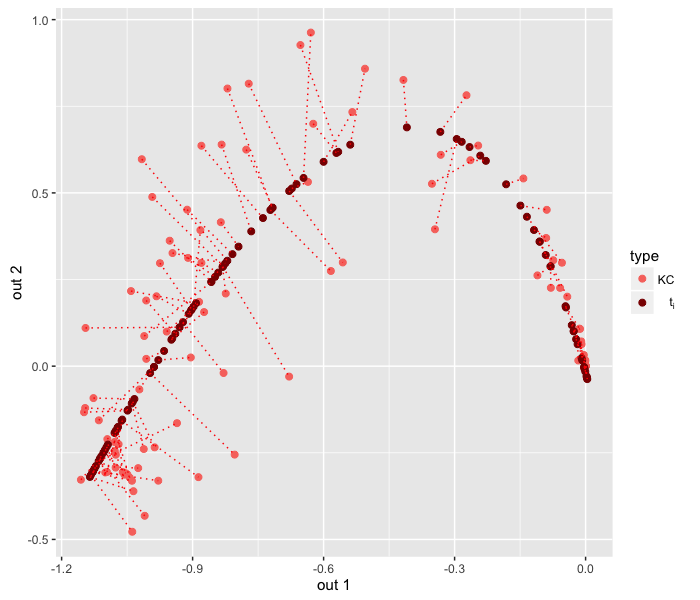}
	\caption{\label{fig:YPYQ21}
		Projection of KC neurons onto the quadratic curve $\Chat$, yielding projection point $t_i$ for each neuron.
		Recall that this is a two-dimensional visualization of six-dimensional structure. Figure duplicated from \cite{MBStructure}.
	}
\end{figure}

\begin{figure}[htbp]
	\centering
	\includegraphics[width=0.63\textwidth]{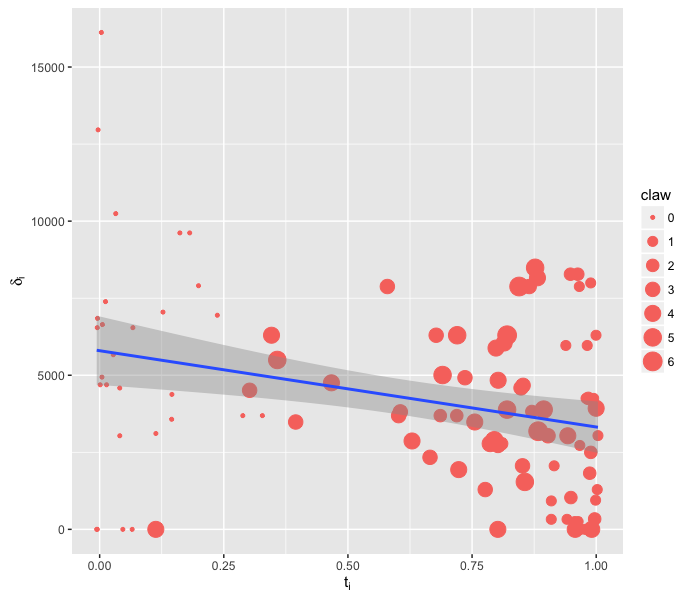}
	\caption{\label{fig:YPYQ22}
		The correlation between the projection points $t_i$ on the quadratic curve $\Chat$ and distance $\delta_i$ (a proxy for age) for the KC neurons
		is highly significant, demonstrating that our semiparametric spectral modeling captures biologically relevant neuronal properties. Figure duplicated from \cite{MBStructure}.
	}
\end{figure}

In summary, motivated by the results of a spectral clustering investigation of
the recently-reconstructed synapse-level larval {\it Drosophila} mushroom body structural connectome,
which demonstrate conclusively that modeling the Kenyon Cells (KC) demands additional latent space structure,
we have developed semiparametric spectral modeling.
Exploratory data analysis suggests that the MB connectome can be productively approximated by a 
four-component latent structure model (LSM),
and the resulting 
MB connectome code 
derived via $\smclustase$
captures biologically relevant neuronal properties.  Data and code for all our analyses are available at
\url{http://www.cis.jhu.edu/~parky/MBstructure.html}.

Of course, the true connectome code is more elaborate,
and cannot be completely encompassed by any simple latent position model --
such a model precludes the propensity for transitivity, e.g.\ --
but our semiparametric spectral modeling provides another step along the path.
In terms of a (partial) ladder of biological scales
-- e.g., {\it C.\ elegans}, {\it Drosophila}, zebrafish, mouse, primate, and humans --
this works moves us off the first rung for analysis of a complete
neurons-as-vertices and synapses-as-edges connectome.


\section{Conclusion: complexities, open questions, and future work}\label{sec:Complexities} 

Our paradigm for statistical inference on random graphs is anchored by the familiar pillars of classical Euclidean inference. We exhibit estimates for graph parameters that satisfy (uniform) consistency and asymptotic normality, and we demonstrate how these estimates can be exploited in a bevy of subsequent inference tasks: community detection in heterogeneous networks, multi-sample graph hypothesis testing, and exploratory data analysis in connectomics. The random dot product graph model in which we ground our techniques has both linear algebraic transparency and wide applicability, since a hefty class of independent-edge graphs is well-approximated by RDPGs. The lynchpins of our approach are spectral decompositions of adjacency and Laplacian matrices, but many of our results and proof techniques can be applied to more general random matrices. In recent work, for example, we examine eigenvalue concentration for certain classes of random matrices \cite{cape_16_conc}, and accurate estimation of covariance matrices \cite{cape_covar}. As such, our methodology is a robust edifice from which to explore questions in graph inference, data analysis, and random matrix theory. 

The results we summarize here are but the tip of the iceberg, though, and there persist myriad open problems in spectral graph inference---for random dot product graphs in particular, of course, but also for random graphs writ large. In this section, we outline some current questions of interest and describe future directions for research. 

Our consistency results for the adjacency spectral embedding depend on knowing the correct embedding dimension.  In real data, this optimal embedding dimension is typically not only not known, but, since the RDPG model is only an approximation to any true model, may well depend on the desired subsequent inference task.
As we have pointed out in Section \ref{subsec:MBStructure}, multiple methods exist for estimating embedding dimension, and the universal singular value thresholding of \cite{chatterjee2015} and other thresholding methods \cite{fishkind2013consistent} are theoretically justified in the large-$n$ limit.
For finite $n$, however, model selection is considerably trickier. 
Underestimation of the embedding dimension can markedly---and provably---bias subsequent inference. While we do not address it here, we remark that asymptotic results can be shown for the adjacency spectral embedding of a $d$-dimensional RDPG when the chosen embedding dimension is $d'<d$.  On the other hand, if the embedding dimension is overestimated, no real signal is lost; therefore, most embedding methods continue to perform well, albeit with some loss of efficiency due to increased variance. Precisely quantifying this loss of efficiency is very much an open question and is related to analogous classical questions in random matrix theory \cite{tao2012random}.

In our RDPG model, an important assumption is that the $\bP$ matrix be positive semidefinite. While this limits the model, a slight generalization of the RDPG shares many of its important properties.
Considering a matrix of latent positions $\mathbf{X} \in \mathbb{R}^{p+q}$, one can set $\mathbf{P} = \mathbf{X} \mathbf{I}_{p,q} \mathbf{X}^{\top}$ where $\mathbf{I}_{p,q}$ is the diagonal matrix of size $(p+q) \times (p+q)$ with $p$ diagonal entries being $1$ and $q$ diagonal entries being $-1$. Under this generalization, any $\bP$ can be obtained, provided that $p+q$ is appropriately chosen. This then implies that any latent position graph model, even a non-positive-semidefinite one, can be approximated arbitrarily closely by this generalized RDPG.

We also wish to adapt our procedures to weighted graphs.  For simple weighting, such as Poisson-distributed weights, our existing methodology applies quite naturally to the weighted adjacency matrix. More intricate questions arise when the weights are contaminated or their distribution is heavily skewed. In such cases, one can ``pass to ranks"; that, is replace the nonzero weight by its normalized rank among all the edge weights.  This mitigates skew and works well in many examples, but a deeper theoretical analysis of this procedure, as well as other approaches to weighted graph inference, remain open.

To address graphs with self-edges, note that the random dot product graph model does not preclude such edges.
Indeed, since $\bP$ is not constrained to be hollow, the adjacency matrix $\bA$ thus generated need not be hollow, either.  
However, the generative mechanism for self-edges in a graph may be different from the mechanism for edges between two different vertices.  One approach to addressing this is to set the diagonal entries of $\bA$ to zero and then {\em augment} the diagonal artificially with imputed values. 
In fact, even when there are no self loops, such procedures can improve finite-sample inferential performance.
In \cite{Marchette2011VN}, it is demonstrated that augmenting the diagonal via $\bA_{ii}=d_i/(n+1)$, where $d_i$ is the degree of vertex $i$, can improve inference in certain cases. 
Similarly, \cite{Scheinerman2010} describes an iterative procedure to find a diagonal augmentation consistent with the low-rank structure of $\bP$.
It is still unclear exactly what augmentation of the diagonal, if any, might be optimal for each particular inference task.

In the case when the vertices or edges are corrupted by occlusion or noise, \cite{priebes.:_statis} and \cite{levin_lyzin_laplacian} consider approaches to vertex classification and graph recovery, demonstrating that spectral embeddings are robust to certain sources of graph error. Violations of the independent edge assumption, though, can lead to more significant hurdles, both theoretical practical, since it is a requirement for a number of the concentration inequalities on which we depend.  

For joint graph inference and testing, open problems abound. We mention, for instance, the analysis of the omnibus embedding when the $m$ graphs are correlated, or when some are corrupted; a closer examination of the impact of the Procrustes alignment on power; the development of an analogue to a Tukey test for determining which graphs differ when we test for the equality of distribution for more than two graphs; the comparative efficiency of the omnibus embedding relative to other spectral decompositions; and a quantification of the trade-off for subsequent inference between a large number of independent graphs and large graph size (\cite{runze_law_large_graphs}).

In sum, the random dot product graph is a compact, manageable, and applicable model. The Euclidean nature of the adjacency and Laplacian spectral embedding for a random dot product graph allows us to approach statistical inference in this setting from a familiar Euclidean perspective.  Both the adjacency spectral embedding and the Laplacian spectral embedding can be profitably leveraged for latent position estimation and single- and multi-sample graph hypothesis testing. Moreover, our distributional results for these spectral embeddings provide reassuring classical analogues of asymptotic normality for estimators, and in current ongoing work, we consider how to compare asymptotic relative efficiency of different estimators for graph parameters.
While spectral methods may not always be optimal for a given task, they are often feasible and can provide a way to accurately initialize more complex procedures. Moreover, these Euclidean representations of graph data render possible the application of techniques for analysis of Euclidean data---clustering, classification, and density estimation, for instance---to graphs.
As we have outlined above, while many important theoretical and practical challenges remain, spectral embeddings for random dot product graphs constitute an important piece of the greater puzzle of random graph inference.

\section{Appendix}\label{sec:Appendix}
In the appendix, we provide details for the proofs of our results on consistency and asymptotic normality of the adjacency spectral embedding, as well as an outline of the proof of the central limit theorem for the Laplacian spectral embedding.  Our notation remains as stated in Section \ref{sec:ASE_Inference_RDPG}. We begin with a detailed proof of our result on the consistency, in the $2 \to \infty$ norm, of the ASE for latent position recovery in RDPGs.
\subsection*{Proof of Theorem~\ref{thm:minh_sparsity}}
\label{sec:minh}
Let us recall {\bf Theorem \ref{thm:minh_sparsity}}:
Let $\bA_n \sim \mathrm{RDPG}(\bX_n)$ for $n \geq 1$ be a sequence of random dot product graphs where the $\bX_n$ is assumed to be of rank $d$ for all $n$ sufficiently large. Denote by $\hat{\bX}_n$ the adjacency spectral embedding of $\bA_n$ and let $(\hat{\bX}_{n})_{i}$ and $(\bX_n)_{i}$ be the $i$-th row of $\hat{\bX}_n$ and $\bX_n$, respectively. Let $E_n$ be the event that there
exists an orthogonal transformation $\bW_n \in \mathbb{R}^{d \times d}$ such that
\begin{equation*}
\max_{i} \| (\hat{\bX})_{i} - \bW_n (\bX_n)_{i} \| \leq 
\frac{C d^{1/2} \log^2{n}}{\delta^{1/2}(\mathbf{P}_n)}
\end{equation*}
where $C > 0$ is some fixed constant and $\mathbf{P}_n = \mathbf{X}_n \mathbf{X}_n^{\top}$. Then $E_n$ occurs asymptotically almost surely; that is, $\Pr(E_n) \rightarrow 1$ as $n \rightarrow \infty$.

The proof of Theorem~\ref{thm:minh_sparsity} will follow from a succession of supporting results.  
We note that Theorem~\ref{thm:minh_frob}, which deals with the accuracy of spectral embedding estimates in Frobenius norm, may be of independent interest.  
We begin with the following simple but essential proposition, in which we show that $\UP^{\top} \UA$ is close to an orthogonal transformation. For ease of exposition, in the remainder of this subsection we shall suppress the subscript index $n$ from our matrices $\mathbf{X}_n$, $\mathbf{A}_n$ and $\hat{\mathbf{X}}_n$.
\begin{proposition}
	\label{prop:uptua_close_rotation}
	Let $\bA \sim \mathrm{RDPG}(\bX)$ and let
	$\bW_1 \bm{\Sigma} \bW_2^{\top}$ be the singular value
	decomposition of $\UP^{\top}
	\UA$. 
	Then with high probability,
	\begin{equation*}
	\|\UP^{\top} \UA -
	\bW_1 \bW_2^{\top} \| = O(\delta^{-1}(\mathbf{P})).
	\end{equation*}
\end{proposition}
\begin{proof}
	Let $\sigma_1, \sigma_2, \dots, \sigma_d$ denote the singular values of
	$\UP^{\top} \UA$ (the diagonal
	entries of $\bm{\Sigma}$). Then $\sigma_i = \cos(\theta_i)$ where
	the $\theta_i$ are the principal angles between the subspaces
	spanned by $\UA$ and
	$\UP$. Furthermore, by the Davis-Kahan Theorem,
	\begin{align*}
	\| \UA \UA^{\top} -
	\UP \UP^{\top} \| &=
	\max_{i} | \sin(\theta_i) |\leq
	\frac{C \sqrt{d} \|\bA - \bP\|}{\lambda_{d}(\bP)}\\
	\end{align*}
	for sufficiently large $n$. Recall here $\lambda_{d}(\bP)$ denotes 
	the $d$-th largest eigenvalue of $\bP$. The spectral norm bound for $\bA
	- \bP$ from Theorem~\ref{thm:lu_peng} along with the assumption that $\lambda_d(\mathbf{P})/\delta(\mathbf{P}) \geq c_0$ for some constant $c_0$ in Assumption~\ref{ass:max_degree_assump}
	yield
	$$\| \UA \UA^{\top} -
	\UP \UP^{\top} \|\leq\frac{C \sqrt{d}}{\delta^{1/2}(\mathbf{P})}.$$
	We thus have
	\begin{equation*}
	\begin{split}
	\|\UP^{\top} \UA -
	\bW_1 \bW_2^{\top} \| &= \|\bm{\Sigma} -
	I \| = \max_i |1 - \sigma_i| \leq
	\max_i (1 - \sigma_i^{2})  \\ & = \max_i \sin^{2}(\theta_i) = 
	\|\UA \UA^{T} -
	\UP \UP^{\top} \|^{2}
	= O(\delta^{-1}(\mathbf{P}))
	\end{split}
	\end{equation*}
	as desired. 
\end{proof}
Denote by $\bW^{*}$ the orthogonal matrix
$\bW_1 \bW_2^{\top}$ as defined in the above
proposition. We now establish the following key lemma. The lemma
allows us to exchange the order of the orthogonal transformation
$\bW^{*}$ and the diagonal scaling transformation $\SA$ or $\SP$. 
\begin{lemma}
	\label{lem:order_bounds_on_minh_differences}
	Let $(\bA, \bX) \sim \mathrm{RDPG}(F)$ with sparsity factor $\rho_n$. Then asymptotically almost surely,
	\begin{equation}\label{eq:approx_commute1_nosqrt}
	\|\bW^{*} \SA - \SP \bW^{*} \|_{F} = O(\log n)
	\end{equation}
	and
	\begin{equation}\label{eq:approx_commute2_sqrt}
	\|\bW^{*} \SA^{1/2}  - \SP^{1/2}
	\bW^{*} \|_{F} = O((\log n) \,\delta^{-1/2}(\mathbf{P}))
	\end{equation} 
\end{lemma}
\begin{proof}
	Let $\bR = \UA - \UP\UP^{\top} \UA$. We note
	that $\bR$ is the residual after projecting
	$\UA$ orthogonally onto the column space of
	$\UP$, and thus
	\begin{align*} 
	\|\UA - \UP
	\UP^{\top} \UA \|_{F} \leq \min_{\mathbf{W}} \|\mathbf{U}_{\mathbf{A}} - \mathbf{U}_{\mathbf{P}} \mathbf{W} \|_{F}
	\end{align*}  
	where the minimization is over all orthogonal matrices $\mathbf{W}$. 
	The variant of the Davis-Kahan Theorem given in Eq.~\eqref{eq:Davis_Kahan_variant1} then implies
	$$ \|\UA - \UP
	\UP^{\top} \UA \|_{F} \leq O(\delta^{-1/2}(\mathbf{P})).$$ 
	We next derive that
	\begin{align*}
	\bW^{*}& \SA =  (\bW^{*} -
	\UP^{\top} \UA)
	\SA  +
	\UP^{\top} \UA
	\SA
	= (\bW^{*} -
	\UP^{\top} \UA) \SA +
	\UP^{\top} \bA \UA
	\\ &=  (\bW^{*} -
	\UP^{\top} \UA) \SA  +\UP^{\top}(\bA - \bP)
	\UA +
	\UP^{\top}\bP \UA \\
	&=  (\bW^{*} -
	\UP^{\top} \UA) \SA +
	\UP^{\top}(\bA - \bP) \bR+ \UP^{\top}(\bA - \bP)
	\UP \UP^{\top}
	\UA + \UP^{\top}
	\bP \UA \\
	&= (\bW^{*} -
	\UP^{\top} \UA) \SA +
	\UP^{\top}(\bA - \bP) \bR+ \UP^{\top}(\bA - \bP)
	\UP \UP^{\top}
	\UA + \SP \UP^{\top}
	\UA 
	\end{align*}
	Writing $\SP \UP^{\top}
	\UA = \SP
	(\UP^{\top} \UA -
	\bW^{*}) + \SP \bW^{*}$ and
	rearranging terms, we obtain
	\begin{equation*}
	\begin{split}
	\|\bW^{*} \SA -
	\SP \bW^{*}\|_{F} & \leq  \|\bW^{*}
	- \UP^{\top} \UA \|_{F}
	(\|\SA\| + \|\SP\|) + 
	\|\UP^{\top}(\bA - \bP)
	\bR\|_{F}\\
	&+  \|\UP^{\top}(\bA -
	\bP) \UP\UP^{\top}\UA\|_{F} \\
	& \leq  O(1) + O(1) + \|\UP^{\top}(\bA -
	\bP) \UP\|_{F}\|\UP^{\top}\UA\|
	\end{split}
	\end{equation*}
	asymptotically almost surely. Now, $\|\UP^{\top}\UA\| \leq 1$. Hence we can focus on the term  $\UP^{\top}(\bA -
	\bP) \UP$, which is a $d \times d$ matrix
	whose $ij$-th entry is of the form
	\begin{align*}
	(\UP^{\top} (\bA - \bP) \UP)_{ij} &=
	\sum_{k=1}^{n} \sum_{l=1}^{n} (\bA_{kl} -
	\bP_{kl}) \bU_{ki} \bU_{lj} \\
	&= 2 \sum_{k,l : k < l}
	(\bA_{kl} - \bP_{kl})\bU_{ki}\bU_{lj} - \sum_{k}
	\bP_{kk} \bU_{ki} \bU_{kj}
	\end{align*}
	where $\bU_{\cdot i}$ and $\bU_{\cdot j}$ are the $i$-th and $j$-th
	columns of $\UP$. Thus, conditioned on
	$\bP$, $(\UP^{\top} (\bA - \bP)
	\UP)_{ij}$ is a sum of independent mean $0$ random
	variables and a term of order $O(1)$. 
	Now, by Hoeffding's inequality, 
	\begin{align*}
	\Pr&\left[ \bigg|\sum_{k,l : k < l} 2 (\bA_{kl}- \bP_{kl}) \bU_{ki} u_{lj}\bigg|\geq t \right] \leq 2\exp
	\Bigl( \frac{-2t^2}{\sum_{k,l : k < l} (2\bU_{ki} \bU_{lj})^2}\Bigr) 
	\leq 2\exp(-t^2).
	\end{align*}
	Therefore, each entry of
	$\UP^{\top}(\bA - \bP)
	\UP$ is of order
	$O(\log n)$ with high probability, and as a consequence, since $\UP^{\top} (\bA-\bP)\UP$ is a $d \times d$ matrix,
	\begin{equation}\label{eq:hoeffding_up_upt}
	\|\UP^{\top}(\bA - \bP)
	\UP\|_{F} = O(\log n)
	\end{equation} 
	with high probability.
	We establish that
	$$\|\bW^{*} \SA^{1/2} -
	\SP^{1/2} \bW^{*} \|_{F} = O((\log n) \lambda_d^{-1/2}(\mathbf{P}))$$
	by
	noting that the $ij$-th entry of $\bW^{*} \SA^{1/2} -
	\SP^{1/2} \bW^{*}$ can be written as
	$$\bW^{*}_{ij} (\lambda_i^{1/2}(\bA) -
	\lambda_{j}^{1/2}(\bP)) = \bW^{*}_{ij} \frac{\lambda_{i}(\bA) -
		\lambda_{j}(\bP)}{\lambda_{i}^{1/2}(\bA) + \lambda_{j}^{1/2}(\bP)},  $$ 
	and the desired bound follows from the above after bounding $\lambda_i(\bA)$, either by Weyl's inequality and Theorem~\ref{thm:lu_peng}, or, alternatively, by a Kato-Temple inequality from \cite{cape_16_conc}.
\end{proof}
We next present Theorem~\ref{thm:minh_frob}, which allows us to write the Frobenius norm difference of the adjacency spectral embedding $\hat{\mathbf{X}}$ and the true latent position $\bX$ in terms of the Frobenius norm difference of $\mathbf{A} - \mathbf{P}$ and smaller order terms. 
\begin{theorem}
	\label{thm:minh_frob}
	Let $\bA \sim \mathrm{RDPG}(\bX)$. Then there exists an orthogonal matrix $\mathbf{W}$ such that, with high probability,
	\begin{align*}
	\|\Xhat - \bX \bW
	\|_{F} = \|(\bA - \bP) \UP
	\SP^{-1/2} \|_{F} + O((\log n)  \delta^{-1/2}(\mathbf{P}))
	\end{align*}
\end{theorem}
\begin{proof}
	Let 
	\begin{align*}
	\bR_1 &= \UP\UP^{\top} \UA -\UP \bW^{*}\\
	\bR_2 &= (\bW^{*}
	\SA^{1/2} - \SP^{1/2} \mathbf{W}^{*}).
	\end{align*}
	We deduce that
	\begin{equation*}
	\begin{split}
	\Xhat - \UP
	\SP^{1/2} \bW^{*} &=
	\UA \SA^{1/2}
	-\UP W^{*}
	\SA^{1/2} + \UP
	(\bW^{*}
	\SA^{1/2} - \SP^{1/2} W^{*}) \\
	&= (\UA - \UP \UP^{\top}
	\UA) \SA^{1/2} +
	\bR_1 \SA^{1/2} + \UP
	\bR_2 
	\\ &= \UA \SA^{1/2} -
	\UP \UP^{\top}
	\UA \SA^{1/2}  + \bR_1 \SA^{1/2} + \UP
	\bR_2 
	\end{split}
	\end{equation*}
	Observe that $\UP \UP^{\top}\bP = \bP$ and $\UA\SA^{1/2} = \bA \UA
	\SA^{-1/2}$. Hence
	\begin{align*}
	\Xhat - \UP\SP^{1/2} \bW^{*} = &(\bA -
	\bP) \UA \SA^{-1/2} - \UP
	\UP^{\top} (\bA - \bP)
	\UA \SA^{-1/2} + \bR_1 \SA^{1/2} + \UP\bR_2 
	\end{align*}
	Writing 
	\begin{equation*}
	\bR_3 = \UA - \UP
	\bW^{*} = \UA - \UP
	\UP^{\top} \UA + \bR_1,
	\end{equation*}
	we derive that
	\begin{equation*}
	\begin{split}
	\Xhat - \UP\SP^{1/2} \bW^{*} =& (\bA -
	\bP) \UP \bW^{*}
	\SA^{-1/2}- \UP\UP^{\top}(\bA - \bP)
	\UP \bW^{*}
	\SA^{-1/2} \\
	&+ (\bI -
	\UP \UP^{\top}) (\bA - \bP)
	\bR_{3} \SA^{-1/2}
	+ \bR_1 \SA^{1/2} + \UP
	\bR_2 
	\end{split}
	\end{equation*}
	Recalling that $\|\UA - \UP
	\UP^{\top} \UA \|_{F}  = O(\delta^{-1/2}(\mathbf{P}))$ with high probability, we have
	\begin{align*}
	\|\bR_{1}\|_{F} = O(\delta^{-1}(\mathbf{P})), \quad \|\bR_2\|_{F} = O((\log n) \, \delta^{-1/2}(\mathbf{P})), \quad \text{and} \,\, 
	\|\bR_3\|_{F} = O(\delta^{-1/2}(\mathbf{P}))
	\end{align*}
	with high probability.
	Furthermore, a similar application of Hoeffding's inequality to that in the proof of Lemma~\ref{lem:order_bounds_on_minh_differences}, along with an application of Weyl's inequality and Theorem~\ref{thm:lu_peng} to bound $\lambda_i(\mathbf{A})$, ensures that
	\begin{equation*}
	\|\UP
	\UP^{\top}(\bA - \bP)
	\UP  \bW^{*}
	\SA^{-1/2} \|_{F}  \leq \|\UP^{\top}(\bA - \bP)
	\UP\|_{F} \|\SA^{-1/2}
	\|_{F}= O(\log n  \delta^{-1/2}(\mathbf{P}))/
	\end{equation*}
	As a consequence, with high probability
	\begin{align*}
	\|\Xhat - \UP\SP^{1/2} \bW^{*}\|_{F}
	&=  \|(\bA -
	\bP) \UP \bW^{*}
	\SA^{-1/2}\|_{F} + O((\log n) \, \delta^{-1/2}(\mathbf{P})) \\
	&= 
	\|(\bA - \bP) \UP
	\SP^{-1/2} \bW^{*}-  (\bA - \bP) \UP
	(\SP^{-1/2} \bW^{*} - \bW^{*}
	\SA^{-1/2}) \|_{F}\\
	&\hspace{5mm}+ O((\log n) \, \delta^{-1/2}(\mathbf{P})) 
	\end{align*}
	A very similar argument to that employed in the proof of  Lemma~\ref{lem:order_bounds_on_minh_differences} implies that
	\begin{equation*}
	\|\SP^{-1/2} \bW^{*} - \bW^{*}
	\SA^{-1/2} \|_{F} = O((\log n) \, \delta^{-3/2}(\mathbf{P}))
	\end{equation*}
	with high probability.
	We thus obtain 
	\begin{align}
	\label{eq:2}
	\|\Xhat - \UP
	\SP^{1/2} \bW^{*}\|_{F}
	&= \|(\bA - \bP) \UP
	\SP^{-1/2} \bW^{*} \|_{F}+ O((\log n) \, \delta^{-1/2} (\mathbf{P})) 
	\notag\\ &= \|(\bA - \bP) \UP
	\SP^{-1/2} \|_{F} + O((\log n) \, \delta^{-1/2}(\mathbf{P}))
	\end{align}
	with high probability. 
	Finally, to complete the proof, we note that 
	$\bX =
	\UP \SP^{1/2} \tilde{\bW}$
	for some orthogonal matrix $\bW$. Since $\bW^{*}$ is also
	orthogonal, we conclude that there exists some orthogonal $\bW$
	for which 
	$\bX \bW = \UP
	\SP^{1/2} \bW^{*},$
	as desired. 
\end{proof}
We are now ready to prove the $2 \rightarrow \infty$ consistency we assert in Theorem~\ref{thm:minh_sparsity}.
\begin{proof}
	
	To establish Theorem~\ref{thm:minh_sparsity}, we note that from Theorem~\ref{thm:minh_frob}
	\begin{align*}
	\|\Xhat - \bX \bW
	\|_{F} &= \|(\bA - \bP) \UP
	\SP^{-1/2} \|_{F}
	+ O((\log n) \, \delta^{-1/2}(\mathbf{P}))
	\end{align*}
	and hence
	\begin{align*}
	\max_{i} \| \Xhat_i - \rho_n^{1/2} \bW \bX_i \| 
	&\leq 
	\frac{1}{\lambda_{d}^{1/2}(\bP)} 
	\max_{i} \|((\bA - \bP) \UP)_{i} \|+ O((\log n) \, \delta^{-1/2}(\mathbf{P}))
	\\
	& \leq \frac{d^{1/2}}{\lambda_{d}^{1/2}(\bP)} 
	\max_{j} \|(\bA - \bP) \bU_{\cdot j} \|_{\infty}+ O((\log n) \, \delta^{-1/2}(\mathbf{P}))
	\end{align*}
	where $\bU_{\cdot j}$ denotes the $j$-th column of
	$\UP$. 
	Now, for a given $j$ and a given index $i$, the $i$-th element of
	the vector $(\bA - \bP) \bU_{\cdot j}$ is of the form
	\begin{equation*}
	\sum_{k} (\bA_{ik} - \bP_{ik}) \bU_{kj}
	\end{equation*}
	and once again, by Hoeffding's inequality, the above term is $O(\log n)$ asymptotically almost surely. Taking the union bound
	over all indices $i$ and all columns $j$ of $\UP$, we conclude that with high probability,
	\begin{align*}
	\max_{i} \| \Xhat_i - \rho_n^{1/2} \bW \bX_i \|& \leq
	\frac{C d^{1/2}}{\lambda_{d}^{1/2}(P)} \log^2 n
	+ O((\log n) \, \delta^{-1/2}(\mathbf{P}))\\
	& \leq \frac{C d^{1/2} \log^2{n}}{\delta^{1/2}(\mathbf{P})}
	\end{align*}
	as desired.
\end{proof}

Now, we move to our distributional results.
\subsection{Proof of the central limit theorem for the adjacency spectral embedding}\label{subsec:CLT_proofdetails}.

Recall {\bf Theorem \ref{thm:clt_orig_but_better}}: Let $(\bA_n, \bX_n) \sim \mathrm{RDPG}(F)$ be a sequence of adjacency matrices and associated latent positions of a $d$-dimensional random dot product graph according to an inner product distribution $F$. Let $\Phi(\bx,\bSigma)$ denote the cdf of a (multivariate)
Gaussian with mean zero and covariance matrix $\bSigma$,
evaluated at $\bx \in \R^d$. Then
there exists a sequence of orthogonal $d$-by-$d$ matrices
$( \Wn )_{n=1}^\infty$ such that for all $\bm{z} \in \R^d$ and for any fixed index $i$,
$$ \lim_{n \rightarrow \infty}
\Pr\left[ n^{1/2} \left( \Xhat_n \Wn - \bX_n \right)_i
\le \bm{z} \right]
= \int_{\supp F} \Phi\left(\bm{z}, \bSigma(\bx) \right) dF(\bx), $$
where
\begin{equation}
\label{def:sigma}\bSigma(\bx) 
= \Delta^{-1} \E\left[ (\bx^{\top} \bX_1 - ( \bx^{\top} \bX_1)^2 ) \bX_1 \bX_1^{\top} \right] \Delta^{-1}; \quad \text{and} \,\, \Delta = \mathbb{E}[\bX_1 \bX_1^{\top}].
\end{equation}

To prove this, we begin with the following simple lemma which indicates that when the rows of $\mathbf{X}$ are sampled i.i.d. from some distribution $F$, that the eigenvalues of $\mathbf{X}$ grows proportionally with $n$. 
\begin{lemma}\label{lem:eigs_lower_bound}
	With notation as above, let $F$ be an inner product distribution and suppose $\bX_1, \cdots, \bX_n, \mathbf{Y}$ be i.i.d $F$. Suppose also that $\Delta = \mathbb{E}[\bX_1 \bX_1^{\top}]$ is of rank $d$. 
	Then for $1 \leq i \leq d$, $\lambda_i(\bP) = \Omega( n \lambda_i(\Delta) )$ almost surely.
\end{lemma}
\begin{proof}
	Recall that for any matrix $\bH$, the nonzero eigenvalues of $\bH^{\top} \bH$ are the same as those of $\bH \bH^{\top}$, so
	$\lambda_i(\bX \bX^{\top}) = \lambda_i( \bX^{\top} \bX).$
	In what follows, we remind the reader that $\bX$ is a matrix whose rows are the tranposes of the column vectors $\bX_i$, and $\bY$ is a $d$-dimensional vector that is independent from and has the same distribution as that of the $\bX_i$. We observe that
	$ (\bX^{\top} \bX - n\E \mathbf{Y} \mathbf{Y}^{\top})_{ij} = \sum_{k=1}^n (\bX_{ki} \bX_{kj} - \E \bY_i \bY_j) $
	is a sum of $n$ independent
	zero-mean random variables, each contained in $[-1,1]$.
	Thus, Hoeffding's inequality yields, for all $i,j \in [d]$,
	$$ \Pr\left[ |(\bX^{\top} \bX) - n\E \bY \bY^{\top}|_{ij} \ge 2\sqrt{ n \log n } \right]
	\le \frac{2}{n^2}. $$
	A union bound over all $i,j \in [d]$ implies that
	$ \|\bX^{\top} \bX - n\E \bY \bY^{\top} \|_F^2 \le 4d^2 n \log n $
	with probability at least $1 - 2d^2/n^2$.
	Taking square roots and noting that the Frobenius norm is an upper bound
	on the spectral norm, we have that
	$ \| \bX^{\top} \bX - n\E (\bY \bY^{\top}) \| \le 2d \sqrt{ n \log n } $
	with probability at least $1 - 2d^2/n^2$,
	and Weyl's inequality~\cite{horn85:_matrix_analy} yields that
	for all $1 \le i \le d$,
	$ | \lambda_i(\bX \bX^{\top} ) - n\lambda_i( \E (\bY \bY^{\top} ) | \le 2d \sqrt{ n \log n } $
	with probability at least $1 - 2d^2/n^2$.
	Of course, since the vector $\bY$ has the same distribution as any of the latent positions $\bX_i$, we see that $\E(\bY \bY^{\top})=\Delta$.
	By the reverse triangle inequality, for any $1 \leq i \leq d$, we have
	$$ \lambda_i(\bX \bX^{\top})\ge\lambda_d(\bX \bX^{\top})
	\ge | n\lambda_d( \Delta) - 2d \sqrt{n \log n} |
	= \Omega( n ). $$ 
	Multiplying through by $\rho_n$, we find that there exists some constant $C$ so that for all $n$ sufficiently large, $\lambda_d(\rho_n\bX \bX^{\top}) \geq n\rho_n \lambda_d(\Delta)$ with probability at least $1 - 2d^2/n^2$.
\end{proof}

As described previously, to prove our central limit theorem, we require somewhat more
precise control on certain residual terms,
which we establish in the following key lemma.  In the lemmas and proofs that follow, we frequently suppress the dependence of the sequence of graphs and embeddings on $n$.
\begin{lemma}
	\label{lem:stringent_control_residuals}
	Let $\bR_1, \bR_2, \bR_3$ be defined, as above, by 
	\begin{equation*} \begin{aligned}
	\bR_1 &= \UP \UP^{\top} \UA - \UP \bW^* \\
	\bR_2 &= \bW^* \SA^{1/2} - \SP^{1/2} \bW^*\\
	\bR_3 &= \UA - \UP \UP^{\top} \UA + \bR_1 = \UA - \UP \bW^*.
	\end{aligned} \end{equation*}
	where, as before, $\bW^*$ is the orthogonal transformation $\bW^*=\bW_1 \bW_2^{\top}$ with $\bW_1 \Sigma \bW_2$ being the singular value decomposition of $\UP^{\top} \UA$.  
	Then the following convergences in probability hold:
	\begin{equation} \label{eq:tozero1}
	\sqrt{n}\left[ (\bA-\bP)\UP(\bW^* \SA^{-1/2} - \SP^{-1/2} \bW^*) \right]_h
	\inprob \zeromx,
	\end{equation}
	\begin{equation} \label{eq:tozero2}
	\sqrt{n} \left[ \UP \UP^{\top} (\bA-\bP) \UP \bW^* \SA^{-1/2} \right]_h
	\inprob \zeromx,
	\end{equation}
	\begin{equation} \label{eq:tozero3}
	\sqrt{n} \left[ (\bI - \UP \UP^{\top})(\bA-\bP) \bR_3 \SA^{-1/2} \right]_h
	\inprob \zeromx,
	\end{equation}
	and with high probability,
	$$\|\bR_1\SA^{1/2}+\UP \bR_2\|_F \le \frac{ C\log n}{n^{1/2}}. $$
\end{lemma}
\begin{proof}
	We begin by observing that
	$$ \| \bR_1 \SA^{1/2} + \UP \bR_2 \|_F
	\le \| \bR_1 \|_F \| \SA^{1/2} \| + \| \bR_2 \|_F. $$
	Proposition \ref{prop:uptua_close_rotation} and the trivial upper bound on the eigenvalues
	of $\bA$ ensures that
	$$ \| \bR_1 \|_F \| \SA^{1/2} \|
	\le \frac{ C\log n }{ n^{1/2} }
	\text{ w.h.p. }, $$
	Combining this with Eq. \eqref{eq:approx_commute2_sqrt} in Lemma \ref{lem:order_bounds_on_minh_differences}, we conclude that
	$$ \| \bR_1 \SA^{1/2} + \UP \bR_2 \|_F
	\le \frac{ C \log n }{ n^{1/2} } \text{ w.h.p. } $$
	
	We will establish \eqref{eq:tozero1}, \eqref{eq:tozero2}
	and \eqref{eq:tozero3} order.
	To see \eqref{eq:tozero1}, observe that
	\begin{equation*}
	\sqrt{n} \| (\bA-\bP)\UP(\bW^* \SA^{-1/2} - \SP^{-1/2} \bW^{*}) \|_F
	\le \sqrt{n} \| (\bA-\bP) \UP \| \| \bW^* \SA^{-1/2} - \SP^{-1/2} \bW^* \|_F,
	\end{equation*}
	and Lemma~\ref{lem:order_bounds_on_minh_differences} imply that with high probability
	$$ \sqrt{n} \| (\bA-\bP)\UP(\bW^* \SA^{-1/2} - \SP^{-1/2} \bW^*) \|_F
	\le \frac{ C\log n }{ \sqrt{n} }, $$
	which goes to $0$ as $n \rightarrow \infty$.
	
	To show the convergence in \eqref{eq:tozero2}, we recall that 
	$\UP \SP^{1/2} \bW=\bX$ for some orthogonal matrix $\bW$ and observe that
	since the rows of the latent position matrix $\bX$ are necessarily
	bounded in Euclidean norm by $1$,
	and since the top $d$ eigenvalues of $\bP$ are of order $n$ (recall Lemma \ref{lem:eigs_lower_bound}),
	it follows that
	\begin{equation}\label{eq:UP2toinfty}
	\|\UP\|_{\tti} \le C n^{-1/2} \text{ w.h.p. }
	\end{equation}
	Next, \eqref{eq:hoeffding_up_upt} and Lemma \ref{lem:eigs_lower_bound} imply that
	\begin{equation*} \begin{aligned}
	\| (\UP \UP^{\top} (\bA-\bP) \UP \bW^{*} \SA^{-1/2})_h \|
	&\le \|\UP\|_{\tti}\| \UP^{\top} (\bA-\bP) \UP \| \| \SA^{-1/2} \| \\
	&\le \frac{ C \log n }{ n } \text{ w.h.p.},
	\end{aligned} \end{equation*}
	which implies~\eqref{eq:tozero2}.
	Finally, to establish~\eqref{eq:tozero3},
	we must bound the Euclidean norm of the vector
	\begin{equation} \label{eq:toughguy}
	\left[ (\bI - \UP \UP^{\top})(\bA-\bP) \bR_3 \SA^{-1/2} \right]_h,
	\end{equation}
	where, as defined above, $\bR_3 = \UA - \UP \bW^{*}$.
	Let $\bB_1$ and $\bB_2$ be defined as follows:
	\begin{equation}\label{eq:def_B1_B2}
	\begin{aligned}
	\bB_1&=
	(\bI - \UP \UP^{\top})(\bA-\bP)(\bI-\UP\UP^{\top}) \UA \SA^{-1/2} \\
	\bB_2&=(\bI - \UP \UP^{\top})(\bA-\bP)\UP(\UP^{\top} \UA - \bW^{*}) \SA^{-1/2}
	\end{aligned}
	\end{equation}
	Recalling that $\bR_3=\UA-\UP \bW^*$, we have
	\begin{equation*} \begin{aligned}
	(\bI - \UP \UP^{\top})(\bA-\bP) \bR_3 \SA^{-1/2}
	&= (\bI - \UP \UP^{\top})(\bA-\bP)(\UA-\UP\UP^{\top} \UA) \SA^{-1/2} \\
	&~~~~~~+ (\bI - \UP \UP^{\top})(\bA-\bP)(\UP\UP^{\top} \UA -\UP \bW^*) \SA^{-1/2} \\
	&= \bB_1 + \bB_2.
	\end{aligned} \end{equation*}
	We will bound the Euclidean norm of the
	$h$-th row of each of these two matrices on the right-hand side,
	from which a triangle inequality will yield our desired bound on the quantity in Equation~\eqref{eq:toughguy}.
	Recall that we use $C$ to denote a positive constant,
	independent of $n$ and $m$, which may change from line to line.
	
	Let us first consider
	$\bB_2 = (\bI - \UP \UP^{\top})(\bA-\bP)\UP(\UP^{\top} \UA - \bW^{*}) \SA^{-1/2}$.
	We have
	\begin{equation*}
	\| \bB_2 \|_F
	\le \| (\bI - \UP \UP^{\top})(\bA-\bP)\UP \|
	\| \UP^{\top} \UA - \bW^{*} \|_F \| \SA^{-1/2} \|.
	\end{equation*}
	By submultiplicativity of the spectral norm and Theorem\ref{thm:oliveira},
	$\| (\bI - \UP \UP^{\top})(\bA-\bP)\UP \| \le C n^{1/2} \log^{1/2} n$
	with high probability (and indeed, under our degree assumptions and Theorem \ref{thm:lu_peng}, the $\log n$ factor can be dropped).
	From Prop. \ref{prop:uptua_close_rotation} and Lemma \ref{lem:eigs_lower_bound},
	respectively, we have with high probability
	\begin{equation*}
	\| \UP^{\top} \UA - \bW^{*} \|_F \le C n^{-1} \log n
	\enspace \text{  and  } \enspace
	\| \SA^{-1/2} \| \le Cn^{-1/2} 
	\end{equation*}
	Thus, we deduce that with high probability,
	\begin{equation} \label{eq:B2bound}
	\| \bB_2 \|_F \le \frac{ C \log^{3/2} n }{ n }
	\end{equation}
	from which it follows that $\| \sqrt{n} \bB_2 \|_F \inprob 0$,
	and hence $\| \sqrt{n} (\bB_2)_h \| \inprob 0$.
	
	Turning our attention to $\bB_1$, and recalling that $\UA^{\top} \UA= \mathbf{I}$, we note that
	\begin{equation*} \begin{aligned}
	\| (\bB_1)_h \| &=
	\left\| \left[ (\bI - \UP \UP^{\top})(\bA-\bP)(\bI-\UP\UP^{\top})
	\UA \SA^{-1/2} \right]_h \right\| \\
	&= \left\| \left[ (\mathbf{I} - \UP \UP^{\top})(\bA-\bP)(\bI-\UP\UP^{\top})
	\UA \UA^{\top} \UA \SA^{-1/2} \right]_h \right\| \\
	&\le  \left\| \UA \SA^{-1/2} \right\|
	\left\| \left[ (\bI - \UP \UP^{\top})(\bA-\bP)(\bI-\UP\UP^{\top})\UA \UA^{\top}
	\right]_h \right\|.
	\end{aligned} \end{equation*}
	Let $\epsilon > 0$ be a constant. We will show that
	\begin{equation} \label{eq:inprobwewant}
	\lim_{n \rightarrow \infty}
	\Pr\left[ \| \sqrt{n}(\bB_1)_h \| > \epsilon \right]
	= 0.
	\end{equation}
	For ease of notation, define
	$$ \bE_1 = (\bI - \UP \UP^{\top})(\bA-\bP)(\bI-\UP\UP^{\top})\UA \UA^{\top}. $$
	We will show that
	\begin{equation} \label{eq:E1inprob}
	\lim_{n \rightarrow \infty}
	\Pr\left[ \sqrt{n} \left\| \left[ \bE_1 \right]_h \right\|
	> n^{1/4} \right]= 0,
	\end{equation}
	which will imply \eqref{eq:inprobwewant}
	since, by Lemma \ref{lem:eigs_lower_bound},  
	$\| \UA \SA^{-1/2} \| \le Cn^{-1/2}$ with high probability.
	Let $\bQ \in \R^{n \times n}$ be any permutation matrix.
	We observe that
	$$ \bQ \UP \UP^{\top}\bQ^{\top} \bQ\bP \bQ^{\top}=\bQ\bP \bQ^{\top}, $$
	and thus $\bQ\UP \UP^{\top}\bQ^{\top}$ is a projection matrix for $\bQ\bP \bQ^{\top}$
	if and only if $\UP \UP^{\top}$ is a projection matrix for $\bP$.
	A similar argument applies to the matrix $\UA \UA^{\top}$.
	Combining this with the exchangeability structure of the matrix
	$\bA - \bP$, it follows that the Frobenius norms
	of the rows of $\bE_1$ are equidistributed.
	This row-exchangeability for $\bE_1$ implies that
	$ n \E \| (\bE_1)_h \|^2 = \| \bE_1 \|_F^2 $.
	Applying Markov's inequality,
	\begin{equation} \label{eq:markov} \begin{aligned}
	\Pr\left[ \left\|\sqrt{n}\left[ \bE_1 \right]_h \right\| > t \right]
	&\le \frac{ n \E \left\| \left[
		(\bI - \UP \UP^{\top})(\bA-\bP)(\bI-\UP\UP^{\top})\UA \UA^{\top}
		\right]_h \right\|^2 }{ t^2 } \\
	&= \frac{ \E \left\| (\bI - \UP \UP^{\top})(\bA-\bP)(\bI-\UP\UP^{\top})\UA \UA^{\top}
		\right\|_F^2 }{ t^2 }.
	\end{aligned} \end{equation}
	We will proceed by showing that with high probability,
	\begin{equation} \label{eq:intermed}
	\left\| (\bI - \UP \UP^{\top})(\bA-\bP)(\bI-\UP\UP^{\top})\UA \UA^{\top}
	\right\|_F \le C  \log n,
	\end{equation}
	whence choosing $t = n^{1/4}$ in \eqref{eq:markov} yields that
	$$ \lim_{n\rightarrow \infty}
	\Pr\left[ \left\|\sqrt{n}\left[ (\bI - \UP \UP^{\top})(\bA-\bP)
	(\bI-\UP\UP^{\top})\UA \UA^{\top} \right]_h \right\| > n^{1/4} \right]
	= 0, $$
	and \eqref{eq:inprobwewant} will follow.
	We have
	\begin{equation*}
	\left\| (\bI - \UP \UP^{\top})(\bA-\bP)(\bI-\UP\UP^{\top})\UA \UA^{\top}
	\right\|_F
	\le \| \bA-\bP \| \| \UA - \UP\UP^{\top} \UA \|_F \| \UA \|
	\end{equation*}
	Theorem \ref{thm:lu_peng} and Lemma \ref{lem:eigs_lower_bound} implies that the first term in this product
	is at most $C n^{1/2}$  high probability,
	and the final term in this product is, trivially, at most $1$.
	To bound the second term, we will follow reasoning similar to that
	in Lemma \ref{lem:order_bounds_on_minh_differences}, combined with the Davis-Kahan theorem.
	The Davis-Kahan Theorem \cite{DavKah1970,Bhatia1997}
	implies that for a suitable constant $C > 0$,
	\begin{equation*}
	\| \UA \UA^{\top} - \UP\UP^{\top} \|
	\le \frac{ C \| \bA - \bP \| }{ \lambda_d(\bP) }.
	\end{equation*}
	By Theorem 2 in \cite{DK_usefulvariant},
	there exists orthonormal $\bW \in \R^{d \times d}$ such that
	$$ \| \UA - \UP \bW \|_F \le C \| \UA \UA^{\top} - \UP \UP^{\top} \|_F. $$
	We observe further that the multivariate linear least squares problem
	\begin{equation*}
	\min_{\bT \in \R^{d \times d}} \| \UA - \UP \bT \|_F^2
	\end{equation*}
	is solved by $\bT = \UP^{\top} \UA$.
	Combining all of the above, we find that 
	\begin{equation*} \begin{aligned}
	\| \UA - \UP\UP^{\top} \UA \|_F^2
	&\le \| \UA - \UP \bW \|_F^2
	\le C \| \UA \UA^{\top} - \UP \UP^{\top} \|_F^2 \\
	&\le C \| \UA \UA^{\top} - \UP \UP^{\top} \|^2 
	\le \left(\frac{ C \| \bA - \bP \| }{ \lambda_d(\bP) }\right)^2
	\le \frac{ C}{ n } \enspace \text{ w.h.p. }
	\end{aligned} \end{equation*}
	We conclude that
	\begin{equation*} \begin{aligned}
	\left\| (\bI - \UP \UP^{\top})(\bA-\bP)(\bI-\UP\UP^{\top})\UA \UA^{\top}
	\right\|_F
	&\le C\| \bA-\bP \| \| \UA - \UP\UP^{\top} \UA \|_F \| \UA \|_S \\
	&\le C  \log n \text{ w.h.p. },
	\end{aligned} \end{equation*}
	which implies \eqref{eq:intermed}, as required,
	and thus the convergence in \eqref{eq:tozero3} is established,
	completing the proof.
\end{proof}
Next, recall the inherent nonidentifiability in the RDPG model: suppose the ``true" latent positions are some matrix $\bX$. Then
$\bX = \bU_{\mathbf{P}} \bS_{\mathbf{P}}^{1/2} \bW$ for some suitably-chosen orthogonal matrix $\bW$.  We now consider the asymptotic distribution of 
$$ n^{1/2} \Wn^{\top} \left[ (\bA - \bP) \UP \SP^{-1/2} \right]_h$$
conditional on $\bX_i={\bf x}_i$.  Because we can write suitable terms of $\bA-\bP$ as the sum of independent random variables, we can invoke the Lindeberg-Feller Central Limit Theorem to establish the asymptotic normality of 
$$ n^{1/2} \Wn^{\top} \left[ (\bA - \bP) \UP \SP^{-1/2} \right]_h$$
as follows.
\begin{lemma} \label{lem:inlaw}
	Fix some $i \in [n]$.
	Conditional on $\bX_i = \bx_i \in \R^d$,
	there exists a sequence of $d$-by-$d$ orthogonal matrices $\{ \Wn \}$ such that
	$$ n^{1/2} \Wn^{\top} \left[ (\bA - \bP) \UP \SP^{-1/2} \right]_h
	\inlaw \calN( 0, \bSigma(\bx_i) ), $$
	where $\bSigma(\bx_i) \in \R^{d \times d}$ is a covariance matrix that
	depends on $\bx_i$.
\end{lemma}
\begin{proof}
	For each $n$, choose orthogonal $\Wn \in \R^{d \times d}$
	so that $\bX = \mathbf{U}_{\mathbf{P}} \mathbf{S}_{\mathbf{P}}^{1/2} \Wn$.
	At least one such $\Wn$ exists for each value of $n$, since, as discussed
	previously, the true latent positions $\bX$ are specified
	only up to some orthogonal transformation. 
	We then have
	\begin{equation*} \begin{aligned}
	n^{1/2} \Wn^{\top} \left[ (\bA - \bP) \UP \SP^{-1/2} \right]_i 
	&= n^{1/2} \Wn^{\top} \SP^{-1} \Wn \left[ \bA \bX - \bP \bX \right]_i \\
	&= n^{1/2} \Wn^{\top} \SP^{-1} \Wn \left(
	\sum_{j=1}^n
	\left(\bA_{ij} - \bP_{ij} \right)
	\bX_j \right) \\
	&= n^{1/2} \Wn^{\top} \SP^{-1} \Wn \left(
	\sum_{j \neq i} \left(\bA_{ij} - \bP_{ij}\right) \bX_j
	\right) - n^{1/2} \Wn^{\top} \SP^{-1} \Wn \bP_{ii} \bX_i \\
	&= \left( n \Wn^{\top} \SP^{-1} \Wn \right)
	\left[ n^{-1/2}
	\sum_{j \neq i} \left(\bA_{ij} - \bP_{ij}\right) \bX_j \right] - n \Wn^{\top} \SP^{-1} \Wn \frac{ \bP_{ii} \bX_i }{ n^{1/2} }.
	\end{aligned} \end{equation*}
	
	Conditioning on $\bX_i = \bx_i \in \R^d$, we first observe that
	$
	\tfrac{ \bP_{ii} }{ n^{1/2} } \bX_i
	= \frac{ \bx_i^{\top} \bx_i }{ n^{1/2} } x_i \rightarrow 0$ almost surely.
	Furthermore, 
	\begin{equation*} \begin{aligned}
	n^{-1/2} \sum_{j \neq i}  \left(\bA_{ij} - \bP_{ij}\right) \bX_j 
	= n^{-1/2} \sum_{j \neq i} \left(
	\bA_{ij} - \bX_j^{\top} \bx_i\right) \bX_j
	\end{aligned} \end{equation*}
	is a scaled sum of $n-1$ independent $0$-mean random variables,
	each with covariance matrix given by
	$$  \Sigmatilde(\bx_i) =
	\E \left[  \left( \bx_i^{\top} \bX_j - (\bx_i^{\top} \bX_j)^2 \right) \bX_j \bX_j^{\top} \right].
	$$
	The multivariate central limit theorem thus implies that
	\begin{equation} \label{eq:inlaw1}
	n^{-1/2} \sum_{j \neq i} \left(
	\bA_{ij} - \bX_j \bx_i^{\top})\bX_j\right)
	\inlaw \calN( \zeromx, \Sigmatilde(\bx_i) ) .
	\end{equation}
	Finally, by the strong law of large numbers,
	$$ n^{-1} \mathbf{W}_n^{\top} \mathbf{S}_{\mathbf{P}_n} \mathbf{W}_n = \frac{1}{n} \bX^{\top} \bX \rightarrow \bDelta \text{ a.s. } $$
	and thus $(n \mathbf{W}_n^{\top} \mathbf{S}^{-1}_{\mathbf{P}_n} \mathbf{W}_n) \rightarrow \bDelta^{-1}$ almost surely. 
	Combining this fact with~\eqref{eq:inlaw1},
	the multivariate version of Slutsky's theorem yields
	$$ n^{1/2} \Wn^{\top} \left[ (\bA - \bP) \UP \SP^{-1/2} \right]_h
	\inlaw \calN\left(\zeromx, \bSigma(\bx_i) \right)
	$$
	where $\bSigma(\bx_i) = \bDelta^{-1} \Sigmatilde(\bx_i) \bDelta^{-1}$.
	Integrating over the possible values of $\bx_i$ with respect to distribution
	$F$ completes the proof.
\end{proof}
Lemmas \ref{lem:stringent_control_residuals} and \ref{lem:inlaw} are the main ingredients in the proof of Theorem \ref{thm:clt_orig_but_better}, whose proof now follows easily:
\begin{proof}
	[Proof of Theorem \ref{thm:clt_orig_but_better}].
	We start with the following decomposition that was originally used in the proof of Theorem~\ref{thm:minh_frob}.
	\begin{equation} \label{eq:expansion} \begin{aligned}
	\sqrt{n}\left(\UA \SA^{1/2} - \UP \SP^{1/2} \bW^{*}\right)
	&= \sqrt{n}(\bA-\bP)\UP \SP^{-1/2} \bW^{*}
	+ \sqrt{n}(\bA-\bP)\UP(\bW^{*} \SA^{-1/2} - \SP^{-1/2} \bW^{*}) \\
	&~~~~~~-\sqrt{n} \UP \UP^{\top} (\bA-\bP) \UP \bW^{*}\SA^{-1/2} \\
	&~~~~~~+ \sqrt{n}(\bI - \UP \UP^{\top})(\bA-\bP) \bR_3 \SA^{-1/2}
	+ \sqrt{n}\bR_1 \SA^{1/2} + \sqrt{n}\UP \bR_2.
	\end{aligned} \end{equation}
	Now given any index $i$, Lemma~\ref{lem:stringent_control_residuals} can be used to bound the $\ell_2$ norm of the $i$-th row $(\bA-\bP)\UP(\bW^{*} \SA^{-1/2} - \SP^{-1/2} \bW^{*})$, $\UP \UP^{\top} (\bA-\bP) \UP \bW^{*}\SA^{-1/2}$ and $(\bI - \UP \UP^{\top})(\bA-\bP) \bR_3 \SA^{-1/2}$. The $\ell_2$ norm of the $i$-th row of $\bR_1 \SA^{1/2}$ and $\UP \bR_2$ can be bounded from above by bound $\|\bR_1 \SA^{1/2}\|_{\tti}$ and $\|\UP \bR_2\|_{\tti}$. Now since $\|\mathbf{U}_{\mathbf{P}}\|_{\tti} = O(n^{-1/2})$ by Eq.~\eqref{eq:UP2toinfty}, we conclude that $\sqrt{n} \|\bR_1 \SA^{1/2}\|_{\tti} = O(n^{-1/2} \log^{1/2}{n})$ and $\sqrt{n} \|\UP \bR_2\|_{\tti} = O(n^{-1/2} \log^{1/2}{n})$. Therefore, for any fixed index $i$, we have
	$$ \sqrt{n}\left(\UA \SA^{1/2} - \UP \SP^{1/2} \bW^{*}\right)_i = \sqrt{n}((\bA-\bP)\UP \SP^{-1/2} \bW^{*})_{i} + O(n^{-1/2} \log^{1/2} n)$$
	with high probability. Since $\mathbf{X} = \mathbf{U}_{\mathbf{P}} \mathbf{S}_{\mathbf{P}} \bW$, we can rewrite the above expression as
	$$ \sqrt{n}\left(\UA \SA^{1/2} (\bW^{*})^{\top} \mathbf{W} - \UP \SP^{1/2} \mathbf{W} \right)_i = \sqrt{n}((\bA-\bP)\UP \SP^{-1/2} \mathbf{W})_{i} + O(n^{-1/2} \log^{1/2} n).$$ Lemma~\ref{lem:inlaw} then establishes the asymptotic normality of $\sqrt{n}((\bA-\bP)\UP \SP^{-1/2} \mathbf{W})_{i}$ as desired.
\end{proof}

We now turn out attention a brief sketch of the proof of the central limit theorem for the Laplacian spectral embedding.

\subsection*{Sketch of proof of Theorem~\ref{THM:LSE}}
We present in this subsection a sketch of the main ideas in the proof of Theorem~\ref{THM:LSE};  detailed proofs are given in \cite{tang_lse}. 
We first introduce some additional notation. For $(\mathbf{X}_n, \mathbf{A}_n) \sim \mathrm{RDPG}(F)$, let $\mathbf{T}_n = \mathrm{diag}(\mathbf{P}_n \bm{1})$ be the $n \times n$ diagonal matrices whose diagonal entries are the {\em expected} vertex degrees. Then defining $\tilde{\bf X}_n = \mathbf{T}_n^{-1/2} \mathbf{X}_n$, and noting that $\tilde{\mathbf{X}}_n \tilde{\mathbf{X}}_n^{\top} = \mathcal{L}(\mathbf{P}_n) = \mathbf{T}_n^{-1/2} \mathbf{P}_n \mathbf{T}_n^{-1/2}$, Theorem~\ref{THM:LSE} depends on showing that there exists an orthogonal matrix $\mathbf{W}_n$ such that 
\begin{equation}
\label{eq:LSE-main1}
\breve{\bX}_n {\bf W}_n - \tilde{\bf X}_n = \mathbf{T}_n^{-1/2} (\mathbf{A}_n - \mathbf{P}_n) \mathbf{T}_n^{-1/2}
\tilde{\mathbf{X}}_n (\tilde{\mathbf{X}}_n^{\top} \tilde{\mathbf{X}}_n)^{-1} + \tfrac{1}{2}(\mathbf{I} - \mathbf{D}_n \mathbf{T}_n^{-1}) \tilde{\mathbf{X}}_n + \mathbf{R}_n 
\end{equation}  
where $\|\mathbf{R}_n\|_{F} = O(n^{-1})$ with high probability. The motivation behind Eq.~\eqref{eq:LSE-main1} is as follows. 
Given $\tilde{\mathbf{X}}_n$, the entries of the right hand side of Eq.~\eqref{eq:LSE-main1}, except for the term $\mathbf{R}_n$, can be expressed explicitly in terms of linear combinations of the entries $a_{ij} - p_{ij}$ of $\mathbf{A}_n - \mathbf{P}_n$. This is in contrast with the left hand side of Eq.~\eqref{eq:LSE-main1}, which depends on the quantities $\mathbf{U}_{\mathcal{L}(\mathbf{A}_n)}$ and $\mathbf{S}_{\mathcal{L}(\mathbf{A}_n)}$ (recall Definition~\ref{def:LSE}). Since the quantities $\mathbf{U}_{\mathcal{L}(\mathbf{A}_n)}$ and $\mathbf{S}_{\mathcal{L}(\mathbf{A}_n)}$ can not be expressed explicitly in terms of the entries of $\mathbf{A}_n$ and $\mathbf{P}_n$, we conclude that the right hand side of Eq.~\eqref{eq:LSE-main1} is simpler to analyze. 

Once Eq.~\eqref{eq:LSE-main1} is established, we can derive Theorem~\ref{THM:LSE} as follows. 
Let $\xi_i$ denote the $i$-th row of $n ( \breve{\mathbf{X}}_n \bW_n - \tilde{\mathbf{X}}_n)$ and let $r_i$ denote the $i$-th row of $\mathbf{R}_n$. Eq.~\eqref{eq:LSE-main1} then implies 
\begin{equation*}
\begin{split} 
\xi_i &= (\tilde{\mathbf{X}}_n^{\top} \tilde{\mathbf{X}}_n)^{-1} \frac{n}{\sqrt{t_i}} \Bigl( \sum_{j} \frac{a_{ij} -
	p_{ij}}{\sqrt{t_j}} (\tilde{\bX}_n)_j \Bigr) + \frac{n (t_i - d_i)}{2t_i} (\tilde{\bX}_n)_i + n r_i \\
&= (\tilde{\mathbf{X}}_n^{\top} \tilde{\mathbf{X}}_n)^{-1}
\frac{\sqrt{n}}{\sqrt{t_i}} \Bigl( \sum_{j
} \frac{\sqrt{n \rho_n} (a_{ij} - p_{ij}) (\bX_n)_j}{ t_j} \Bigr) - \frac{n (\bX_n)_i}{2 t_i^{3/2}} \sum_{j 
} (a_{ij} - p_{ij}) + n r_i \\
&= 
\frac{\sqrt{n}}{\sqrt{t_i}} \sum_{j
} \frac{(a_{ij} - p_{ij})}{\sqrt{n}} \Bigl(\frac{(\tilde{\mathbf{X}}_n^{\top} \tilde{\mathbf{X}}_n)^{-1} (\bX_n)_j}{ t_j/n} - \frac{(\bX_n)_{i}}{2 t_i/n} \Bigr) + n r_i
\end{split}
\end{equation*}
where $a_{ij}$ and $p_{ij}$ are the $ij$-th entries of $\mathbf{A}$ and $\mathbf{P}$, respectively, and $t_i$ is the $i$-th diagonal entry of $\mathbf{T}_n$. 
We can then show that $n r_i \overset{\mathrm{d}}{\longrightarrow} 0$. Indeed, there are $n$ rows in $\mathbf{R}_n$ and $\|\mathbf{R}_n\|_{F} = O(n^{-1})$; hence, on average, for each index $i$, $\|r_i\|^{2} = O(n^{-3})$ with high probability (a more precise argument similar to that used in proving Lemma~\ref{lem:stringent_control_residuals} is needed to establish this rigorously). 
Furthermore, $$t_i/n = \sum_{j} (\bX_n)_i^{\top} (\bX_n)_j/n \overset{\mathrm{a.s.}}{\longrightarrow} (\bX_n)_{i}^{\top} \bm{\mu}$$ as $n \rightarrow \infty$. Finally, $$\tilde{\mathbf{X}}_n^{\top} \tilde{\mathbf{X}}_n = \sum_{i} \bigl((\bX_n)_{i} (\bX_n)_{i}^{\top}/(\sum_{j} (\bX_n)_i^{\top} (\bX_n)_j)\bigr),$$ and this can be shown to converge to $\tilde{\Delta} = \mathbb{E}\bigl[\tfrac{\bX_1 \bX_1^{\top}}{\bX_1^{\top} \bm{\mu}}\bigr]$ as $n \rightarrow \infty$. We therefore have, after additional algebraic manipulations, that
\begin{equation*}
\begin{split}
\xi_i &=   \frac{\sqrt{n}}{\sqrt{t_i}} \sum_{j
} \frac{(a_{ij} - p_{ij})}{\sqrt{n} } \Bigl(\frac{\tilde{\Delta}^{-1} (\bX_n)_j}{ (\bX_n)_j^{\top} \bm{\mu}} - \frac{(\bX_n)_i}{2 (\bX_n)_{i}^{\top} \bm{\mu}} \Bigr) + o(1) \\
&= \frac{\sqrt{n}}{\sqrt{t_i}} \sum_{j
} \frac{(a_{ij} - (\bX_n)_i^{\top} (\bX_n)_j)}{\sqrt{n} } \Bigl(\frac{\tilde{\Delta}^{-1} (\bX_n)_j}{ (\bX_n)_j^{\top} \bm{\mu}} - \frac{(\bX_n)_{i}}{2 (\bX_n)_i^{\top} \bm{\mu}} \Bigr) + o(1)
\end{split}
\end{equation*}
with high probability.
Conditioning on $(\bX_n)_{i} = \bm{x}$, the above expression for $\xi_i$ is roughly a sum of independent and identically distributed mean $0$ random variables. The multivariate central limit theorem can then be applied to the above expression for $\xi_i$, thereby yielding Theorem~\ref{THM:LSE}. 

We now sketch the derivation of Eq.~\eqref{eq:LSE-main1}. For simplicity of notation, we shall ignore the subscript $n$ in the matrices $\mathbf{A}_n$, $\mathbf{X}_n$, $\mathbf{P}_n$ and related matrices. First, consider the following expression.
\begin{equation}
\label{eq:sketch1}
\begin{split}
\mathbf{U}_{\mathcal{L}(\mathbf{A})} \mathbf{S}_{\mathcal{L}(\mathbf{A})}^{1/2} - \mathbf{U}_{\mathcal{L}(\mathbf{P})} &\mathbf{S}_{\mathcal{L}(\mathbf{P})}^{1/2} \mathbf{U}_{\mathcal{L}(\mathbf{P})}^{\top} \mathbf{U}_{\mathcal{L}(\mathbf{A})}  = \mathcal{L}(\mathbf{A}) \tilde{\mathbf{U}}_{\mathcal{L}(\mathbf{A})} \tilde{\mathbf{S}}_{\mathcal{L}(\mathbf{A})}^{-1/2} - \mathcal{L}(\mathbf{P}) \tilde{\mathbf{U}}_{\mathcal{L}(\mathbf{P})} \tilde{\mathbf{S}}_{\mathcal{L}(\mathbf{P})}^{-1/2} \tilde{\mathbf{U}}_{\mathcal{L}(\mathbf{P})}^{\top} \tilde{\mathbf{U}}_{\mathcal{L}(\mathbf{A})} \\
&= \mathcal{L}(\mathbf{A}) \ULA \ULA^{\top} \ULA \SLA^{-1/2} - \mathcal{L}(\mathbf{P}) \ULP \SLP^{-1/2} \ULP^{\top} \ULA
\end{split}
\end{equation}
Now $\mathcal{L}(\mathbf{A})$ is concentrated around  $\mathcal{L}(\mathbf{P})$: namely, in the current setting, 
$$\|\mathcal{L}(\mathbf{A}) - \mathcal{L}(\mathbf{P})\| = O(n^{-1/2})$$ with high probability (see Theorem~2 in \cite{lu13:_spect}). Since $\|\mathcal{L}(\mathbf{P})\| = \Theta(1)$ and the non-zero eigenvalues of $\mathcal{L}(\mathbf{P})$ are all of order $\Theta(1)$, this again implies, by the Davis-Kahan theorem, that the eigenspace spanned by the $d$ largest eigenvalues of $\mathcal{L}(\mathbf{A})$ is ``close'' to that spanned by the $d$ largest eigenvalues of $\mathcal{L}(\mathbf{P})$. More precisely, $\|\ULA \ULA - \ULP \ULP \| = O(n^{-1/2})$ with high probability, and 
\begin{equation*}
\begin{split}
\ULA \SLA^{1/2} - \ULP \SLP^{1/2} \ULP^{\top} \ULA &= \mathcal{L}(\mathbf{A}) \ULP \ULP^{\top} \ULA \SLA^{-1/2}  \\ & - \mathcal{L}(\mathbf{P}) \ULP \SLP^{-1/2} \ULP^{\top} \ULA + \mathbf{R}_n
\end{split}
\end{equation*}
where $\|\mathbf{R}_n\| = O(n^{-1})$ with high probability. In addition, $\|\ULA \ULA - \ULP \ULP \| = O(n^{-1/2})$ also implies that 
there exists an orthogonal matrix $\mathbf{W}^{*}$ such that $\|\ULP^{\top} \ULA - \mathbf{W}^{*} \| = O(n^{-1})$ with high probability.

We next consider the terms $\SLP^{-1/2} \ULP^{\top} \ULA$ and $\ULP^{\top} \ULA \SLA^{-1/2}$. Note that the both are $d \times d$ matrices; 
furthermore, since $\SLA$ and $\SLP$ are diagonal matrices, the $ij$-th entry of $\SLP^{-1/2} \ULP^{\top} \ULA - \ULP^{\top} \ULA \SLA^{-1/2}$ can be written as the $\zeta_{ij} \times h_{ij}$ where $\zeta_{ij}$ is the $ij$-th entry of $\SLP \ULP^{\top} \ULA - \ULP^{\top} \ULA \SLA$ and the $h_{ij}$ are functions of $\lambda_{i}(\mathcal{L}(\bA))$ and $\lambda_{j}(\mathcal{L}(\bP)$. In particular, $|h_{ij}| < C$ for some positive constant $C$ for all $n$. We then have that
\begin{equation*}
\begin{split}\SLP \ULP^{\top} \ULA - \ULP^{\top} \ULA \SLA &= \ULP^{\top} (\mathcal{L}(\bP) - \mathcal{L}(\bA)) \ULA \\ &= \ULP^{\top} (\mathcal{L}(\bP) - \mathcal{L}(\bA)) \ULP \ULP^{\top} \ULA \\ &+ \ULP^{\top}((\mathcal{L}(\bP) - \mathcal{L}(\bA)) (\mathbf{I} - \ULP \ULP^{\top}) \ULA
\end{split}
\end{equation*}
Now, conditioning on $\bP$, the $ij$-th entry of $\ULP^{\top} (\mathcal{L}(\bP) - \mathcal{L}(\bA)) \ULP$ can be written as a linear combination of the entries of $\mathbf{A} - \mathbf{P}$ (which are independent) and the rows of $\mathbf{X}$; hence, it can be bounded using Hoeffding's inequality.
Meanwhile, the term
$$\ULP^{\top}((\mathcal{L}(\bP) - \mathcal{L}(\bA)) (\mathbf{I} - \ULP \ULP^{\top}) \ULA$$ can be bounded by the Davis-Kahan Theorem and the spectral norm difference of $\|\mathcal{L}(\bA) - \mathcal{L}(\bP)\|$. We therefore arrive at the important fact that 
$$\|\SLP \ULP^{\top} \ULA - \ULP^{\top} \ULA \SLA\|_{F} = O(n^{-1})$$
with high probability, and hence 
$$\|\SLP^{-1/2} \ULP^{\top} \ULA - \ULP^{\top} \ULA \SLA^{-1/2}\| = O(n^{-1})$$ with high probability. 

We can juxtapose $\ULP^{\top} \ULA$ and $\SLA^{-1/2}$ in the expression for Eq.~\eqref{eq:sketch1} and then replace $\ULP^{\top} \ULA$ by the orthogonal matrix $\mathbf{W}^{*}$, thereby obtaining
$$ \ULA \SLA^{1/2} - \ULP \SLP^{1/2} \mathbf{W}^{*} = (\mathcal{L}(\mathbf{A}) - \mathcal{L}(\mathbf{P})) \ULP \SLP^{-1/2} \mathbf{W}^{*} + \tilde{\mathbf{R}}_n $$
where $\|\tilde{\mathbf{R}}_n\| = O(n^{-1})$ with high probability. 
Since $$\tilde{\mathbf{X}} \tilde{\mathbf{X}}^{\top} = \mathcal{L}(\mathbf{P}) = \ULP \SLP \ULP^{\top},$$ 
we have $\tilde{\mathbf{X}} = \ULP \SLP^{1/2} \tilde{\bW}$ for some orthogonal matrix $\tilde{\mathbf{W}}$.
Therefore, 
\begin{equation*}
\begin{split}
\ULA \SLA^{1/2} - \tilde{\mathbf{X}} \tilde{\mathbf{W}}^{\top} \mathbf{W}^{*} & = 
(\mathcal{L}(\mathbf{A}) - \mathcal{L}(\mathbf{P})) \ULP \SLP^{-1/2} \mathbf{W}^{*} + \tilde{\bR}_n\\
& = (\mathcal{L}(\mathbf{A}) - \mathcal{L}(\mathbf{P})) \ULP \SLP^{1/2} \tilde{\mathbf{W}} \tilde{\mathbf{W}}^{\top} \SLP^{-1} \tilde{\mathbf{W}} \tilde{\mathbf{W}}^{\top} \mathbf{W}^{*} + \tilde{\bR}_n\\
& = (\mathcal{L}(\mathbf{A}) - \mathcal{L}(\mathbf{P})) \tilde{\mathbf{X}} (\tilde{\mathbf{X}}^{\top} \tilde{\mathbf{X}})^{-1} \tilde{\mathbf{W}}^{\top} \mathbf{W}^{*} + \tilde{\bR}_n.
\end{split}
\end{equation*}
Equivalently, 
\begin{equation} 
\label{eq:sketch-1}
\ULA \SLA^{1/2} (\mathbf{W}^{*})^{\top} \tilde{\mathbf{W}} - \tilde{\mathbf{X}} = (\mathcal{L}(\mathbf{A}) - \mathcal{L}(\mathbf{P})) \tilde{\mathbf{X}} (\tilde{\mathbf{X}}^{\top} \tilde{\mathbf{X}})^{-1} + \tilde{\bR}_n (\mathbf{W}^{*})^{\top} \tilde{\mathbf{W}}.
\end{equation}
The right hand side of Eq.~\eqref{eq:sketch-1} --- except for the residual term $\tilde{\bR}_n (\mathbf{W}^{*})^{\top} \tilde{\mathbf{W}}$ which has norm of order $O(n^{-1})$ with high probability --- can now be written explicitly in terms of the entries of $\mathbf{A}$ and $\mathbf{P}$. However, since $$\mathcal{L}(\mathbf{A}) = \mathbf{D}^{-1/2} \mathbf{A} \mathbf{D}^{-1/2},$$ the entries of the right hand side of Eq.~\eqref{eq:sketch-1} are neither linear nor affine combinations of the entries of $\mathbf{A} - \mathbf{P}$. Nevertheless, a Taylor-series expansion of the entries of $\mathbf{D}^{-1/2}$ allows us to conclude that
$$\|\mathbf{D}^{-1/2} - \mathbf{T}^{-1/2} - \tfrac{1}{2} \mathbf{T}^{-3/2} (\mathbf{T} - \mathbf{D})\| = O(n^{-3/2})$$ with high probability. 
Substituting this into Eq.~\eqref{eq:sketch-1} yields, after a few further algebraic simplifications, Eq.~\eqref{eq:LSE-main1}.

	\bibliographystyle{plain}
	\bibliography{biblio_summary}
\end{document}